\documentclass[letterpaper, 10 pt, conference]{ieeeconf}  % Comment this line out if you need a4paper

\IEEEoverridecommandlockouts                              % This command is only needed if 
\usepackage{amsmath} % assumes amsmath package installed
\usepackage{amssymb}  % assumes amsmath package installed
\usepackage{amsfonts}
\usepackage{mathtools}
\usepackage{bbm}
\usepackage{bm}
\usepackage{color}
\usepackage{graphics} % for pdf, bitmapped graphics files
\usepackage{graphicx}
\usepackage{epsfig}   % for postscript graphics files
\usepackage{epstopdf}
\usepackage[noadjust]{cite} % basic fos bibliography [1]-[5]
\usepackage{tikz}
\usetikzlibrary{calc,positioning,shapes,shadows,arrows,fit}
\usepackage{fancyhdr}
\usepackage{framed}
\usepackage{comment}
\usepackage{caption}

\usepackage{amsthm}

\newtheorem{theorem}{Theorem}
\newtheorem{lemma}{Lemma}

\theoremstyle{definition}
\newtheorem{proposition}{Proposition}
\newtheorem{definition}{Definition}
\newtheorem{remark}{Remark}
\newtheorem{assumption}{Assumption}
\newtheorem{example}{Example}
\newtheorem{problem}{Problem}
\newtheorem{property}{Property}

\usepackage{algorithm}
\usepackage{algpseudocode}
\usepackage{pifont}

\title{\LARGE \bf
Decentralized Abstractions and Timed Constrained Planning of a General Class of Coupled Multi-Agent Systems}

\author{Alexandros Nikou, Shahab Heshmati-alamdari, Christos Verginis and Dimos V. Dimarogonas% <-this % stops a space
	\thanks{Alexandros Nikou, Christos Verginis and Dimos V. Dimarogonas are with the ACCESS Linnaeus Center, School of Electrical Engineering, KTH Royal Institute of Technology, SE-100 44, Stockholm,
		Sweden and with the KTH Center for Autonomous Systems. Email: {\tt\small \{anikou, cverginis, dimos\}@kth.se}. Shahab  Heshmati-alamdari is with the Control Systems Lab, Department of Mechanical Engineering, National Technical University of Athens, 9 Heroon Polytechniou Street, Zografou 15780, Athens, Greece. Email: {\tt\small \{shahab\}@mail.ntua.gr}. This work was supported by the H2020 ERC Starting Grant BUCOPHSYS, the Swedish Research Council (VR), the Swedish Foundation for Strategic Research, the Knut och Alice Wallenberg Foundation, the European Union's Horizon 2020 Research and Innovation Programme under the Grant Agreement No. 644128 (AEROWORKS) and the EU H2020 Research and Innovation Programme under GA No. 731869 (Co4Robots).}
}

\begin{document}

\maketitle
\thispagestyle{empty}
\pagestyle{empty}

%%%%%%%%%%%%%%%%%%%%%%%%%%%%%%%%%%%%%%%%%%%%%%%%%%%%%%%%%%%%%%%%%%%%%%%%%%%%%%%%
\begin{abstract}
This paper presents a fully automated procedure for controller synthesis for a general class of multi-agent systems under coupling constraints. Each agent is modeled with dynamics consisting of two terms: the first one models the coupling constraints and the other one is an additional bounded control input. We aim to design these inputs so that each agent meets an individual high-level specification given as a Metric Interval Temporal Logic (MITL). Furthermore, the connectivity of the initially connected agents, is required to be maintained. First, assuming a polyhedral partition of the workspace, a novel decentralized abstraction that provides controllers for each agent that guarantee the transition between different regions is designed. The controllers are the solution of a Robust Optimal Control Problem (ROCP) for each agent. Second, by utilizing techniques from formal verification, an algorithm that computes the individual runs which provably satisfy the high-level tasks is provided. Finally, simulation results conducted in MATLAB environment verify the performance of the proposed framework. \\
\end{abstract}

\begin{keywords}
	multi-agent systems, cooperative control, hybrid systems.
\end{keywords}

\section{Introduction}

Cooperative control of multi-agent systems has traditionally focused on designing distributed control laws in order to achieve global tasks such as consensus and formation control, and at the same time fulfill properties such as network connectivity and collision avoidance. Over the last few years, the field of control of multi-agent systems under temporal logic specifications has been gaining attention. In this work, we aim to additionally introduce specific time bounds into these tasks, in order to include specifications such as: ``Robot 1 and robot 2 should visit region $A$ and $B$ within 4 time units respectively or ``Both robots 1 and 2 should periodically survey regions $A_1$, $A_2$, $A_3$, avoid region $X$ and always keep the longest time between two consecutive visits to $A_1$ below 8 time units".

The qualitative specification language that has primarily been used to express the high-level tasks is Linear Temporal Logic (LTL) (see, e.g., \cite{loizou_2004, muray_2010_receding}). There is a rich body of literature containing algorithms for verification and synthesis of multi-agent systems under temporal logic specifications (\cite{guo_2015_reconfiguration, frazzoli_vehicle_routing, zavlanos_2016_multi-agent_LTL}). A three-step hierarchical procedure to address the problem of multi-agent systems under LTL specifications is described as follows (\cite{fainekos_girard_2009_temporal, belta_2010_product_system, belta_cdc_reduced_communication}): first the dynamics of each agent is abstracted into a Transition System (TS). Second, by invoking ideas from formal verification, a discrete plan that meets the high-level tasks is synthesized for each agent. Third, the discrete plan is translated into a sequence of continuous time controllers for the original continuous dynamical system of each agent.

Controller synthesis under timed specifications has been considered in \cite{liu_MTL, murray_2015_stl, topcu_2015, fainekos_mtl_2015_robot, baras_MTL_2016_new}. However, all these works are restricted to single agent planning and are not extendable to multi-agent systems in a straightforward way. The multi-agent case has been considered in \cite{frazzoli_MTL}, where the vehicle routing problem was addressed, under Metric Temporal Logic (MTL) specifications. The corresponding approach does not rely on automata-based verification, as it is based on a construction of linear inequalities and the solution of a Mixed-Integer Linear Programming (MILP) problem. 

An automata-based solution was proposed in our previous work \cite{alex_2016_acc}, where MITL formulas were introduced in order to synthesize controllers such that every agent fulfills an individual specification and the team of agents fulfills a global specification. Specifically, the abstraction of each agent's dynamics was considered to be given and an upper bound of the time that each agent needs to perform a transition from one region to another was assumed. Furthermore, potential coupled constraints between the agents were not taken into consideration. Motivated by this, in this work, we aim to address the aforementioned issues. We assume that the dynamics of each agent consists of two parts: the first part is a nonlinear function representing the coupling between the agent and its neighbors, and the second one is an additional control input which will be exploited for high-level planning. Hereafter, we call it a free input. A decentralized abstraction procedure is provided, which leads to an individual Weighted Transition System (WTS) for each agent and provides a basis for high-level planning.

Abstractions for both single and multi-agent systems have been provided e.g. in \cite{alur_2000_discrete_abstractions, belta_2004_abstraction, helwa2014block, habets_2006_reachability, zamani_2012_symbolic, liu_2016_abstraction, abate_2014_finite_abstractions, pola_2016_symbolic, boskos_cdc_2015}. In this paper, we deal with the complete framework of both abstractions and controller synthesis of multi-agent systems. We start from the dynamics of each agent and we provide controllers that guarantee the transition between the regions of the workspace, while the initially connected agents remain connected for all times. The decentralized controllers are the solution of an ROCP. Then, each agent is assigned an individual task given as an MITL formulas. We aim to synthesize controllers, in discrete level, so that each agent performs the desired individual task within specific time bounds as imposed by the MITL formulas. In particular, we provide an automatic controller synthesis method of a general class of coupled multi-agent systems under high-level tasks with timed constraints. Compared to existing works on multi-agent planning under temporal logic specifications, the proposed approach considers dynamically coupled multi-agent systems under timed temporal specifications in a distributed way. 

In our previous work \cite{alex_acc_8_pages}, we treated a similar problem, but the under consideration dynamics were linear couplings and connectivity maintenance was not guaranteed by the proposed control scheme. Furthermore, the procedure was partially decentralized, due to the fact that a product Wighted Transition System (WTS) was required, which rendered the framework computationally intractable. To the best of the authors' knowledge, this is the first time that a fully automated framework for a general class of multi-agent systems consisting of both constructing purely decentralized abstractions and conducting timed temporal logic planning is considered.

This paper is organized as follows. In Section \ref{sec: preliminaries} a description of the necessary mathematical tools, the notations and the definitions are given. Section \ref{sec: prob_formulation} provides the dynamics of the system and the formal problem statement. Section \ref{sec: solution} discusses the technical details of the solution. Section \ref{sec: simulation_results} is devoted to a simulation example. Finally, conclusions and future work are discussed in Section \ref{sec: conclusions}.

\section{Notation and Preliminaries} \label{sec: preliminaries}

\subsection{Notation}

We denote by $\mathbb{R}, \mathbb{Q}_+, \mathbb{N}$ the set of real, nonnegative rational and natural numbers including 0, respectively. $\mathbb{R}^{n}_{\ge 0}$ and $\mathbb{R}^{n}_{> 0}$ are the sets of real $n$-vectors with all elements nonnegative and positive, respectively. Define also $\mathbb{T}_{\infty} = \mathbb{T} \cup \{\infty\}$ for a set $\mathbb{T} \subseteq \mathbb{R}$. Given a set $S$, denote by $|S|$ its cardinality, by $S^N = S \times \dots \times S$ its $N$-fold Cartesian product, and by $2^S$ the set of all its subsets. Given the sets $S_1, S_2$, their \emph{Minkowski addition} is defined by $S_1 \oplus S_2 = \{s_1 + s_2 : s_1 \in S_1, s_2 \in S_2\}$. $I_n \in \mathbb{R}^{n \times n}$ stands for the identity matrix. The notation $\|x\|$ is used for the Euclidean norm of a vector $x \in \mathbb{R}^n$. $\|A\| = \text{max} \{\|Ax\| : \|x\| = 1\}$ stands for the induced norm of a matrix $A \in \mathbb{R}^{n \times n}$. $\mathcal{B}(c, \underline{r}) = \{x \in \mathbb{R}^2: \|x-c\|\le \underline{r} \}$ is the disk of center $c \in \mathbb{R}^2$ and $\underline{r} \in \mathbb{R}_{> 0}$. The absolute value of the maximum singular value and the absolute value of the minimum eigenvalue of a matrix $A \in \mathbb{R}^{n \times n}$ are denoted by $\sigma_{\max}(A), \lambda_{\min}(A)$, respectively. The indexes $i$ and $j$ stand for agent $i$ and its neighbors (see Sec. \ref{sec: prob_formulation} for the definition of neighbors), respectively; $\mu, z \in \mathbb{N}$ are indexes used for sequences and sampling times, respectively.

\begin{definition}
	Consider two sets $A, B \subseteq \mathbb{R}^n$. Then, the \emph{Pontryagin difference} is defined by:
	\begin{equation*}
	A \sim B = \{x \in \mathbb{R}^n: x+y \in A, \forall \ y \in B\}.
	\end{equation*}
\end{definition}

\begin{definition} \label{def:class_K}
	(\cite{khalil_nonlinear_systems}) A continuous function $\alpha: [0, a) \to \mathbb{R}_{\ge 0}$ is said to belong to \emph{class} $\mathcal{K}$, if it is strictly increasing and $\alpha(0) = 0$. It is said to belong to class $\mathcal{K}_{\infty}$ if $a = \infty$ and $\alpha(r) \to \infty$, as $r \to \infty$.
\end{definition}

\begin{definition} \label{def:class_KL}
	(\cite{khalil_nonlinear_systems}) A continuous function $\beta: [0, a) \times \mathbb{R}_{\ge 0} \to \mathbb{R}_{\ge 0}$ is said to belong to \emph{class} $\mathcal{KL}$, if:
	\begin{itemize}
		\item For each fixed $s$, $\beta(r, s) \in \mathcal{K}$ with respect to $r$.
		\item For each fixed $r$, $\beta(r, s)$ is decreasing with respect to $s$ and $\beta(r, s) \to 0$, at $s \to \infty$.
	\end{itemize}
\end{definition}

\begin{definition} \label{def:ISS}
	(\cite{sontag_2008_ISS}) A nonlinear system $\dot{x} = f(x,u)$ with initial condition $x(t_0)$ is said to be \emph{Input to State Stable (ISS)} if there exist functions $\beta \in \mathcal{KL}$ and $\sigma \in \mathcal{K}_{\infty}$ such that:
	\begin{equation*}
	\|x(t)\| \le \beta(\|x(t_0)\|, t)+\sigma(\|u\|).
	\end{equation*}
\end{definition}

\begin{definition} \label{def:ISS_lyapunov}
	(\cite{sontag_2008_ISS}) A Lyapunov function $V(x,u)$ for the nonlinear system $\dot{x} = f(x,u)$ with initial condition $x(t_0)$ is said to be \emph{ISS-Lyapunov function} if there exist functions $\alpha, \sigma \in \mathcal{K}_\infty$ such that:
	\begin{equation} \label{eq:lyapunov_iss}
	\dot{V}(x,u) \le - \alpha(\|x\|)+\sigma(\|u\|), \forall x, u.
	\end{equation}
\end{definition}

\begin{theorem}
	A nonlinear system $\dot{x} = f(x,u)$ with initial condition $x(t_0)$ is said to be ISS if and only if it admits a ISS-Lyapunov function.
\end{theorem} 
\begin{proof}
	The proof can be found in \cite{sontag_1995_ISS_proofs}.
\end{proof}

\subsection{Partitions}

In the subsequent analysis a discrete partition of the workspace will be considered which is formalized through the following definition. 

\begin{definition} \label{def:partition}
	Given a set $S$, we say that a family of sets $\{S_\ell\}_{\ell \in \mathbb{I}}$ forms a partition of $S$ if $S \neq \emptyset$, $\displaystyle \bigcup_{\ell \in \mathbb{I}} S_\ell = S$ and for every $S,S' \in S$ with $S \neq S'$ it holds $S \cap S' = \emptyset$. 
\end{definition}
Hereafter, every region $S_\ell$ of a partition $S$ will be called \textit{region}.

\subsection{Time Sequence, Timed Run and Weighted Transition System}

In this section we include some definitions that are required to analyze our framework.

An infinite sequence of elements of a set $X$ is called an \textit{infinite word} over this set and it is denoted by $\chi = \chi(0)\chi(1) \ldots $ The $z$-th element of a sequence is denoted by $\chi(z)$. For certain technical reasons that will be clarified in the sequel, we will assume hereafter that $\mathbb{T} = \mathbb{Q}_{+}$.

\begin{definition} (\cite{alur1994}) A \textit{time sequence} $\tau = \tau(0) \tau(1) \ldots$ is an infinite sequence of time values $\tau(\mu) \in \mathbb{T} = \mathbb{Q}_{+}$, satisfying the following properties: 1) Monotonicity: $\tau(\mu) < \tau(\mu+1)$ for all $j \geq 0$; 2) Progress: For every $t \in \mathbb{T}$, there exists $\mu \ge 1$, such that $\tau(\mu) > t$.
\end{definition}

An \textit{atomic proposition} $\sigma$ is a statement %over the problem variables and parameters 
that is either True $(\top)$ or False $(\bot)$. % at a given time instance.

\begin{definition} (\cite{alur1994})
	Let $\Sigma$ be a finite set of atomic propositions. A \textit{timed word} $w$ over the set $\Sigma$ is an infinite sequence $w^t = (w(0), \tau(0)) (w(1), \tau(1)) \ldots$ where $w(0) w(1) \ldots$ is an infinite word over the set $2^{\Sigma}$ and $\tau(0) \tau(1) \ldots$ is a time sequence with $\tau(\mu) \in \mathbb{T}, \mu \geq 0$.
\end{definition}

\begin{definition} \label{def: WTS}
A Weighted Transition System (\textit{WTS}) is a tuple $(S, S_0, Act, \longrightarrow, d, \Sigma, L)$ where $S$ is a finite set of states; $S_0 \subseteq S$ is a set of initial states; $Act$ is a set of actions; $\longrightarrow \subseteq S \times Act \times S$ is a transition relation; $d: \longrightarrow \rightarrow \mathbb{T}$ is a map that assigns a positive weight to each transition; $\Sigma$ is a finite set of atomic propositions; and $L: S \rightarrow 2^{\Sigma}$ is a labeling function.
\end{definition}

\begin{definition}\label{run_of_WTS}
A \textit{timed run} of a WTS is an infinite sequence $r^t = (r(0), \tau(0))(r(1), \tau(1)) \ldots$, such that $r(0) \in S_0$, and for all $\mu \geq 1$, it holds that $r(\mu) \in S$ and $(r(\mu), \alpha(\mu), r(\mu+1)) \in \longrightarrow$ for a sequence of actions $\alpha(1) \alpha(2) \ldots$ with $\alpha(\mu) \in Act, \forall \ \mu \geq 1$. The \textit{time stamps} $\tau(\mu), \mu \geq 0$ are inductively defined as: 1) $\tau(0) = 0$; 2) $\displaystyle \tau(\mu+1) =  \tau(\mu) + d(r(\mu), \alpha(\mu), r(\mu+1)), \forall \ \mu \geq 1$. Every timed run $r^t$ generates a \textit{timed word} $w(r^t) = (w(0), \tau(0)) \ (w(1), \tau(1)) \ldots$ over the set $2^{\Sigma}$ where $w(\mu) = L(r(\mu))$, $\forall \ \mu \geq 0$ is the subset of atomic propositions that are true at state $r(\mu)$. 
\end{definition}

\subsection{Metric Interval Temporal Logic (MITL)}

The syntax of \textit{Metric Interval Temporal Logic (MITL)} over a set of atomic propositions $\Sigma$ is defined by the grammar:
\begin{equation*}
\varphi := p \ | \ \neg \varphi \ | \ \varphi_1 \wedge \varphi_2 \ | \ \bigcirc_I \varphi  \ | \ \Diamond_I \varphi \mid \square_I \varphi \mid  \varphi_1 \ \mathcal{U}_I \ \varphi_2,
\end{equation*}
where $\sigma \in \Sigma$, and $\bigcirc$, $\Diamond$, $\square$ and $\mathcal U$ are the next, eventually, always and until temporal operators, respectively; $I = [a, b] \subseteq \mathbb{T}$ where $a, b \in [0, \infty]$ with $a < b$ is a non-empty timed interval. MITL can be interpreted either in continuous or point-wise semantics \cite{pavithra_expressiveness}. In this paper, the latter approach is utilized, since the consideration of point-wise (event-based) semantics is more suitable for the automata-based specifications considered in a discretized state-space. The MITL formulas are interpreted over timed words like the ones produced by a WTS it is given in Def. \ref{run_of_WTS}.

\begin{definition} \label{def:mitl_semantics} (\cite{pavithra_expressiveness}, \cite{quaknine_decidability})
	Given a timed word $w^t = (w(0),\tau(0))(w(1),\tau(1)) \dots$, an MITL formula $\varphi$ and a position $i$ in the timed word, the satisfaction relation $(w^t, i) \models \varphi$, for $\ i \geq 0$ (read $w^t$ satisfies $\varphi$ at position $\mu$) is inductively defined as follows:
	\begin{align*} \label{eq: for1}
	&(w^t, \mu) \models p \Leftrightarrow p \in w(\mu), \\
	&(w^t, \mu) \models \neg \varphi \Leftrightarrow (w^t, i) \not \models \varphi, \\
	&(w^t, \mu) \models \varphi_1 \wedge \varphi_2 \Leftrightarrow (w^t, \mu) \models \varphi_1 \ \text{and} \ (w^t, \mu) \models \varphi_2, \\
	&(w^t, \mu) \models \bigcirc_I \ \varphi \Leftrightarrow (w^t, \mu+1) \models \varphi \notag \\
	&\hspace{35mm} \text{and} \ \tau(\mu+1) - \tau(i) \in I, \\
	&(w^t, \mu) \models \Diamond_I \varphi \Leftrightarrow \exists \mu' \ge \mu, \ \text{such that} \notag \\
	&\hspace{35mm} (w^t, j) \models \varphi, \tau(\mu')-\tau(\mu) \in {I},
	\end{align*}
	\begin{align*}
	&(w^t, \mu) \models \square_I \varphi \Leftrightarrow \forall \mu' \ge \mu, \notag \\
	&\hspace{25mm} \tau(\mu')-\tau(\mu) \in {I} \Rightarrow (w^t, \mu') \models \varphi,  \\
	&(w^t, \mu) \models \varphi_1 \ \mathcal{U}_I \ \varphi_2 \Leftrightarrow \exists \mu' \ge \mu, \ \text{s.t. } (w^t, \mu') \models \varphi_2, \\ 
	&\qquad \tau(\mu')-\tau(\mu) \in I \ \text{and} \ (w^t, \mu'') \models \varphi_1, \forall \ \mu \leq \mu'' < \mu'.
	\end{align*}
	We say that a timed run $r^t = (r(0),\tau(0))(r(1),\tau(1)) \dots$ satisfies the MITL formula $\varphi$ (we write $r^t \models \varphi$) if and only if the corresponding timed word $w(r^t) = (w(0),\tau(0))(w(1),\tau(1)) \dots$ with $w(\mu) = L(r(\mu)), \forall \mu \ge 0$, satisfies the MITL formula ($w(r^t) \models \varphi$).
\end{definition}

It has been proved that MITL is decidable in infinite words and point-wise semantics, which is the case considered here (see \cite{alur_mitl, reynold} for details). The model checking and satisfiability problems are \textit{EXPSPACE}-complete. It should be noted that in the context of timed systems, EXSPACE complexity is fairly low \cite{bouyer_phd}. An example with a WTS and two runs $r_1^t, r_2^t$ that satisfy two MITL formulas can be found in \cite[Section II, page 4]{alex_acc_2017}.

\subsection{Timed B\"uchi Automata} \label{sec: timed_automata}
\textit{Timed B\"uchi Automata (TBA)} were introduced in \cite{alur1994} and in this work, we also partially adopt the notation from \cite{bouyer_phd, tripakis_tba}. %{is this right? did they have a buchi condition already?}
Let $C = \{c_1, \ldots, c_{|C|}\}$ be a finite set of \textit{clocks}. The set of \textit{clock constraints} $\Phi(C)$ is defined by the grammar
\begin{equation*}
	\phi :=  \top \mid \ \neg \phi \ | \ \phi_1 \wedge \phi_2 \ | \ c \bowtie \psi, 
\end{equation*}
where $c \in C$ is a clock, $\psi \in \mathbb{T}$ is a clock constant and $\bowtie \ \in  \{ <, >, \geq, \leq, = \}$. A clock \textit{valuation} is a function $\nu: C \rightarrow\mathbb{T}$ that assigns a value to each clock. A clock $c_i$ has valuation $\nu_i$ for $i \in \{1, \ldots, |C|\}$, and $\nu = (\nu_1, \ldots, \nu_{|C|})$. We denote by $\nu \models \phi$ the fact that the valuation $\nu$ satisfies the clock constraint $\phi$.

\begin{definition}
A \textit{Timed B\"uchi Automaton} is a tuple $\mathcal{A} = (Q, Q^{\text{init}}, C, Inv, E, F, \Sigma, \mathcal{L})$ where $Q$ is a finite set of locations; $Q^{\text{init}} \subseteq Q$ is the set of initial locations; $C$ is a finite set of clocks; $Inv: Q \rightarrow \Phi(C)$ is the invariant; $E \subseteq Q \times \Phi(C) \times 2^C \times Q$ gives the set of edges; $F \subseteq Q$ is a set of accepting locations; $\Sigma$ is a finite set of atomic propositions; and $\mathcal{L}: Q \rightarrow 2^{\Sigma}$ labels every state with a subset of atomic propositions.
\end{definition}

For the semantics of TBA we refer the reader to \cite[Section II, page 4]{alex_acc_2017}. The problem of deciding the emptiness of the language of a given TBA $\mathcal{A}$ is PSPACE-complete \cite{alur1994}. Any MITL formula $\varphi$ over $\Sigma$ can be algorithmically translated to a TBA with the alphabet $2^{\Sigma}$, such that the language of timed words that satisfy $\varphi$ is the language of timed words produced by the TBA (\cite{alur_mitl, maler_MITL_TA, nickovic_timed_aut}). An example of a TBA and accepting runs of it can be found in \cite[Section II, page 4]{alex_acc_2017}.

\begin{remark}
Traditionally, the clock constraints and the TBAs are defined with $\mathbb{T} = \mathbb{N}$. However, they can be extended to accommodate $\mathbb{T} = \mathbb{Q}_+$, by multiplying all the rational numbers that are appearing in the state invariants and the edge constraints with their least common multiple.
\end{remark}

\section{Problem Formulation} \label{sec: prob_formulation}

\subsection{System Model}
Consider a system of $N$ agents, with $\mathcal{V} = \{1,\dots, N\}, N \ge 2$, operating in a workspace $W \subseteq \mathbb{R}^2$. The workspace is assumed to be closed, bounded and connected. Let $x_i: \mathbb{R}_{\ge 0} \to D$ denotes the position of each agent in the workspace at time $t \in \mathbb{R}_{\ge 0}$. Each agent is equipped with a sensor device that can sense omnidirectionally. Let the disk $\mathcal{B}(x_i(t), \underline{r})$ model the sensing zone of agent $i$ at time $t \in \mathbb{R}_{\ge 0}$, where $\underline{r} \in \mathbb{R}_{\ge 0}$ is the sensing radius. The sensing radius is the same for all the agents. Let also $h > 0$ denote the constant sampling period of the system. We make the following assumption:

\begin{assumption} \label{ass:measurement_assumption}
(Measurements Assumption) It is assumed that each agent $i$, is able to measure its own position and all agents' positions that are located within agent's $i$ sensing zone without any delays.
\end{assumption}
According to Assumption \ref{ass:measurement_assumption}, the agent's $i$ neighboring set at time $t_0$ is defined by $\mathcal{N}_i  = \{j \in \mathcal{V}: x_j(t_0) \in \mathcal{B}(x_i(t_0), \underline{r})\}$. For the neighboring set $\mathcal{N}_i$ define also $N_i = |\mathcal{N}_i|$. Note that $i \in \mathcal{N}_j \Leftrightarrow j \in \mathcal{N}_i, \forall \ i,j \in \mathcal{V}, i \neq j$. The control design for every agent $i$ should guarantee that it remains connected with all its neighbors $j \in \mathcal{N}_i$, for all times.

Consider the neighboring set $\mathcal{N}_i$. The coupled dynamics of each agent are given in the form:
\begin{equation} \label{eq:system}
	\dot{x}_i = f(x_i, \bar{x}_i)+u_{i}, \ x_i \in W, \  i \in \mathcal{V},
\end{equation}
where $f:W$ $\times$ $W^{N_i}$ $\to W$, is a nonlinear function representing the coupling between agent $i$ and its neighbors $i_1, \dots, i_{N_i}$. The notation $\bar{x}_i = [ x_{i_1}^\top, \ldots, x_{i_{N_i}}^\top]^\top \in W^{N_i}$ is used for the vector of the neighbors of agent $i$, and $u_i: \mathbb{R}_{\ge 0} \to \mathbb{R}^2, \ i \in \mathcal{V}$ is the control input of each agent. For the dynamics of each agent the following assumption are taken.

\begin{assumption} \label{ass:dynamic_control_bounds}
	There exist constants $u_{\max}, \bar{M}$ with $0 < u_{\max} < \bar{M} < \infty$ such that the following holds $\forall i \in \mathcal{V}, (x_i, \bar{x}_i) \in W \times W^{N_i}$:
	\begin{subequations}
	\begin{align}
	&\|f_i(x_i, \bar{x}_i) \| \leq \bar{M}, \  \\
	& u_i \in \mathcal{U}_i \triangleq \{u_i \in \mathbb{R}^2: \|u_i\| \le u_{\max}\}. \label{eq:control_constraints}
	\end{align}
	\end{subequations}
\end{assumption}

\begin{assumption} \label{ass:lipsitch_f_x_bar_x}
The functions $f_i(x_i, \bar{x}_i), i \in \mathcal{V}$ are \emph{Lipschitz continuous} in $W \times W^{N_i}$. Thus, there exists constants $L_i, \bar{L}_i > 0$ such that the following inequalities hold:
\begin{subequations}
\begin{align}
\|f_i(x_i, \bar{x}_i) - f_i(y_i, \bar{x}_i) \| &\le L_i \|x_i - y_i\|, \label{eq:lip_1} \\
\|f_i(x_i, \bar{x}_i) - f_i(x_i, \bar{y}_i) \| &\le \bar{L}_i \|\bar{x}_i - \bar{y}_i\|, \label{eq:lip_2}
\end{align} 
\end{subequations}
for all $x_i, y_i \in W, \bar{x}_i, \bar{y}_i \in W^{N_i}, i \in \mathcal{V}$.
\end{assumption}

\begin{remark}
The coupling terms $f_i(x_i, \bar{x}_i), i \in \mathcal{V}$ are encountered in a large set of multi-agent protocols \cite{mesbahi_2010_graph_theory}, including consensus, connectivity maintenance, collision avoidance and formation control. In addition, \eqref{eq:system} may represent internal dynamics of the system as for instance in the case of smart buildings (see e.g., \cite{andreasson_2014_smart_building}) where the temperature $T_i, i \in \mathcal{V}$ of each room evolves according to the law $\dot{T}_i = \sum_{j \in \mathcal{N}_i} \alpha_{ij}(T_j-T_i)+u_i$, with $\alpha_{ij}$ representing the heat conductivity between rooms $i$ and $j$ and $u_i$ the heating/cooling capabilities of the room. 
\end{remark}

\subsection{Specification}
Our goal is to control the multi-agent system \eqref{eq:system} so that each agent obeys a given individual specification. In particular, it is required to drive each agent to a sequence of desired subsets of the \textit{workspace} $W$ within certain time limits and provide certain atomic tasks there. Atomic tasks are captured through a finite set of atomic propositions $\Sigma_i, i \in \mathcal{V}$, with $\Sigma_i \cap \Sigma_j = \emptyset$, for all $i,j \in \mathcal{V}, i \neq j,$ which means that the agents do not share any atomic propositions. Each position $x_i$ of each agent $i \in \mathcal{V}$ is labeled with atomic propositions that hold there. Initially, a labeling function 
\begin{equation} \label{eq:label_lambda}
\Lambda_i: W \to 2^{\Sigma_i},
\end{equation}
is introduced for each agent $i \in \mathcal{V}$ which maps each state $x_i \in \mathbb{R}^2$ with the atomic propositions $\Lambda_i(x_i)$ which hold true at $x_i$ i.e., the subset of atomic propositions that hold for agent $i$ in position $x_i$. Define also by $\displaystyle \Lambda(x) = \bigcup_{i \in \mathcal V} \Lambda_i(x)$ the union of all the labeling functions. Let us now introduce the following assumption which is important for defining the problem properly.
\begin{assumption}  \label{assumption: AP_cell_decomposition}
There exists a partition $D = \{D_\ell\}_{\ell \in \mathbb I}$ of the workspace $W$ which respects the labeling function $\Lambda$ i.e., for all $D_\ell \in D$ it holds that $\Lambda(x) = \Lambda(x'), \forall \ x, x' \in D_\ell$. This assumption, intuitively, and without loss of generality, means that the same atomic propositions hold at all the points that belong to the same region of the partition.
\end{assumption}

Although the regions $D_\ell, \ell \in \mathbb{I}$ of the partition $D$ may have different geometric shape, without loss of generality, we assume that they are hexagons with side length $R$. Define also for each agent $i$ a labeling function:
\begin{equation} \label{eq:label_mathcal_lambda}
{L}_i: D \to 2^{\Sigma_i},
\end{equation}
which maps every region of the partition $D$ to the subset of the atomic propositions which hold true there. Furthermore, we assume that a time step $T > h > 0$ is given. This time step models the required time in which each agent should transit from a region to a neighboring region and is the same for all the agents.

The trajectory of each agent $i$ is denoted by $x_i(t), t \geq 0, i \in \mathcal{V}$. The trajectory $x_i(t)$ is associated with a unique sequence:
\begin{equation*}
r_{x_i}^t = (r_i(0), \tau_i(0))(r_i(1), \tau_i(1))(r_i(2), \tau_i(2))\ldots,
\end{equation*}
of regions that the agent $i$ crosses, where for all $\mu \ge 0$ it holds that: $x_i(\tau_i(\mu) \in r_i(\mu)$ and $\Lambda_i(x_i(t)) = {L}_i(r_i(\mu)), \forall \ t \in [\tau_i(\mu), \tau_i(\mu+1))$ for some $r_i(\mu) \in D$ and $r_i(\mu) \ne r_i(\mu+1)$. The timed word:
\begin{align*}
w_{x_i}^t &= (L_i(r_i(0)), \tau_i(0))({L}_i(r_i(1)), \tau_i(1)) \notag \\ 
&\hspace{38mm}(L_i(r_i(2)), \tau_i(2))\ldots,
\end{align*}
where $w_i(\mu) = {L}_i(r_i(\mu)), \mu \ge 0, i \in \mathcal{V}$, is associated uniquely with the trajectory $x_i(t)$.

%Instead of computing the time stamps $\tau_i(0), \tau_i(1), \dots$ in which the trajectory $x_i(t)$ crosses the facets of the regions $r_i(0), r_i(1), \dots$, respectively, in this paper, the following relaxation is adopted. 
 
\begin{definition} 
For each agent $i \in \mathcal{V}$ we define the \textit{relaxed timed word} as:
\begin{align} \label{eq:relaxed_timed_word}
\widetilde{w}_{i}^t &= (w_i(0), \widetilde{\tau}_i(0))(w_i(1), \widetilde{\tau}_i(1)) (w_i(2), \widetilde{\tau}_i(2))\ldots,
\end{align}
where $w_i(\mu) = L_i(r_i(\mu)), \tilde{\tau}_i(\mu) \in [\tau_i(\mu), \tau_i(\mu+1)), \forall \ \mu \ge 0$. 
\end{definition}

The time stamp $\tau_i(0) = \widetilde{\tau}_i(0) = t_0, i \in \mathcal{V}$ models the initial starting time of the agents. The time stamps $\tau_i(\mu), \mu \ge 1$ models the exact time in which the agent $i$ crosses the boundary of the regions $r_i(\mu-1)$ and $r_i(\mu)$. The time stamps $\widetilde{\tau}_i(\mu)$ model a time instant in which the agent $i$ is in the region $r_i(\mu)$ of the workspace (see Example \ref{ex: example_01} below). The specification task $\varphi_i$ given as an MITL formula over the set of atomic propositions $\Sigma_i$, represents desired tasks that are imposed to each agent $i \in \mathcal{I}$. We say that a trajectory $x_i(t)$ \emph{satisfies} a formula $\varphi_i$  given in MITL over the set $\Sigma_i$, and we formally write:
\begin{equation*} \label{eq:traj_satisf_formula}
	x_i(t) \models \varphi_i, \forall t \ge 0,
\end{equation*}
if and only if there exists a \emph{relaxed timed word} $\widetilde{w}_{i}^t$ that complies with $x_i(t)$ and satisfies $\varphi_i$ according to the semantics of MITL in \ref{def:mitl_semantics}.

\begin{figure}[t!]
	\centering
	\begin{tikzpicture}[scale = 0.8]
	
	% plot the grid		
	\draw[step=2.5, line width=.04cm] (-2.5, -5.0) grid (0,0);
	\draw[line width=.04cm] (-7.5,0.0) -- (-2.5,0.0);
	\draw[line width=.04cm] (-7.5,-2.5) -- (-2.5,-2.5);
	\draw[line width=.04cm] (-7.5,-5.0) -- (-2.5,-5.0);
	\draw[step=2.5, line width=.04cm] (-10.0, -5.0) grid (-7.5,0);
	
	% plot the colours
	\filldraw[fill=black!10, line width=.04cm] (-10, -2.5) rectangle +(2.5, 2.5);
	\filldraw[fill=orange!40, line width=.04cm] (-10, -5.0) rectangle +(2.5, 2.5);
	\filldraw[fill=red!20, line width=.04cm] (-7.5, -2.5) rectangle +(5.0, 2.5);
	\filldraw[fill=yellow!40, line width=.04cm] (-2.5, -2.5) rectangle +(2.5, 2.5);
	\filldraw[fill=blue!20, line width=.04cm] (-2.5, -5.0) rectangle +(2.5, 2.5);
	\filldraw[fill=green!40, line width=.04cm] (-7.5, -5.0) rectangle +(5.0, 2.5);
	
	% draw the trajectory for agent 1
	\draw [color=red,thick,->,>=stealth'](-8.8, -3.7) .. controls (-7.50, -0.0) .. (-5.8, -3.2);
	
	% draw the trajectory for agent 2
	\draw [color=red,thick,->,>=stealth'](-5.1, -3.7) .. controls (-0.2, -2.5) .. (-3.0, -1.1);
	
	% draw the points of agent 1
	\draw (-8.8, -3.7) node[circle, inner sep=0.8pt, fill=black, label={below:{$x_1(0)$}}] (A1) {};
	\draw (-8.35, -2.5) node[circle, inner sep=0.8pt, fill=black] (B1) {};
	\draw (-7.50, -0.85) node[circle, inner sep=0.8pt, fill=black, label={left:{$x_1(t_2)$}}] (C1) {};
	\draw (-6.2, -2.5) node[circle, inner sep=0.8pt, fill=black] (D1) {};
	
	% draw the points of agent 1
	\draw (-1.15, -2.5) node[circle, inner sep=0.8pt, fill=black] (A2) {};
	\draw (-2.5, -3.05) node[circle, inner sep=0.8pt, fill=black] (B2) {};
	\draw (-5.1, -3.7) node[circle, inner sep=0.8pt, fill=black, label={below:{$x_2(0)$}}] (C2) {};
	\draw (-2.5, -1.35) node[circle, inner sep=0.8pt, fill=black] (D2) {};
	
	% cell decomposition
	\node at (-8.7, 0.5) {$D_1$};
	\node at (-5.0, 0.5) {$D_2$};
	\node at (-1.3, 0.5) {$D_3$};
	\node at (-1.3, -5.5) {$D_4$};
	\node at (-5.0, -5.5) {$D_5$};
	\node at (-8.7, -5.5) {$D_6$};
	\node at (-9,-2.2) {$x_1(t_1)$};
	\node at (-5.5,-2.2) {$x_1(t_3)$};
	\node at (-1.8,-3.3) {$x_2(t_1')$};
	\node at (-0.57,-2.1) {$x_2(t_2')$};
	\node at (-1.9,-1.1) {$x_2(t_3')$};
	\end{tikzpicture}
	
	\caption{An example of two agents performing in a partitioned workspace.}
	\label{fig: example_01}
\end{figure}
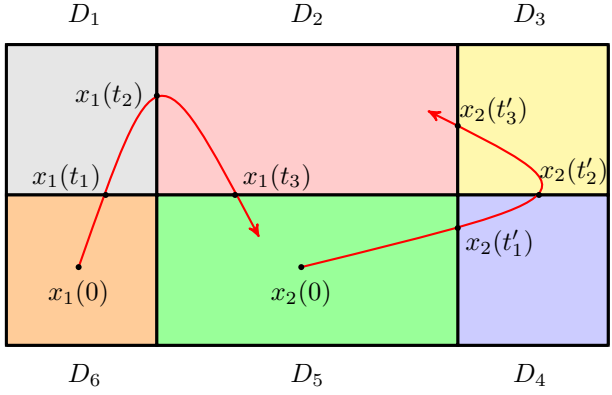

\begin{example} \label{ex: example_01}
	Consider $N = 2$ agents performing in the partitioned environment of Fig. \ref{fig: example_01}. Both agents have the ability to pick up, deliver and throw two different balls. Their sets of atomic propositions are $\Sigma_1 = \{\rm pickUp1, deliver1, throw1\}$ and $\Sigma_2 = \{\rm pickUp2, deliver2, throw2\}$, respectively, and satisfy $\Sigma_1 \cap \Sigma_2 = \emptyset$. Three points of the agents' trajectories that belong to different regions with different atomic propositions are captured. Assume that $t_1 < t_1' < t_2 < t_2 < t_2' < t_3 < t_3'.$ The trajectories $x_1(t), x_2(t), t \ge 0$ are depicted with the red lines. According to Assumption \ref{assumption: AP_cell_decomposition}, the partition $D = \{D_\ell\}_{\ell \in \mathbb I} = \{D_1, \ldots, D_6\}$ is given where $\mathbb I = \{1, \ldots, 6\}$ respects the labeling functions $\Lambda_i, {L}_i, i \in \{1,2\}$. In particular, it holds that: 
	\begin{align}
	&\Lambda_1(x_1(t)) = {L}_1(r_1(0)) = \{\rm pickUp1\}, t \in [0, t_1), \notag \\
	&\Lambda_1(x_1(t)) = {L}_1(r_1(1)) = \{\rm throw1\}, t \in [t_1, t_2), \notag \\ 
	&\Lambda_1(x_1(t)) = {L}_1(r_1(2)) = \{\rm deliver1\}, t \in [t_2, t_3), \notag \\
	&\Lambda_1(x_1(t)) = {L}_1(r_1(3)) = \emptyset, t \ge t_3. \notag \\
	&\Lambda_2(x_2(t)) = {L}_2(r_2(0)) = \{\rm pickUp2\}, t \in [0, t_1'), \notag \\
	&\Lambda_2(x_2(t)) = {L}_2(r_2(1))  = \{\rm deliver2\}, t \in [t_1', t_2'), \notag \\ 
	&\Lambda_2(x_2(t)) = {L}_2(r_2(2)) = \{\rm throw2\}, t \in [t_2', t_3'), \notag \\
	&\Lambda_2(x_2(t)) = {L}_2(r_2(3))  = \emptyset, t \ge t_3'. \notag
	\end{align}
	By the fact that $w_i(\mu) = {L}(r_i(\mu)), \forall \ i \in \{1,2\}, \mu \in \{1,2,3\}$, the corresponding individual timed words are given as:
	\begin{align}
	w^t_{x_1} &= (\{\rm pickUp1\}, 0)(\{\rm throw1\}, t_1)(\{\rm deliver1\}, t_2)(\emptyset, t_3), \notag \\
	w^t_{x_2} &= (\{\rm pickUp2\}, 0)(\{\rm deliver2\}, t_1')(\{\rm throw2\}, t_2')(\emptyset, t_3'). \notag
	\end{align}
	According to \eqref{eq:relaxed_timed_word}, two two relaxed timed words (depicted with red in Fig. \ref{fig: example_01}) are given as:
	\begin{align}
	&w^t_{1} = (\{\rm pickUp1\}, \widetilde{\tau}_1(0))(\{\rm throw1\}, \widetilde{\tau}_1(1)) \notag \\
	&\hspace{40mm} (\{\rm deliver1\}, \widetilde{\tau}_1(2)) (\emptyset, \widetilde{\tau}_1(3)), \notag \\
	&w^t_{2} = (\{\rm pickUp2\}, \widetilde{\tau}_2(0))(\{\rm deliver2\}, \widetilde{\tau}_2(1)) \notag \\ 
	&\hspace{40mm}(\{\rm throw2\}, \widetilde{\tau}_2(2))(\emptyset, \widetilde{\tau}_2(3)). \notag
	\end{align}
    The time stamps $\widetilde{\tau}_1(\mu), \widetilde{\tau}_2(\mu), \mu \in \{1, 2, 3\}$ should satisfy the following conditions:
	\begin{align*}
	\widetilde{\tau}_1(0) &\in [\tau_1(0), \tau_1(1)) = [0, t_1), \\
	\widetilde{\tau}_1(1) &\in [\tau_1(1), \tau_1(2)) = [t_1, t_2), \\
	\widetilde{\tau}_1(2) &\in [\tau_1(2), \tau_1(3)) = [t_2, t_3), \\
	\widetilde{\tau}_1(3) &\in [\tau_1(3), \cdot) = [t_3, \cdot), \\
	\widetilde{\tau}_2(0) &\in [\tau_2(0), \tau_2(1)) = [0, t_1), \\
	\widetilde{\tau}_2(1) &\in [\tau_2(1), \tau_2(2)) = [t_1, t_2), \\
	\widetilde{\tau}_2(2) &\in [\tau_2(2), \tau_2(3)) = [t_2, t_3), \\ 
	\widetilde{\tau}_2(3) &\in [\tau_2(3), \cdot) = [t_3, \cdot). \\
	\end{align*}
\end{example}
\subsection{Problem Statement}

We can now formulate the problem treated in this paper as follows:

\begin{problem} \label{problem: basic_prob}
Given $N$ agents operating in the bounded workspace $W \subseteq \mathbb{R}^2$, their initial positions $x_1(t_0), \dots, x_N(t_0)$, their dynamics as in \eqref{eq:system}, a time step $T > h > 0$, $N$ task specification formulas $\varphi_1, \ldots, \varphi_N$ expressed in MITL over the sets of services $\Sigma_1, \ldots, \Sigma_{{N}}$, respectively, a partition of the workspace $W$ into hexagonal regions $\{D_\ell\}_{\ell \in \mathbb I}$  with side length $R$ as in Assumption \ref{assumption: AP_cell_decomposition} and the labeling functions $\Lambda_1, \ldots, \Lambda_N, {L}_1, \ldots, {L}_N$, as in \eqref{eq:label_lambda}, \eqref{eq:label_mathcal_lambda}, assign control laws $u_1, \ldots, u_N$ to each agent $1, \dots, N$, respectively, such that the connectivity between the agents that belong to the neighboring sets $\mathcal{N}_1, \dots, \mathcal{N}_N$ is maintained, as well as each agent fulfills its individual MITL specification $\varphi_1, \dots, \varphi_N$, respectively, i.e., $x_1(t) \models \varphi_1, \dots, x_N(t) \models \varphi_N, \forall \ t \in \mathbb{R}_{\ge 0}$.
\end{problem}

\begin{remark}
The initial positions $x_1(t_0), \dots, x_N(t_0)$ should be such that the agents which are required to remain connected for all times need to satisfy the inequality $\|x_i(t_0)-x_{i'}(t_0)\| < 2\underline{r}, i, i' \in \mathcal{V}, i \neq i'$. %It will be guaranteed by the proposed controller scheme that the agents that satisfy the last inequality, will remain neighbors for all times t > t_0$.
\end{remark}

\begin{remark}
It should be noted that, in this work, the dependencies between the agents are induced through the coupled dynamics \eqref{eq:system}  and not in the discrete level, by allowing for couplings between the services (i.e., $\Sigma_i \cap \Sigma_j \ne \emptyset$, for some $i, j \in \mathcal{V}$). Hence, even though the agents do not share atomic propositions, the constraints on their motion due to the dynamic couplings and the connectivity maintenance specifications may restrict them to fulfill the desired high-level tasks. Treating additional couplings through individual atomic propositions in the discrete level is a topic of current work.
\end{remark}

\begin{remark}
In our previous work on the multi-agent controller synthesis framework under MITL specifications \cite{alex_2016_acc}, the multi-agent system was considered to have fully-actuated dynamics. The only constraints on the system were due to the presence of time constrained MITL formulas. In the current framework, we have two types of constraints: the constraints due to the coupling dynamics of the system \eqref{eq:system}, which constrain the motion of each agent, and, the timed constraints that are inherently imposed from the time bounds of the MITL formulas. Thus, there exist formulas that cannot be satisfied either due to the coupling constraints or the time constraints of the MITL formulas. These constraints, make the procedure of the controller synthesis in the discrete level substantially different and more elaborate than the corresponding multi-agent LTL frameworks in the literature (\cite{guo_2015_reconfiguration, frazzoli_vehicle_routing, belta_2010_product_system, belta_cdc_reduced_communication}).
\end{remark}

\section{Proposed Solution} \label{sec: solution}

In this section, a systematic solution to Problem~\ref{problem: basic_prob} is introduced. Our overall approach builds on abstracting the system in \eqref{eq:system} through a WTS for each agent and exploiting the fact that the timed runs in the $i$-th WTS project onto the trajectories of agent $i$ while preserving the satisfaction of the individual MITL formulas $\varphi_i, i \in \mathcal{V}$. In particular, the following analysis is performed:
\begin{enumerate}
	\item We propose a novel decentralized abstraction technique for the multi-agent system, i.e., discretization of the time into time steps $T$ for the given partition $D = \{D_\ell\}_{\ell \in \mathbb{I}}$, such that the motion of each agent is modeled by a WTS $\mathcal{T}_i, \ i \in \mathcal{I}$ (Section \ref{sec: abstration}). We adopt here the technique of designing Nonlinear Model Predictive Controllers (NMPC), for driving the agents between neighboring regions.
	\item A three-step automated procedure for controller synthesis which serves as a solution to Problem \ref{problem: basic_prob} is provided in Section \ref{sec: synthesis}.
	\item Finally, the computational complexity of the proposed approach is discussed in Section \ref{sec:complexity}.
\end{enumerate}
The next sections provide the proposed solution in detail.

\subsection{Discrete System Abstraction} \label{sec: abstration}

In this section we provide the abstraction technique that is designed in order to capture the dynamics of each agent into WTSs. Thereafter, we work completely at discrete level, which is necessary in order to solve Problem \ref{problem: basic_prob}.

\begin{figure}[t!]
	\centering
	\begin{tikzpicture}  
	% plot polygon 0	
	\node[regular polygon, regular polygon sides=6, minimum width=2cm,draw, fill = green!20] (reg1) at (0,0){};   
	\node at (0.0, 0.0) {$\bullet$};
	\node at (0.0, -0.32) {$x_i(t_k)$};
	\node at (0.0, +0.33) {$P(i, k)$};
	
	% plot polygon 2
	\node at (1.52, 0.87) {$\widetilde{P}(i,k,2)$};
	\node[regular polygon, regular polygon sides=6, minimum width=2cm,draw] (reg2) at (1.52,0.87){};
	
	% plot polygon 3
	\node at (1.52,-0.87) {$\widetilde{P}(i,k,3)$};
	\node[regular polygon, regular polygon sides=6, minimum width=2cm,draw] (reg3) at (1.52,-0.87){};
	
	% plot polygon 4
	\node at (0,-1.74) {$\widetilde{P}(i,k,4)$};
	\node[regular polygon, regular polygon sides=6, minimum width=2cm,draw] (reg4) at (0,-1.74){};
	
	% plot polygon 5
	\node at (-1.52,-0.87) {$\widetilde{P}(i,k,5)$};
	\node[regular polygon, regular polygon sides=6, minimum width=2cm,draw] (reg5) at (-1.52,-0.87){};
	
	% plot polygon 6
	\node at (-1.52,0.87) {$\widetilde{P}(i,k,6)$};
	\node[regular polygon, regular polygon sides=6, minimum width=2cm,draw] (reg5) at (-1.52,0.87){};
	
	% plot polygon 1
	\node at (0,1.74) {$\widetilde{P}(i,k,1)$};
	\node[regular polygon, regular polygon sides=6, minimum width=2cm,draw] (reg5) at (0,1.74){};
	\end{tikzpicture}
	\caption{Illustration of agent $i$ occupying region $P(i, k)$, depicted by green, at time $t_k = t_0+k T$ with $\bar{P}(i, k) = \bigcup_{\widetilde{\ell} \in \mathbb{L}} \widetilde{P}(i,k, \widetilde{\ell})$ being the set of regions that the agent can transit at exactly time $T$.}
	\label{fig:agent_i_widetilde_P}
\end{figure}
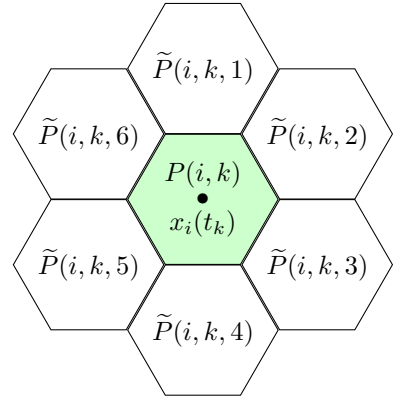

\subsubsection{Workspace Geometry} 

Consider an enumeration $\mathbb{I}$ of the regions of the workspace, the index variable $\ell \in \mathbb{I}$ and the given time step $T$. The time step $T$ models the time duration that each agent needs to transit between two neighboring regions of the workspace. Consider also a timed sequence:
\begin{equation} \label{eq:time_sequence}
\mathcal{S} = \{t_0, t_1 = t_0+T, \dots, t_k = t_0+k T, \dots \}, k \in \mathbb{N}.
\end{equation} 
$S$ models the time stamps in which the agents are required to occupy different neighboring regions. For example, if at time $t_k$ agent $i$ occupies region $D_{\ell}$, at the next time stamp $t_k +T$ is required to occupy a neighboring region of $D_{\ell}$. The agents are always forced to change region for every different time stamp. Let us define the mapping:
\begin{equation*}
P:\mathcal{V} \times \mathbb{N}\to D,
\end{equation*}
which denotes the fact that the agent $i \in \mathcal{V}$, at time instant $$t_k = t_0+k T, k \in \mathbb{N},$$ occupies the region $D_{\scriptscriptstyle \ell_i} \in D$ for an index $\ell_i \in \mathbb{I}$. Define the mapping:
\begin{equation*}
\widetilde{P}: \mathcal{V} \times \mathbb{N} \times \mathbb{L} \to D.
\end{equation*}
where $\mathbb{L} = \{1,\dots,6\}$. By $\widetilde{P}(i, k, \widetilde{\ell}), \widetilde{\ell} \in \mathbb{L}$ we denote one and only one out of the six neighboring regions of region $P(i, k)$ that agent $i$ occupies at time $t_k$. Define also by $\bar{P}(i,k)$ the union of all the six neighboring regions of region $P(i, k)$, i.e.,
\begin{equation*}
\bar{P}(i,k) = \bigcup_{\widetilde{\ell} \in \mathbb{L}} \widetilde{P}(i,k, \widetilde{\ell}),
\end{equation*} 
with $|\bar{P}(i,k)| = 6$. An example of agent $i$ being at the region $P(i, k)$ along with its neighboring regions is depicted in Fig. \ref{fig:agent_i_widetilde_P}.

We start by giving a graphical example for the abstraction technique that will be adopted in this work. Consider agent $i$ occupying the green region $P(i, k) = D_{\scriptscriptstyle \ell_i}$ at time $t_k = t_0+kT$ and let its neighbors $j_1, j_2$ occupying the red and blue regions $P(j_1, k) = D_{\scriptscriptstyle \ell_{j_1}}, P(j_2, k) = D_{\scriptscriptstyle \ell_{j_2}}$, respectively, as is depicted in Fig. \ref{fig:agent_i_j1_j2_example}. The neighboring regions $\bar{P}(i,k), \bar{P}(j_1,k)$ and $\widetilde{P}(j_2,k, \widetilde{\ell}), \widetilde{\ell} \in \{4,5,6\}$ for agent $i, j_1, j_2$, respectively, are also depicted. All the agents start their motion at time $t_k$ simultaneously. The goal is to design a decentralized feedback control law $u_i(x_i, x_{j_1}, x_{j_2})$, that drives agent $i$ in the neighboring region $D_{\scriptscriptstyle \ell_{\text{des}}}$ exactly at time $T$, \emph{regardless of the transitions of its neighbors to their neighboring regions}. If such controller exists, it is stored in the memory a new search for the next region is performed. This procedure is repeated for all possible neighboring regions i.e., six times, and for all the agents. For the example of Fig. \ref{fig:agent_i_j1_j2_example}, the procedure is performed $6^3$ times (six times for each agent). With this procedure, we are able to: 1) synchronize the agents so that each of them knows at every time step $T$ its position in the workspace as well as the region that occupies; 2) know which controller brings each agent in its desired region for any possible choice of controllers of its corresponding neighbors. We will hereafter present a formal approach of this procedure.
We will hereafter present a formal approach of this procedure.

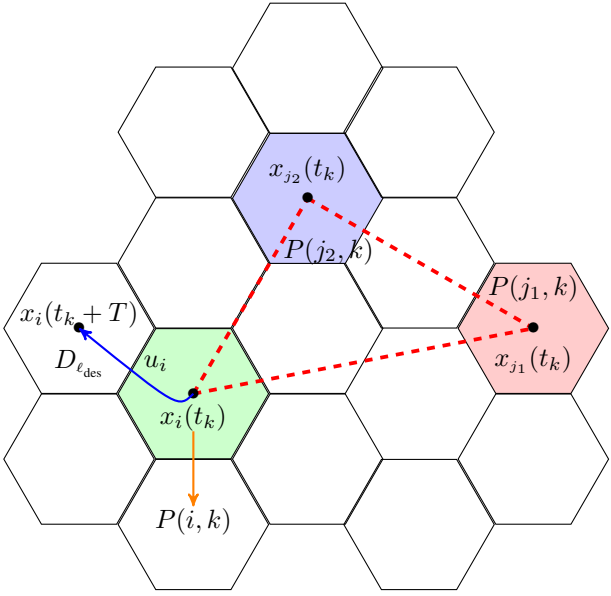
\begin{figure}[t!]
	\centering
	\begin{tikzpicture}  
	
	% plot polygon 0 of agent i
	\node[regular polygon, regular polygon sides=6, minimum width=2cm,draw, fill = green!20] (reg01) at (0,0){};   
	\node at (0.0, -0.32) {$x_i(t_k)$};
	%\node at (0.0, +0.33) {$P(i, k)$};
	
	% plot polygon 1 of agent i
	%\node at (0,1.74) {$\widetilde{P}(i,k,1)$};
	\node[regular polygon, regular polygon sides=6, minimum width=2cm,draw] (reg11) at (0,1.74){};
	
	% plot polygon 2 of agent i
	%\node at (1.52, 0.87) {$\widetilde{P}(i,k,2)$};
	\node[regular polygon, regular polygon sides=6, minimum width=2cm,draw] (reg21) at (1.52,0.87){};
	
	% plot polygon 3 of agent i
	%\node at (1.52,-0.87) {$\widetilde{P}(i,k,3)$};
	\node[regular polygon, regular polygon sides=6, minimum width=2cm,draw] (reg31) at (1.52,-0.87){};
	
	% plot polygon 4 of agent i
	%\node at (0,-1.74) {$\widetilde{P}(i,k,4)$};
	\node[regular polygon, regular polygon sides=6, minimum width=2cm,draw] (reg41) at (0,-1.74){};
	
	% plot polygon 5 of agent i
	%\node at (-1.52,-0.87) {$\widetilde{P}(i,k,5)$};
	\node[regular polygon, regular polygon sides=6, minimum width=2cm,draw] (reg51) at (-1.52,-0.87){};
	
	% plot polygon 6 of agent i
	%\node at (-1.52,0.87) {$\widetilde{P}(i,k,6)$};
	\node[regular polygon, regular polygon sides=6, minimum width=2cm,draw] (reg61) at (-1.52,0.87){};
	
	% plot polygon 0 of agent j1
	%\node[regular polygon, regular polygon sides=6, minimum width=2cm,draw, fill = red!20] (reg02) at (3,0){};   
	%\node at (3.0, -0.32) {$x_{\scriptscriptstyle j_1}(t_k)$};
	%\node at (0.0, +0.33) {$P(i, k)$};
	
	% plot polygon 1 of agent j1
	%\node at (0,1.74) {$\widetilde{P}(i,k,1)$};
	\node[regular polygon, regular polygon sides=6, minimum width=2cm,draw] (reg12) at (3,1.74){};
	
	% plot polygon 2 of agent j1
	%\node at (0,1.74) {$\widetilde{P}(i,k,1)$};
	\node[regular polygon, regular polygon sides=6, minimum width=2cm,draw, fill = red!20] (reg12) at (4.52,0.87){};
	\node at (4.52,0.45) {$x_{\scriptscriptstyle j_1}(t_k)$};
	
	% plot polygon 3 of agent j1
	%\node at (1.52,-0.87) {$\widetilde{P}(i,k,3)$};
	\node[regular polygon, regular polygon sides=6, minimum width=2cm,draw] (reg31) at (4.52,-0.87){};
	
	% plot polygon 4 of agent j1
	%\node at (0,-1.74) {$\widetilde{P}(i,k,4)$};
	\node[regular polygon, regular polygon sides=6, minimum width=2cm,draw] (reg41) at (3,-1.74){};
	
	% plot polygon 0 of agent j2
	%\node at (1.52, 0.87) {$\widetilde{P}(i,k,2)$};
	\node[regular polygon, regular polygon sides=6, minimum width=2cm,draw, fill = blue!20] (reg03) at (1.52,2.60){};
	\node at (1.52, 2.98) {$x_{\scriptscriptstyle j_2}(t_k)$};
	
	% plot polygon 1 of agent j2
	%\node at (1.52, 0.87) {$\widetilde{P}(i,k,2)$};
	\node[regular polygon, regular polygon sides=6, minimum width=2cm,draw] (reg13) at (1.52,4.33){};
	
	% plot polygon 2 of agent j2
	%\node at (0,1.74) {$\widetilde{P}(i,k,1)$};
	\node[regular polygon, regular polygon sides=6, minimum width=2cm,draw] (reg12) at (3,3.48){};
	
	% plot polygon 5 of agent i
	%\node at (0,1.74) {$\widetilde{P}(i,k,1)$};
	\node[regular polygon, regular polygon sides=6, minimum width=2cm,draw] (reg11) at (0,3.48){};
	
	% draw the connections
	\draw[line width=.05cm, color = red, dashed] (0.0, 0.0) -- (1.52,2.60);
	\draw[line width=.05cm, color = red, dashed] (0.0, 0.0) -- (4.52,0.87);
	\draw[line width=.05cm, color = red, dashed] (1.52,2.60) -- (4.52,0.87);
	
	% draw the bullets
	\node at (0.0, 0.0) {$\bullet$};
	\node at (1.52,2.60) {$\bullet$};
	\node at (4.52,0.87) {$\bullet$};
	\node at (-1.52,0.87) {$\bullet$};
	\node at (-1.52,1.05) {$x_{i}(t_k+T)$};
	\node at (-1.52, 0.40) {$D_{\scriptscriptstyle \ell_{\text{des}}}$};
	\draw [line width = 0.04cm, color=blue,thick,->,>=stealth'] (0.0, 0.0) .. controls (-0.2, -0.2) .. (-1.52,0.87);
	\draw [line width = 0.04cm, color=orange,thick,->,>=stealth'] (0.0, -0.5) to (0.0, -1.5);
	\node at (0.0, -1.7) {$P(i,k)$};
	\node at (1.80,1.90) {$P(j_2,k)$};
	\node at (4.52, 1.40) {$P(j_1,k)$};
	\node at (-0.5, 0.4) {$u_i$};	
	\end{tikzpicture}
	\caption{Illustration of three connected agents $i, j_1, j_2$. The agents are occupying the regions $P(i,k) = D_{\scriptscriptstyle \ell_{i}}, P(j_1, k) = D_{\scriptscriptstyle \ell_{j_1}}$ and $P(j_2, k) = D_{\scriptscriptstyle \ell_{j_2}}$ at time $t_k = t_0+kT$, depicted by green, red and blue color, respectively. Their corresponding neighboring regions $\bar{P}(i,k), \bar{P}(j_1,k)$ and $\widetilde{P}(j_2,k, \widetilde{\ell}), \widetilde{\ell} \in \{4,5,6\}$, respectively, are also depicted. $\widetilde{P}(i,k,6) = D_{\scriptscriptstyle \ell_{\text{des}}}$ is the desired region in which agent $i$ needs to move at time $T$ by applying a decentralized control law $u_i(x_i, x_{j_1}, x_{j_2})$.}
	\label{fig:agent_i_j1_j2_example}
\end{figure}

\subsubsection{Decentralized Controller Specification} 

Consider a time interval $[t_k, t_k+T]$. We state here the specifications that a decentralized feedback controller $u_i(x_i, \bar{x}_i)$ needs to guarantee so as agent $i$ to have a \emph{well-defined transition} between two neighboring regions within the time interval $[t_k, t_k+T]$.

\textbf{(S1)} The controller needs to take into consideration the dynamics \eqref{eq:system} and the constraints that are imposed by the bounds of Assumption 1.

\textbf{(S2)} Agent $i$ should move from one region $P(i,k) \in D$ to a neighboring region $\widetilde{P}(i, k, \widetilde{\ell})$, without intersecting other regions, irrespectively of which region its neighbors are moving to. Thus, since the duration of the transition is $T$, it is required that $x_i(t_k) \in P(i, k)$, $x_i(t_k+T) \in \widetilde{P}(i, k, \widetilde{\ell})$ and $x_i(t) \in P(i,k) \cup \widetilde{P}(i, k, \widetilde{\ell}), t \in (t_k, t_k+T)$. The neighbors of agent $i$ will move also to exactly one of their corresponding neighboring regions.

\begin{remark}
The reason for imposing the aforementioned constraints is due to the need of imposing safety specifications to the agents. Thus, it is required to be guaranteed that the agents will not cross more than one neighboring region within the duration of a transition $T$.
\end{remark}

\subsubsection{Error Dynamics} Let us define by $x_{i,k,\widetilde{\ell},\text{des}} \in \widetilde{P}(i, k, \widetilde{\ell})$ a reference point of the desired region $\widetilde{P}(i, k, \ell)$ which agent $i$ needs to occupy at time $t_k+T$. Define also by:
\begin{equation}
e_i(t) = x_i(t)-x_{i, k, \widetilde{\ell}, \text{des}}, t \in [t_k, t_k+T],
\end{equation}
the error which the controller $u_i$ needs to guarantee to become zero in the time interval $t \in [t_k, t_k+T]$. Then, the \emph{nominal error dynamics} are given by: 
\begin{equation} \label{eq:error_dynamics} 
\dot{e}_i(t) = g_i(e_i(t), \bar{x}_i(t), u_i(t)), t \in [t_k, t_k+T],
\end{equation}
with initial condition $e_i(t_k) = x_i(t_k)-x_{i, k, \widetilde{\ell}, \text{des}}$, where: $$g_i(e_i(t), \bar{x}_i(t), u_i(t)) = f_i(e_i(t)+x_{i, k, \widetilde{\ell}, \text{des}}, \bar{x}_i(t))+u_i(t).$$

\begin{property} \label{property:error_bound}
	According to Assumption \eqref{ass:dynamic_control_bounds}, at every time $s \in [t_k, t_k+T]$, with $t_k = t_0+ k T$, the error $e_i(s)$ of the state of agent $i$ is upper bounded by:
	\begin{equation} \label{eq:error_bound}
	\|e_i(s)\| \le \|e_i(t_k)\|+ (s-t_k)(M+u_{\max}), i \in \mathcal{V}.
	\end{equation}
\end{property}
\begin{proof}
	The proof can be found in Appendix \ref{app:proof_of_property_1}.
\end{proof}

\subsubsection{State Constraints} Before defining the ROCP we state here the state constraints that are imposed to the state of each agent. Define the set:
\begin{align*}
& X_i = \{x_i \in W, \bar{x}_i \in W^{N_i}: \notag \\
&\hspace{3mm} \|f_i(x_i, \bar{x}_i)\| \le M, \|x_i-x_j\| < \underline{r}, \forall j \in \mathcal{N}_i(0), \notag \\
&\hspace{3mm} x_i \in P(i,k) \cup \widetilde{P}(i,k,\widetilde{\ell}),  \widetilde{\ell} \in \mathbb{L} \}, \notag
\end{align*}
as the set that captures the state constraints of agent $i$. The first constraint in the set $X$ stands for the bound of Assumption \ref{ass:dynamic_control_bounds}; the second one stands for the connectivity requirement of agent $i$ with all its neighbors; the last one stands for the requirement each agent to transit from one region to exactly one desired neighboring region. In order to translate the constraints that are dictated for the state $x_i(t)$ into constraints regarding the error state $e_i(t)$ from \eqref{eq:error_dynamics}, define the set $E_i = X_i \oplus (-x_{i,k,\widetilde{\ell},\text{des}}).$ Then, the following implication holds: $x_i \in X_i \Rightarrow e_i \in E_i.$

\subsubsection{Control Design}

This subsection concerns the control design regarding the transition of agent $i$ to one neighboring region $\widetilde{P}(i, k, \widetilde{\ell})$, for some $\widetilde{\ell}\in\mathbb{L}$. The abstraction design, however, concerns all the neighboring regions $\bar{P}(i,k)$, for which we will discuss in the next subsection. 

The timed sequence $\mathcal{S}$ consists of intervals of duration $T$. Within every time interval $[t_k, t_k+T]$, each agent needs to be at time $t_k$ in region $P(i,k)$ and at time $t_k+T$ in a neighboring region $\widetilde{P}(i,k,\widetilde{\ell}), \widetilde{\ell} \in \mathbb{L}$. We assume that $T$ is related to the sampling time $h$ according to: $T = m h, m \in \mathbb{N}$. Therefore, within the time interval $[t_k, t_k+T]$, there exists $m+1$ sampling times. By introducing the notation $t_{\scriptscriptstyle k_{z}} \triangleq t_k+z h, \forall z \in \mathbb{M} \triangleq \{0, \dots, m\}$, we denote by $\{t_{\scriptscriptstyle k_{z}}\}_{z \in \mathbb{M}}$ the sampling sequence within the interval $[t_k, t_k+T]$. Note that $t_{\scriptscriptstyle k_{0}} = t_k$ and $t_{\scriptscriptstyle k_{m}} = t_k+T$. The indexes $k, z$ stands for the interval and for the sampling times within this interval, respectively. As it will be presented hereafter, at every sampling time $t_{\scriptscriptstyle k_{z}}, z \in \mathbb{M}$, each agents solves a ROCP.

%denoted by $\{t_{\scriptscriptstyle k_z}\}, z \in \{1, \ldots ,m+1\}$, with $t_{\scriptscriptstyle k_1} = t_k$, $t_{\scriptscriptstyle k_{m+1}} = t_k + T$, and $t_{\scriptscriptstyle k_z} = t_k + (z-1) h$. For a fixed $z \in \{1,\dots, m+1\}$, it holds that $t_{\scriptscriptstyle k_{z+\lambda}} = t_{\scriptscriptstyle k_z} + \lambda h, \forall \lambda \in \{0, \dots, m\}$. 

Our control design approach is based on Nonlinear Model Predictive Control (NMPC). NMPC has been proven to be efficient for systems with nonlinearities and state/input constraints. For details about NMPC we refer the reader to \cite{morrari_npmpc, cannon_2001_nmpc,  frank_2003_nmpc_bible, frank_1998_quasi_infinite, frank_2003_towards_sampled-data-nmpc, grune_2011_nonlinear_mpc, camacho_2007_nmpc, borrelli_2013_nmpc, fontes_2001_nmpc_stability, camacho_2002_input_to_state}. We propose here a sampled-data NMPC with decreasing horizon in order to design a controller that respects the desired specifications and guarantees the transition between regions at time $T$. In the proposed sampled-data NMPC, an open-loop Robust Optimal Control Problem (ROCP) is solved at every discrete sampling time instant $t_{\scriptscriptstyle k_z}, z \in \mathbb{M}$ based on the current error state information $e_i(t_{\scriptscriptstyle k_z})$. The solution is an optimal control signal $\hat{u}_{i}(t)$, for $t \in [t_{\scriptscriptstyle k_z}, t_{\scriptscriptstyle k_z}+T_{\scriptscriptstyle z}]$, where $T_{\scriptscriptstyle z}$ is defined as follows.
\begin{definition}
	A \emph{decreasing horizon policy} is defined by:
	\begin{equation} \label{eq:decr_horizon}
	T_{z} = T - z  h, z \in \mathbb{M}.
	\end{equation}
\end{definition}
This means that at every time sample $t_{\scriptscriptstyle k_z}$ in which the ROCP is solved, the horizon is decreased by a sampling time $h$.  The specific policy is adopted in order to enforce the controllers $u_i$ to guarantee that agent $i$ will reach the desired neighboring region at time $T$. \eqref{eq:decr_horizon} implies also that $t_{\scriptscriptstyle k_{z}} +T_z = t_k+T, \forall z \in \mathbb{M}$ A graphical illustration of the presented time sequences is given in Fig. \ref{fig:time_sequence}. 

The \emph{open-loop input signal} is applied in between the sampling instants and is given by the solution of the following Robust Optimal Control Problem (ROCP): $\mathcal{O}(k, x_i(t), \bar{x}_i(t), P(i,k), \widetilde{\ell}, x_{i, k, \widetilde{\ell}, \text{des}}), t \in [t_{\scriptscriptstyle k_z}, t_{\scriptscriptstyle k_z}+T_{\scriptscriptstyle z}],$ which is defined as:
\begin{subequations}
	\begin{align}
	&\hspace{-1mm}\min\limits_{\hat{u}_i(\cdot)} J_i(e_i(t_{\scriptscriptstyle k_z})),\hat{u}_i(\cdot)) = \notag \\
	&\hspace{-1mm} \min\limits_{\hat{u}_i(\cdot)} \left\{  V_i(\hat{e}_i(t_{\scriptscriptstyle k_z}+T_{\scriptscriptstyle z})) + \int_{t_{\scriptscriptstyle k_z}}^{t_{\scriptscriptstyle k_z}+T_{z}} \Big[ F_i(\hat{e}_i(s), \hat{u}_i(s)) \Big] ds \right\}  \label{eq:mpc_minimazation} \\
	&\hspace{-1mm}\text{subject to:} \notag \\
	&\hspace{1mm} \dot{\hat{e}}_i(s) = g_i(\hat{e}_i(s), \hat{\bar{x}}_i(s), \hat{u}_i(s)), \hat{e}_i(t_{\scriptscriptstyle k_z}) = e_i(t_{\scriptscriptstyle k_z}), \label{eq:diff_mpc}
	\end{align}
	\begin{align}
	&\hspace{1mm} \hat{e}_i(s) \in E_{s-t_{\scriptscriptstyle k_z}}^{i}, \hat{u}_i(s) \in \mathcal{U}_i, s \in [t_{\scriptscriptstyle k_z}, t_{\scriptscriptstyle k_z}+T_{\scriptscriptstyle z}], \label{eq:mpc_constrained_set} \\
	&\hspace{1mm} \hat{e}_i(t_{\scriptscriptstyle k_z}+T_{\scriptscriptstyle z})\in\mathcal{E}_i. \label{eq:mpc_terminal_set}
	\end{align}
\end{subequations}
The ROCP has as inputs the terms $k, x_i(t)$, $\bar{x}_i(t), P(i,k)$, $\widetilde{\ell}$, $x_{i, k, \widetilde{\ell}, \text{des}}$, for time $t \in [t_{\scriptscriptstyle k_z}, t_{\scriptscriptstyle k_z}+T_{\scriptscriptstyle z}]$. We will explain hereafter all the terms appearing in the ROCP problem \eqref{eq:mpc_minimazation}-\eqref{eq:mpc_terminal_set}. By hat $\hat{(\cdot)}$ we denote the predicted variables (internal to the controller), corresponding to the system \eqref{eq:error_dynamics} i.e., $\hat{e}_i(\cdot)$ is the solution of \eqref{eq:diff_mpc} driven by the control input $\hat{u}_{i}(\cdot): [t_{\scriptscriptstyle k_z}, t_{\scriptscriptstyle k_z}+T_{\scriptscriptstyle z}] \to \mathcal{U}_i$ with initial condition $\hat{e}_i(t_{\scriptscriptstyle k_z}) = e_i(t_{\scriptscriptstyle k_z})$. The set $E_{s-t_{\scriptscriptstyle k_z}}^{i}$ is a subset of $E_i$ and will be explicitly defined later.

\begin{figure}[t!]
	\centering
	\begin{tikzpicture}
	\draw [blue, line width = 1.5] (-0.5, 0 ) -- (7.5, 0);
	
	\fill[green] (-0.5,0) circle (2.5 pt);
	\fill[red] (1.6,0) circle (2.5 pt);
	\fill[red] (5.5,0) circle (2.5 pt);
	\fill[green] (7.5,0) circle (2.5 pt);
	
	\node at (-0.3, -0.4) {$t_{\scriptscriptstyle k_z}$};      
	\node at (1.5, -0.4) {$t_{\scriptscriptstyle k_{z+1}}$};            
	\node at (5.5, -0.4) {$t_{\scriptscriptstyle k_{z}}+ T_{\scriptscriptstyle z+1}$};            
	\node at (7.3, -0.4) {$t_{\scriptscriptstyle k_{z}}+ T_{\scriptscriptstyle z}$};
	\draw [line width = 0.04cm, color=orange,thick,->,>=stealth'] (-0.5, 0.4) to (1.6, 0.4);
	\draw [line width = 0.04cm, color=orange,thick,->,>=stealth'] (1.6, 0.4) to (-0.5, 0.4);         
	\node at (0.6, 0.8) {$h$};                
	\draw [line width = 0.04cm, color=orange,thick,->,>=stealth'] (5.5, 0.4) to (7.5, 0.4);
	\draw [line width = 0.04cm, color=orange,thick,->,>=stealth'] (7.5, 0.4) to (5.5, 0.4);         
	\node at (6.5, 0.8) {$h$}; 
	\draw [line width = 0.04cm, color=orange,thick,->,>=stealth'] (-0.5, -1.0) to (7.5, -1.0);
	\draw [line width = 0.04cm, color=orange,thick,->,>=stealth'] (7.5, -1.0) to (-0.5, -1.0);         
	\node at (3.8, -1.3) {$T_{\scriptscriptstyle z} = T$}; 
	\end{tikzpicture}
	\caption{The prediction horizon of the ROCP along with the times $t_{\scriptscriptstyle k_z} <$ $t_{\scriptscriptstyle k_{z+1}}$ $< t_{\scriptscriptstyle k_{z}} + T_{z+1}$ $< t_{\scriptscriptstyle k_{z}}$ $+ T_{z}$, with $t_{\scriptscriptstyle k_{z}} = t_{\scriptscriptstyle k_z} + z h$ and $T_{\scriptscriptstyle k_{z}} = T - z h, z \in \mathbb{M}$.}
	\label{fig:time_sequence}
\end{figure}
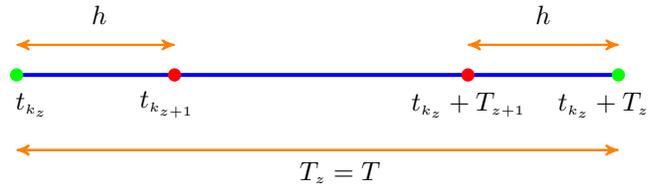
\begin{remark}
In sampled-data NMPC bibliography an ROCP is defined over the time interval $s \in \{t_i, t_{i+1} = t_i + h, \dots , t_i+T\}$, where $T$ is the prediction horizon. Due to the fact that we have denoted by $i$ the agents, and the fact that the ROCP is solved for every time interval, we use the notation $s \in \{t_{\scriptscriptstyle k_z} = t_k, t_{\scriptscriptstyle k_{z+1}} = t_k + h, \dots, t_{\scriptscriptstyle k_z} + T_{\scriptscriptstyle z} = t_{\scriptscriptstyle k_z} +T\}$, instead. The indexes $k, z$ stands for the interval and for the sampling time, respectively. A graphical illustration of the presented time sequence is given in Fig. \ref{fig:time_sequence}.
\end{remark}

\begin{remark}
	Note that the predicted values are not the same with the actual closed-loop values due to the fact that agent $i$, can not know the estimation of the trajectories of its neighbors $\hat{\bar{x}}$, within a predicted horizon. Thus, the term $\hat{\bar{x}}$ is treated as a disturbance to the \emph{nominal system} \eqref{eq:error_dynamics}.
\end{remark}

The term $F_i: E_i \times \mathcal{U}_i \to \mathbb{R}_{\ge 0}$, stands for the \emph{running cost}, and is chosen as:
\begin{equation*}
F_i(e_i, u_i) = e_i^\top Q_i e_i + u_i^\top R_i u_i,
\end{equation*}
where $Q_i=\text{diag}\{q_{i_1},q_{i_2}\}, R_i=\text{diag}\{\xi_{i_1},\xi_{i_2}\}$, with $q_{i_\zeta}\in\mathbb{R}_{\ge 0}, \xi_{i_\zeta}\in\mathbb{R}_{>0}, \zeta\in\{1,2\} $. For the running cost, it holds that $F_i(0,0) = 0$, as well as:
\begin{align} \label{eq:F_lower_bound}
\underline{m}_i \|e_i\|^2 \leq F_i(e_i, u_i) \leq \bar{m}_i \|e_i\|^2,
\end{align}
where $\underline{m}_i, \bar{m}_i$ will be defined later. Note that $\underline{m}_i\|e_i\|^2$ is $\mathcal{K}$ function, according to Definition \ref{def:class_K}.

\begin{lemma}
	The running cost function $F_i(e_i, u_i)$ is Lipschitz continuous in $E_i \times \mathcal{U}_i$, with Lipschitz constant: $$L_{\scriptscriptstyle F_i} = 2\bar{\varepsilon}_i \sigma_{\max}(Q_i),$$ where:
	\begin{equation*}
	\bar{\varepsilon}_i = \sup_{e_i \in E_i} \{\|e_i\|\},
	\end{equation*}
	for all $e_i \in E_i, u_i \in \mathcal{U}_i$.
\end{lemma}
\begin{proof}
	The proof can be found in Appendix \ref{app:lemma_1}.
\end{proof}
Note that, according to \eqref{eq:error_bound}, the terms $\|e_i\|$ are bounded, for all $i \in \mathcal{V}$. The terms $V_i: E_i \to \mathbb{R}_{ > 0}$ and $\mathcal{E}_i \subseteq E_i$ are the \emph{terminal penalty cost} and \emph{terminal set}, respectively, and are used to enforce the stability of the system. The terminal cost is given by:
\begin{equation*}
V_i(e_i(t)) = e_i(t)^\top P_i e_i(t).
\end{equation*}
where $P_i=\text{diag}\{p_{i_1},p_{i_2}\}$, with $p_{i_\zeta}\in\mathbb{R}_{> 0}, \zeta\in\{1,2\}$. We choose $\underline{m}_i = \{q_{i_1},q_{i_2}, \xi_{i_1}, \xi_{i_2}\}$ and $\bar{m}_i = \{q_{i_1},q_{i_2}, \xi_{i_1}, \xi_{i_2}\}$.

The solution of the nominal model \eqref{eq:error_dynamics} at time $s \in [t_{\scriptscriptstyle k_z}, t_{\scriptscriptstyle k_z}+T_{\scriptscriptstyle z}]$, starting at time $t_{\scriptscriptstyle k_z}$ from an initial condition $e_i(t_{\scriptscriptstyle k_z})$, applying a control input $u_{i}: [t_{\scriptscriptstyle k_z}, s] \to \mathcal{U}_i$ is denoted by: $$e_i(s; u_i(\cdot), e_i(t_{\scriptscriptstyle k_z})), s \in [t_{\scriptscriptstyle k_z}, t_{\scriptscriptstyle k_z}+T_{\scriptscriptstyle z}].$$ The predicted state of the system \eqref{eq:error_dynamics} at time $s \in [t_{\scriptscriptstyle k_z}, t_{\scriptscriptstyle k_z}+T_{\scriptscriptstyle z}]$ is denoted by: $$\hat{e}_i(s; u_i(\cdot), e_i(t_{\scriptscriptstyle k_z})), s \in [t_{\scriptscriptstyle k_z}, t_{\scriptscriptstyle k_z}+T_{\scriptscriptstyle z}],$$ and  it is based on the measurement of the state $e_i(t_{\scriptscriptstyle k_z})$ at time $t_{\scriptscriptstyle k_z}$, when a control input $u_{i}(\cdot; e_i(t_{\scriptscriptstyle k_z}))$ is applied to the system \eqref{eq:error_dynamics} for the time period $[t_{\scriptscriptstyle k_z}, s]$. Thus, it holds that:
\begin{equation} \label{eq:predicted_state_relation}
e_i(s) = \hat{e}_i(s; u_i(\cdot), e_i(s)), s \in [t_{\scriptscriptstyle k_z}, t_{\scriptscriptstyle k_z}+T_{\scriptscriptstyle z}].
\end{equation}

The state measurement enters the system via the initial condition of \eqref{eq:diff_mpc} at the sampling instant, i.e. the system model used to predict the future system behavior is initialized by the actual system state. The solution of the ROCP  \eqref{eq:mpc_minimazation}-\eqref{eq:mpc_terminal_set} at time $t_{\scriptscriptstyle k_z}$ provides an optimal control input denoted by $\hat{u}^\star_{i}(t; e(t_{\scriptscriptstyle k_z}))$, for $t \in [t_{\scriptscriptstyle k_z}, t_{\scriptscriptstyle k_z}+T_{\scriptscriptstyle z}]$. It defines the open-loop input that is applied to the system until the next sampling instant $t_{\scriptscriptstyle k_{z+1}}$:
\begin{equation} \label{eq:control_input_star}
u_{i}(t; e_i(t_i)) = \hat{u}^\star_{i}(t_{\scriptscriptstyle k_z}; e_i(t_{\scriptscriptstyle k_z})), t \in [t_{\scriptscriptstyle k_z}, t_{\scriptscriptstyle k_{z+1}}).
\end{equation} 
The corresponding \emph{optimal value function} is given by:
\begin{equation} \label{eq:J_star}
J_{i}^\star(e_i(t_{\scriptscriptstyle k_z})) \triangleq J_{i}(e_i(t_{\scriptscriptstyle k_z}), \hat{u}_{i}^\star(\cdot; e_i(t_{\scriptscriptstyle k_z}))).
\end{equation}
with $J_i(\cdot)$ as is given in \eqref{eq:mpc_minimazation}. The control input $u_i(t; e_i(t_{\scriptscriptstyle k_z}))$ is of the feedback form, since it is recalculated at each sampling instant using the new state information. Define an admissible control input as:

\begin{definition} \label{def:admissible_control_input}
	A control input $u_i: [t_{\scriptscriptstyle k_z}, t_{\scriptscriptstyle k_z}+T_{\scriptscriptstyle z}] \to \mathbb{R}^{2}$ for a state $e(t_{\scriptscriptstyle k_z})$ is called \emph{admissible}, if all the following hold:
	\begin{enumerate}
		\item $u_i(\cdot)$ is piecewise continuous;
		\item $u_i(s) \in \mathcal{U}_i, \forall \ s \in [t_{\scriptscriptstyle k_z}, t_{\scriptscriptstyle k_z}+T_{\scriptscriptstyle z}]$;
		\item $e_i(s; u_i(\cdot), e(t_{\scriptscriptstyle k_z})) \in E_i, \forall \ s \in [t_{\scriptscriptstyle k_z}, t_{\scriptscriptstyle k_z}+T_{\scriptscriptstyle z}]$;
		\item $e_i(T_{\scriptscriptstyle z}; u_i(\cdot), e(t_{\scriptscriptstyle k_z})) \in \mathcal{E}_i$;
	\end{enumerate}
\end{definition}
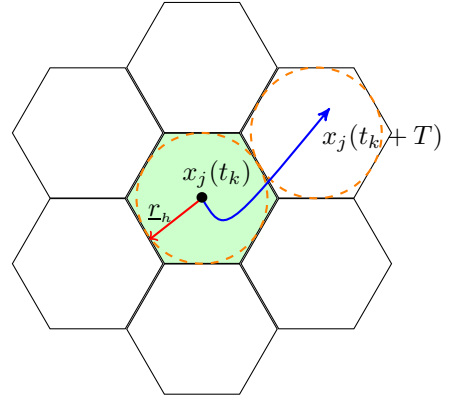
\begin{figure}[t!]
	\centering
	\begin{tikzpicture}  
	% plot polygon 0	
	\node[regular polygon, regular polygon sides=6, minimum width=2cm,draw, fill = green!20] (reg1) at (0,0){};   
	\node at (0.0, 0.0) {$\bullet$};
	%\node at (0.0, -0.32) {$x_i(t_k)$};
	%\node at (0.0, +0.33) {$P(i, k)$};
	
	% plot polygon 2
	%\node at (1.52, 0.87) {$\widetilde{P}(i,k,2)$};
	\node[regular polygon, regular polygon sides=6, minimum width=2cm,draw] (reg2) at (1.52,0.87){};
	
	% plot polygon 3
	%\node at (1.52,-0.87) {$\widetilde{P}(i,k,3)$};
	\node[regular polygon, regular polygon sides=6, minimum width=2cm,draw] (reg3) at (1.52,-0.87){};
	
	% plot polygon 4
	%\node at (0,-1.74) {$\widetilde{P}(i,k,4)$};
	\node[regular polygon, regular polygon sides=6, minimum width=2cm,draw] (reg4) at (0,-1.74){};
	
	% plot polygon 5
	%\node at (-1.52,-0.87) {$\widetilde{P}(i,k,5)$};
	\node[regular polygon, regular polygon sides=6, minimum width=2cm,draw] (reg5) at (-1.52,-0.87){};
	
	% plot polygon 6
	%\node at (-1.52,0.87) {$\widetilde{P}(i,k,6)$};
	\node[regular polygon, regular polygon sides=6, minimum width=2cm,draw] (reg5) at (-1.52,0.87){};
	
	% plot polygon 1
	%\node at (0,1.74) {$\widetilde{P}(i,k,1)$};
	\node[regular polygon, regular polygon sides=6, minimum width=2cm,draw] (reg5) at (0,1.74){};
	
	% draw inner circle 
	\draw[orange, thick, dashed] (0,0) circle (0.87cm);
	\draw[orange, thick, dashed] (1.52,0.87) circle (0.87cm);
	%\draw[orange, thick, dashed] (-1.52,-0.87) circle (0.87cm);
	
	% draw worst case distance
	\draw[->, thick, red] (-0.0, -0.0) -- (-0.70, -0.55);
	%\draw[->, thick, red] (0.70,0.55) -- (-0.70, -0.55);
	
	\draw [line width = 0.04cm, color=blue,thick,->,>=stealth'] (0.0, 0.0) .. controls (0.3, -0.5) .. (1.70,1.2);
	
	\node at (0.0, 0.0) {$\bullet$};
	\node at (0.2, +0.32) {$x_j(t_k)$};
	\node at (2.4, 0.8) {$x_j(t_k+T)$};
	\node at (-0.55, -0.20) {$\underline{r}_{\scriptscriptstyle h}$};
	\end{tikzpicture}
	\caption{Illustration of agent $j$ occupying region $P(j, k)$, depicted by green, at time $t_k = t_0+k T$ along with the regions $\bar{P}(j,k)$. It is desired for agent $j$ to move to region $\widetilde{P}(j,k,2)$ at precise time $T$. The inscribed circle of regions $P(j, k), \widetilde{P}(j, k, 2)$ are depicted with dashed orange color. The radius of the inscribed circle of the depicted hexagons is given by $\underline{r}_{\scriptscriptstyle h} = \frac{\sqrt{3}}{2}R$. By taking into consideration that each agent is moving at most to one neighboring region, according to the constraint set $X_j$, the following holds: $\sup \{\|x-y\|: x \in P(j,k), y \in \bar{P}(j,k)\} = 4\underline{r}_{\scriptscriptstyle h} = 2 \sqrt{3} R.$}
	\label{fig:largest_distance}
\end{figure}
\begin{property}
For the given hexagonal regions with side length $R$, the radius of the inscribed circle is given by $\underline{r}_{\scriptscriptstyle h} = \frac{\sqrt{3}}{2}R$ (two inscribed circles for the given regions are depicted with orange in Fig. \ref{fig:largest_distance}). Thus, according to Fig. \ref{fig:largest_distance}, an upper bound of the norm of differences between the actual position $x_j$ and the estimated position $\hat{x}_j$ of the agent's $i$ neighbors states, is given by:
\begin{equation} \label{eq:bound_geitonon}
\|x_j-\hat{x}_j\| \le 4 \underline{r}_{\scriptscriptstyle h} = 2 \sqrt{3} R, j \in \mathcal{N}_i,
\end{equation}
due to the fact that each agent can transit at most to a neighboring region, according to the constraint set $X_i$.
\end{property}

\begin{lemma} \label{lemma:e-hat_e}
	In view of Assumptions \ref{ass:dynamic_control_bounds}, \ref{ass:lipsitch_f_x_bar_x}, the difference between the actual measurement $e_i(s) = e_i(s; u_i(s; e_i(t_{\scriptscriptstyle k_z})), e_i(t_{\scriptscriptstyle k_z}))$ at time $s \in [t_{\scriptscriptstyle k_z}, t_{\scriptscriptstyle k_z}+T_{\scriptscriptstyle z}]$ and the predicted state $\hat{e}_i(s; u_i(s; e_i(t_{\scriptscriptstyle k_z})), e_i(t_{\scriptscriptstyle k_z}))$ at the same time under the same control law $u_i(s; e_i(t_{\scriptscriptstyle k_z}))$, starting at the same initial state $e_i(t_{\scriptscriptstyle k_z})$, is upper bounded by:
	\begin{align}
	&\|e_i(s) - \hat{e}_i(s; u_i(s; e_i(t_{\scriptscriptstyle k_z})), e_i(t_{\scriptscriptstyle k_z}))\| \le \rho_i(s-t_{\scriptscriptstyle k_z}),
	\end{align}
	where $\rho_i: \mathbb{R}_{\ge 0} \to \mathbb{R}$, with:
	\begin{align}
	&\rho_i(y) = \min \Big\{ \widetilde{\rho}_i \left[ e^{L_i y}- 1 \right],  \notag \\ 
	&\hspace{15mm} 2 \|e_i(t_{\scriptscriptstyle k_z})\|+ 2 y(M+u_{\max}) \Big\},
	\end{align}
	and
	\begin{equation}
	\widetilde{\rho}_i = \frac{2 \sqrt{3} R \bar{L}_i N_i}{L_i}.
	\end{equation}
\end{lemma}
\begin{proof}
	The proof can be found in Appendix \ref{app:lemma_e_ehat}.
\end{proof}

The satisfaction of the constraint on the state along the prediction horizon depends on the future evolution of the neighboring agents trajectories. Under Assumptions \eqref{ass:lipsitch_f_x_bar_x}, \eqref{ass:dynamic_control_bounds} of Lipschitz continuity and bounds of the nominal model, respectively, it is possible to compute a bound on the future effect of the disturbance on the system as is given by Lemma \ref{lemma:e-hat_e}. Then, by considering this effect on the state constraint on the nominal prediction, it is possible to guarantee that the evolution of the real state of the system will be admissible all the time. In view of latter, the state constraint set $E$ of the standard NMPC formulation, is being replaced by a restricted constrained set $E_{s-t_{\scriptscriptstyle k_z}}^{i} \subseteq E_i$ in \eqref{eq:mpc_constrained_set}. This state constraint's tightening for the nominal system \eqref{eq:error_dynamics} is a key ingredient of the robust NMPC controller and guarantees that the evolution of the real system will be admissible. Authors in \cite{camacho_2002_input_to_state, alina_ecc_2011} has considered such a Robust NMPC formulation. The restricted constrained set is then defined as $E_{s-t_{\scriptscriptstyle k_z}}^{i}= E_i \sim B^i_{s-t_{\scriptscriptstyle k_{z}}}$, where:
\begin{align}
&B^i_{s-t_{\scriptscriptstyle k_{z}}} = \notag \\
&\left\{e_i \in \mathbb{R}^{2}: \|e_i(s)\| \le \rho_i(s-t_i)\right\}, s \in [t_{\scriptscriptstyle k_z}, t_{\scriptscriptstyle k_z}+T_{\scriptscriptstyle z}]. \notag
\end{align}

\begin{property}
	For every $s \in [t_{\scriptscriptstyle k_z}, t_{\scriptscriptstyle k_z}+T_{\scriptscriptstyle z}]$, we have that if $$\hat{e}_i(s; u_i(s; e(t_{\scriptscriptstyle k_z})), e_i(t_{\scriptscriptstyle k_z})) \in E^i_{s-t_{\scriptscriptstyle k_{z}}} = E_i \sim B^i_{s-t_{\scriptscriptstyle k_{z}}} \subseteq E_i,$$ then the real state satisfies the constraint $E_i$,  i.e.,  $e_i(s) \in E_i.$
\end{property}
\begin{proof}
The proof can be found in Appendix \ref{app:property_tighthened_sets}.
\end{proof}
For the feasibility and convergence proofs of the ROCP the following assumptions are required.
\begin{assumption} \label{ass:admissible_u_f}
	Assume that there exists a \emph{local stabilizing controller} $u_{f,i} = \kappa_i(e_i) \in \mathcal{U}_i$ satisfying:
	\begin{equation} \label{eq:admiss_controller}
	\displaystyle \frac{\partial V_i}{\partial{e_i}} \left[g_i(e_i, \bar{x}_i, \kappa_i(e_i))\right] + F_i(e_i, \kappa_i(e_i)) \le 0, \forall \ e_i \in \Phi_i,
	\end{equation}
	where $\Phi_i$ is a set given by:
	\begin{equation*}
	\Phi_i \triangleq \{e_i \in \mathbb{R}^{2}: V_i(e_i) \le \alpha_{1,i}\}, \alpha_{1,i} > 0,
	\end{equation*}
	such that:
	\begin{equation*}
	\Phi_i \subseteq \mathbb{E}_i \triangleq  \{e_i \in E^i_{\scriptscriptstyle T_z}: \kappa_i(e_i) \in \mathcal{U}_i\},
	\end{equation*}
	where $E^i_{\scriptscriptstyle T_z} = E_i \sim B^i_{\scriptscriptstyle T_z}.$
\end{assumption}

\begin{lemma} \label{eq:lemma_L_v}
	The terminal penalty function $V_i(\cdot)$ is Lipschitz in $\Phi_i$, with Lipschitz constant $L_{\scriptscriptstyle V, i} = 2 \sigma_{\max}(P_i) \sqrt{ \frac{\alpha_{1,i}}{\lambda_{\min}(P_i)}}$, for all $e_i(t) \in \Phi_i$.
\end{lemma}
\begin{proof}
	The proof can be found in Appendix \ref{app:lemma_lip_L_v}.
\end{proof}

Once the set $\Phi_i$ is computed, the terminal constraint set $\mathcal{E}_i$ is given by the following. Supposing that Assumption \ref{ass:admissible_u_f} holds. Then, by choosing: $\mathcal{E}_i = \{e_i \in \mathbb{R}^2: \|e_i\| \le \sqrt{\frac{\alpha_{2,i}}{\lambda_{\min}(P_i)}} < r_{\scriptscriptstyle h}\}$, with $\alpha_{2,i} \in (0, \alpha_{1,i})$, we guarantee the following: $1)$ $\mathcal{E}_i \subseteq \widetilde{P}(i, k, \widetilde{\ell})$, i.e. the terminal set is a subset of the desired neighboring region; $2)$ for all $e_i \in \Phi_i$ it holds that $g_i(e_i, \kappa_i(e_i)) \in \mathcal{E}_i$.

%\begin{assumption}
%There exist $\alpha_{2,i} \in (0, \alpha_{1,i})$ such that $\mathcal{E}_i \triangleq \{e_i \in %\mathbb{R}^{2}: V_i(e_i) \le \alpha_{2,i}\} \subseteq \widetilde{P}(i, k, \widetilde{\ell}), %\widetilde{\ell} \in \mathbb{L}$ and for all $e_i \in \Phi_i \Rightarrow g_i(e_i, \kappa_i(e_i)) %\in \mathcal{E}_i$.
%\end{assumption}

The following two lemmas are required in order to prove the basic Theorem or this paper.
\begin{lemma}\label{lemma:x_y_proof}
	Let $s \ge t_{\scriptscriptstyle k_{z+1}}$, $x \in E^i_{s-t_{\scriptscriptstyle k_{z}}}$ and $y \in \mathbb{R}^{2}$ such that: $\|x-y\| \le \rho_i(t_{\scriptscriptstyle k_{z+1}}-t_{\scriptscriptstyle k_{z}}) = \rho_i(h)$, as $\rho_i$ is given in Lemma \ref{lemma:e-hat_e}. Then, it holds that $y \in E^i_{s-t_{\scriptscriptstyle k_{z+1}}}$.
\end{lemma}
\begin{proof}
	The proof can be found in Appendix \ref{app:y_E_t_i_plus_1}.
\end{proof}

\begin{lemma} \label{lemma:bounded_trajectories}
	Let $s \ge t_{\scriptscriptstyle k_{z}}$. The difference between two estimated trajectories $\hat{e}_i(s; u_i(\cdot), e_i(t_{\scriptscriptstyle k_{z+1}})), \hat{e}_i(s; u_i(\cdot), e_i(t_{\scriptscriptstyle k_{z}}))$ at time $s$, starting from from initial points $t_{\scriptscriptstyle k_{z+1}}$, $t_{\scriptscriptstyle k_{z}}$, respectively, under the same control input $u_i(\cdot)$, is upper bounded by:
	\begin{align}
	&\|\hat{e}_i(s; u_i(\cdot), e_i(t_{\scriptscriptstyle k_{z+1}}))- \hat{e}_i(s; u_i(\cdot), e_i(t_{\scriptscriptstyle k_{z}}))\| \le \notag \\
	&\hspace{40mm}\rho_i(t_{\scriptscriptstyle k_{z+1}}-t_{\scriptscriptstyle k_{z}}) = \rho_i(h).
	\end{align}
\end{lemma}
\begin{proof}
	The proof can be found in Appendix \ref{app:bounded_trajectories}.
\end{proof}

\begin{theorem}
	Suppose that Assumptions \ref{ass:measurement_assumption}-\ref{ass:admissible_u_f} hold. If the ROCP is feasible at time $t_k$, then, the closed loop system \eqref{eq:error_dynamics} of agent $i$, under the control input \eqref{eq:control_input_star}, starting its motion at time $t_k = t_0 + k T$ from region $P(i,k)$, is Input to State Stable (ISS) (for ISS see \cite{sontag_2008_ISS}) and its trajectory converges to the admissible positively invariant terminal set $\mathcal{E}_i$ exactly at time $t_k+T$, if it holds that $\rho_i(T_{\scriptscriptstyle z}) \le \bar{\rho}_i \triangleq \frac{\alpha_{1,1}-\alpha_{2,i}}{L_{\scriptscriptstyle V_i}}$.
\end{theorem}

\begin{proof}
	The proof consists of two parts: in the first part it is established that initial feasibility implies feasibility afterwards. Based on this result it is then shown that the error $e_i(t)$ converges to the terminal set $\mathcal{E}_i$. The \emph{feasibility analysis} as well as the \emph{convergence analysis} can be found in Appendix \ref{app:feasibility_convergence}.
\end{proof}

Assumption \ref{ass:admissible_u_f} is common in the NMPC literature. Many methodologies on how to compute $\Phi_i$ and controllers $u_{f,i} = \kappa_i(e_i)$, if they exist, have been proposed. We refer the reader to \cite{frank_1998_quasi_infinite, mayne_2000_nmpc}. Regarding the initial feasibility, numerical tools (e.g. \cite{grune_2011_nonlinear_mpc}) can be utilized in order to solve the ROCP and check if the problem is feasible at time $t_k = t_{\scriptscriptstyle k_{z}}$.

\begin{remark}
The term $\bar{\rho}_i, i \in \mathcal{V}$ gives an upper bound on the deviation of the trajectories of the neighboring agents of agent $i$ from their real values. If this bound is satisfied, agent $i$ can transit between the corresponding two neighboring regions, provided the ROCP is feasible at $t_{\scriptscriptstyle k_z}$.
\end{remark}

\begin{remark}
	It should be noted that, due to the nonlinear coupling terms $f_i(x_i, \bar{x}_i)$, the desired connectivity specifications and the bounds of Assumption \ref{ass:dynamic_control_bounds}, some of the ROCPs for $k \in N$ might not have a feasible solution. Let $i' \in \mathcal{V}, k' \in \mathbb{N}, \widetilde{\ell}' \in \mathbb{L}$ represent an agent $i'$ that at time step $t_{k'} = t_0 + k'T$ is desired to transit from region $P(i', k')$ to region $\widetilde{P}(i', k', \widetilde{\ell}')$. If the ROCP $\mathcal{O}(k', x_{i'}(t), \bar{x}_{i'}(t), P(i',k'), \widetilde{\ell}', x_{i', k', \widetilde{\ell}', \text{des}}), t \in [t_{\scriptscriptstyle k'_{z}}, t_{\scriptscriptstyle k'_{z}}+T_{z}],$ has no solution, then there does not exist admissible controller that can drive agent $i'$ from $P(i', k')$ to region $\widetilde{P}(i', k', \widetilde{\ell}')$. Our goal, through the proposed approach, is to seek all the possible solutions of the ROCP, which implies to seek for all possible transitions that will form later the individual WTS $\mathcal{T}_i$ of each agent. In this way, the resulting WTS $\mathcal{T}_i$ will capture the coupling dynamics \eqref{eq:system} and the transition possibilities of agent $i$ in the best possible way.
\end{remark}

\subsubsection{Generating the WTSs} 

Each agent $i \in \mathcal{V}$ solves the ROCP \ref{eq:mpc_minimazation}-\ref{eq:mpc_terminal_set} for every time interval $[t_{\scriptscriptstyle k_z}, t_{\scriptscriptstyle k_z}+T_{\scriptscriptstyle z}], k \in \mathbb{N}$, for all the desired neighboring regions $\widetilde{P}(i,k,\widetilde{\ell}), \widetilde{\ell} \in \mathbb{L}$. This procedure is performed by off-line simulation, i.e., at each sampling time $t_{\scriptscriptstyle k_z}, z \in \mathbb{M}$, each agent exchanges information about its new state with its neighbors and simulates the dynamics \eqref{eq:error_dynamics}. Between the sampling times the estimation $\hat{\bar{x}}_i$ is considered to be a disturbance, as discussed earlier.

Algorithm 1 provides the off-line procedure in order to generate the transition relation for each agent. At time $t_0$ each agent $i$ calls the algorithm in order to compute all possible admissible controllers to all possible neighboring regions of the workspace. The term $\text{Transit}$, which is the output of the algorithm, is a matrix of control input sequences for all pairs of neighboring regions in the workspace, initialized at sequences of zeros. The function $\text{Point}2\text{Region}(\cdot)$ maps the point $x_i(t_k)$ to the corresponding region of the workspace.
The function $\text{Sampling}(\cdot)$ takes as input the interval $[t_k, t_k+T]$ and returns the $m+1$ samples of this interval. The notation $(u_i^\star)_{\scriptstyle k_z}$ stands for the $z$-th element of the vector $(u_i^\star)$. The function $\text{OptSolve}(k,x_i(t),\bar{x}_i(t),p,\widetilde{\ell})$ (i) solves the ROCP and the function $\text{UpdateStates}(x_i,\bar{x}_i)$ updates the states of agent $i$ and its neighbors after every sampling time.  If the OptSolve function does not return a solution, then there does not exist an admissible control input that can drive agent $i$ to the desired neighboring region. After utilizing Algorithm 1, the WTS of each agent is defined as follows:

%Define by $\text{Config}: \mathbb{N} \times W^{N_i+1} \times \mathbb{R}^2 \to W^{N_i+1},$ a function that maps at every step $k$, every current position of agent $i$ and its neighbors, and every control input $u_i$ to the regions that the agents belong at time $t_k + T$ (the time after each transition). For the given parameters, a ROCP is solved. If the problem has no solution then the algorithm searches the next desired neighboring region. If the problem has solution, then there exist an admissible control input that can drive agent $i$ to the desired neighboring region. Then, the transition $(P(i,k), u_i, \widetilde{P}(i,k,\widetilde{\ell})) \in \longrightarrow_i$ is enabled and the function $\text{Config}(\cdot)$ is updated with the new information for the control input of agent $i$ and the positions of agent $i$ and its neighbors at time $t_k+T$. 

\begin{algorithm}[t!]
	\caption{CreateTransitionRelation($\cdot$)}
	\begin{algorithmic}[1]
		\State \textbf{Input:} $i, x_i(t_0), \bar{x}_i(t_0);$
		%\State \textbf{Output1:} $(P(i,k), u_i, \widetilde{P}(i,k,\widetilde{\ell})) \in \longrightarrow_i;$
		\State \textbf{Output:} $\text{Transit}$; \Comment Matrix with regions$\backslash$control inputs;\\
		\State $\text{Transit} \leftarrow \text{zeros}(|\mathbb{I}|, 6)$; k = 0; Flag = False;
		\State $\text{List} \leftarrow \{\text{Point}2\text{Region}(x_i(t_0))\}$; \Comment Initialize
		\While {$\text{List} \neq \emptyset$}
		\For {$p \in \text{List}$} \Comment p is a region of the List;
		\For {$\widetilde{\ell} \in \mathbb{L}$}
		\State $t \leftarrow \text{Sampling} (t_k, t_k+T)$;
		\For {$t_{\scriptscriptstyle k_z} \in t, z \in \mathbb{M}$}
		\State $(u^\star_i)_{\scriptscriptstyle k_z} \leftarrow \text{OptSolve}(k, x_i(t), \bar{x}_i(t), p, \widetilde{\ell});$
		\State UpdateStates($x_i, \bar{x}_{i}$);
		\If {$(u^\star_i)_{\scriptscriptstyle k_z} = \emptyset$}
		\Comment $\nexists$ controller;
		\State Flag = True; \Comment search next region;
		\State break;
		\EndIf
		\EndFor
		\If {Flag = False}
		\State $u^\star_i \leftarrow \{(u^\star_i)_{\scriptscriptstyle k_z} \}_{z \in \mathbb{M}}$ \Comment $u_i$ found
		\State $\text{Transit}(p, \widetilde{\ell})\leftarrow u^\star_i$;
		\State $\text{List} \leftarrow \text{List} \cup \widetilde{P}(i,k,\widetilde{\ell})$
		\Else
		\State Flag = False;
		\EndIf
		\EndFor 
		\State $\text{List} \leftarrow \text{List} \backslash p$;
		\State $k = k +1;$
		\EndFor
		\EndWhile 
	\end{algorithmic}
	\label{alg:basic_algorithm}
\end{algorithm}

\begin{definition} \label{def: indiv_WTS}
	The motion of each agent $i \in \mathcal{V}$ in the workspace is modeled by the WTS $\ \mathcal{T}_i$ $= (S_i$, $S_i^{\text{init}}$, $Act_i$, $\longrightarrow_i$, $d_i$, $\Sigma_i$, $L_i)$ where: $S_i = \{D_\ell\}_{\ell \in \mathbb{I}}$ is the set of states of each agent; $S_i^{\text{init}} = P(i,0) \subseteq S_i$ is a set of initial states defined by the agents' initial positions $x_i(t_0) \in P(i,0)$ in the workspace; $Act_i$ is the set of actions containing the union of all the admissible control inputs $u_i \in \mathcal{U}_i$ that are a feasible solution to the ROCP and can drive agent $i$ between neighboring regions; $\longrightarrow_i \subseteq S_i \times Act_i \times S_i$ is the transition relation. We say that $(P(i,k), u_i, \widetilde{P}(i,k,\widetilde{\ell})) \in \longrightarrow_i$, $k \in \mathbb{N}, \widetilde{\ell} \in \mathbb{L}$ if there exist an admissible controller $u_i \in Act_i$ which at step $k$ drives the agent $i$ from the region $P(i,k)$ to the desired region $\widetilde{P}(i,k,\widetilde{\ell})$. Algorithm 1 gives the steps how the transition relation can be constructed. $d_i: \longrightarrow_i \rightarrow \mathbb{R}_{\ge 0}$, is a map that assigns a positive weight (duration) to each transition. The duration of each transition is exactly equal to $T$; $\Sigma_i$, is the set of atomic propositions; $L_i: S_i \rightarrow 2^{\Sigma_i}$, is the labeling function.
\end{definition}

The individual WTSs of the agents will allow us to work directly in the discrete level and design sequences of controllers that solve Problem \ref{problem: basic_prob}. Every WTS $\mathcal{T}_i, i \in \mathcal{V}$ generates timed runs and timed words of the form $r_i^t = (r_i(0), \tau_i(0))$ $(r_i(1), \tau_i(1))\ldots$, $w_i^t$ $= (w_i(0), \tau_i(0))$ $(w_i(1), \tau_i(1))\ldots$, respectively, over the set $2^{\Sigma_i}$ with $w_i(\mu) = L_i(r_i(\mu)), \tau_i(\mu) = \mu T, \forall \ \mu \ge 0$. The transition relation $\longrightarrow_i$ along with the output of the Algorithm 1, i.e, $\text{Transit}(\cdot)$, allows each agent to have all the necessary information in order to be able to make a decentralized plan in the discrete level that is presented hereafter. The relation between the timed words that are generated by the WTSs $\mathcal{T}_i, i \in \mathcal{V}$ with the timed service words produced by the trajectories $x_i(t), i \in \mathcal{V}, t \ge 0$ is provided through the following remark:

\begin{remark} \label{lemma:compliant_WTS_runs_with_trajectories}
	By construction, each timed word produced by the WTS $\mathcal{T}_i$ is a \emph{relaxed timed word} associated with the trajectory $x_i(t)$ of the system \eqref{eq:system}. Hence, if we find a timed word of $\mathcal{T}_i$ satisfying a formula $\varphi_i$ given in MITL, we also find for each agent $i$ a desired timed word of the original system, and hence trajectories $x_i(t)$ that are a solution to the Problem \ref{problem: basic_prob}. Therefore, the produced timed words of $\mathcal{T}_i$ are compliant with the relaxed timed words of the trajectories $x_i(t)$.
\end{remark}

\subsection{Controller Synthesis} \label{sec: synthesis}

\begin{figure*}[t!]
	\centering
	\begin{tikzpicture}[scale = 0.97] 
	
	% convert formulas phi_i to buchis A_i
	\draw[blue!70, line width=.04cm] (-15.6, 6.0) rectangle +(2.3, 0.9);
	\node at (-14.45, 6.45) {$\text{MITL2TBA}$};
	\draw[blue!70, line width=.04cm] (-15.6, 0.80) rectangle +(2.3, 0.9);
	\node at (-14.45, 1.25) {$\text{MITL2TBA}$};
	
	\draw[-latex, draw=black, line width = 1.0] (-14.5,6.0) -- (-14.5,4.6);
	\draw[-latex, draw=black, line width = 1.0] (-14.5,7.6) -- (-14.5,6.9);
	\draw[-latex, draw=black, line width = 1.0] (-14.5,0.1) -- (-14.5,0.8);
	\draw[-latex, draw=black, line width = 1.0] (-14.5,1.7) -- (-14.5,3.1);
	
	\node at (-14.5, 7.80) {$\varphi_1$};
	\node at (-14.5, -0.1) {$\varphi_N$};
	
	\node at (-14.0, 5.35) {$\mathcal{A}_1$};
	\node at (-14.0, 2.40) {$\mathcal{A}_N$};
	
	\node at (-14.5, 3.25) {$\otimes$};
	\node at (-14.5, 4.45) {$\otimes$};
	
	\node at (-14.5, 3.95) {$\vdots$};
	
	% create \tilde T_i
	\draw[-latex, draw=black, line width = 1.0] (-14.30,4.45) -- (-12.6,4.45);
	\draw[-latex, draw=black, line width = 1.0] (-14.30,3.25) -- (-12.6,3.25);
	
	\node at (-12.3, 4.47) {$\widetilde{\mathcal{T}}_1$};
	\node at (-12.3, 3.95) {$\vdots$};
	\node at (-12.3, 3.26) {$\widetilde{\mathcal{T}}_N$};
	
	% create T_i
	\draw[-latex, draw=black, line width = 1.0] (-16.35,4.45) -- (-14.7,4.45);
	\draw[-latex, draw=black, line width = 1.0] (-16.35,3.25) -- (-14.7,3.25);
	
	\node at (-16.7, 4.47) {$\mathcal{T}_1$};
	\node at (-16.7, 3.95) {$\vdots$};
	\node at (-16.7, 3.26) {$\mathcal{T}_N$};
	
	% create box algorithm
	\draw[-latex, draw=black, line width = 1.0] (-12.00,4.45) -- (-11.3,4.45);
	\draw[-latex, draw=black, line width = 1.0] (-12.00,3.25) -- (-11.3,3.25);
	
	%\draw[red!70, line width=.04cm] (-11.2, 2.80) rectangle +(1.8, 2.2);
	%\node at (-10.30, 4.20) {$\text{graph}$};
	%\node at (-10.30, 3.85) {$\text{search}$};
	%\node at (-10.30, 3.40) {$\text{algorithm}$};
	
	\draw[red!70, line width=.04cm] (-11.3, 2.85) rectangle +(1.85, 0.70);
	\draw[red!70, line width=.04cm] (-11.3, 4.15) rectangle +(1.85, 0.70);
	\node at (-10.32, 4.50) {$\text{synthesis}$};
	\node at (-10.32, 3.20) {$\text{synthesis}$};
	
	% create runs
	\draw[-latex, draw=black, line width = 1.0] (-9.35,4.45) -- (-8.65,4.45);
	\draw[-latex, draw=black, line width = 1.0] (-9.35,3.25) -- (-8.65,3.25);
	
	\node at (-8.35, 4.47) {$\widetilde{r}_1^t$};
	\node at (-8.35, 3.95) {$\vdots$};
	\node at (-8.35, 3.26) {$\widetilde{r}_N^t$};
	
	% create abstractions box
	\draw[-latex, draw=black, line width = 1.0] (-17.70,4.47) -- (-17.00,4.47);
	\draw[-latex, draw=black, line width = 1.0] (-17.70,3.26) -- (-17.00,3.26);
	
	\draw[orange!70, line width=.04cm] (-19.6, 2.90) rectangle +(1.8, 0.7);
	\draw[orange!70, line width=.04cm] (-19.6, 4.20) rectangle +(1.8, 0.7);
	\node at (-18.70, 4.60) {$\text{abstraction}$};
	\node at (-18.70, 3.30) {$\text{abstraction}$};
	
	\draw[-latex, draw=black, line width = 1.0] (-20.35,4.47) -- (-19.65,4.47);
	\draw[-latex, draw=black, line width = 1.0] (-20.35,3.26) -- (-19.65,3.26);
	
	% create the dynamics
	
	\node at (-22.70, 4.47) {$\displaystyle \dot{x}_1 = f_1(x_1, \bar{x}_1)+u_{1}$};
	\node at (-22.70, 3.95) {$\vdots$};
	\node at (-22.70, 3.26) {$\displaystyle \dot{x}_N = f_N(x_N, \bar{x}_N)+u_{N}$};
	
	% create control inputs v_1,...,v_N
	\draw [black, line width = 0.030cm] (-8.35, 4.80) -- (-8.35, 8.50);
	\draw [black, line width = 0.030cm] (-8.35, 8.50) -- (-18.00, 8.50);
	\draw [black, line width = 0.030cm] (-19.00, 8.50) -- (-23.00, 8.50);
	\draw[-latex, draw=black, line width = 1.0] (-23.00, 8.50) -- (-23.00, 5.1);
	
	\draw [black, line width = 0.030cm] (-8.35, 2.90) -- (-8.35, -0.80);
	\draw [black, line width = 0.030cm] (-8.35, -0.80) -- (-18.00, -0.80);
	\draw [black, line width = 0.030cm] (-19.00, -0.8) -- (-23.00, -0.8);
	\draw[-latex, draw=black, line width = 1.0] (-23.00, -0.8) -- (-23.00, 2.6);
	
	\draw[green!70, line width=.04cm] (-19.0, 7.95) rectangle +(1.0, 1.0);
	\node at (-18.5, 8.40) {$u_1(t)$};
	\draw[green!70, line width=.04cm] (-19.0, -1.25) rectangle +(1.0, 1.0);
	\node at (-18.5, -0.80) {$u_N(t)$};
	\end{tikzpicture}
	\centering
	\caption{A graphic illustration of the proposed framework.}
	\label{fig:solution_scheme}
\end{figure*}
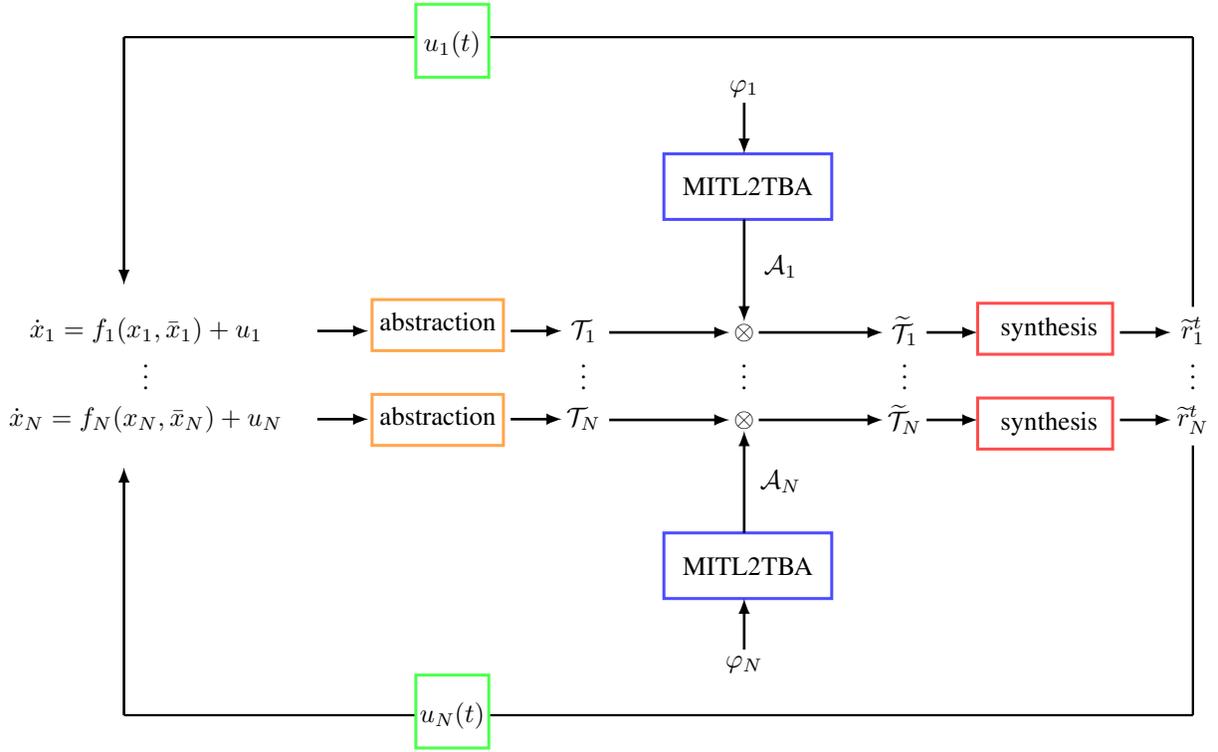

The proposed controller synthesis procedure is described with the following steps:

\begin{enumerate}
	\item $N$ TBAs $\mathcal{A}_i, \ i \in \mathcal{V}$ that accept all the timed runs satisfying the corresponding specification formulas $\varphi_i, i \in \mathcal{V}$ are constructed.
	\item A B\"uchi WTS $\widetilde{\mathcal T}_i = \mathcal{T}_i \otimes \mathcal{A}_i$ (see Def. \ref{def: buchi_WTS} below) is constructed for every $i \in \mathcal{V}$. The accepting runs of $\widetilde{\mathcal T}_i$ are the individual runs of $\mathcal{T}_i$ that satisfy the corresponding MITL formula $\varphi_i, \ i \in \mathcal{V}$.
	\item The abstraction procedure allows to find an explicit feedback law for each
	transition in $\mathcal T_i$. Therefore,
	an accepting run $\widetilde{r}^t_i$ in $\mathcal T_i$ that takes the form of a sequence of transitions is realized in the system in \eqref{eq:system} via the corresponding sequence of feedback laws.
\end{enumerate}

\begin{definition} \label{def: buchi_WTS}
	Given a WTS $\mathcal{T}_i =(S_i, S_{i}^{\text{init}}, Act_i, \longrightarrow_i, d_i, \Sigma_i, L_i)$, and a TBA $\mathcal{A}_i = (Q_i,  Q^\text{init}_i, C_i, Inv_i, E_i, F_i, \\ \Sigma_i, \mathcal{L}_i)$ with $|C_i|$ clocks and let $C^{\max}_i$ be the largest constant appearing in $\mathcal{A}_i$. Then, we define their \textit{B\"uchi WTS} $\widetilde{\mathcal{T}}_i = \mathcal{T}_i \otimes \mathcal{A}_i = (\widetilde{S}_i, \widetilde{S}_{i}^{\text{init}}, \widetilde{Act}_i, {\rightsquigarrow}_{i}, \widetilde{d}_i, \widetilde{F}_i, \Sigma_i, \widetilde{L}_i)$ as follows: 
	\begin{itemize}
		\item $\widetilde{S}_i \subseteq \{(s_i, q_i) \in S_i \times Q_i : L_i(s_i) = \mathcal{L}_i(q_i)\} \times \mathbb{T}_\infty^{|C_i|}$.
		\item $\widetilde{S}_{i}^{\text{init}} = S_i^{\text{init}} \times Q_i^{\text{init}} \times \{0\}^{|C_i|}$.
		\item $\widetilde{Act}_i = Act_i$.
		\item $(\widetilde{q}, act_i, \widetilde{q} ') \in {\rightsquigarrow}_i$ iff
		\begin{itemize}
			\item[$\circ$] $\widetilde{q} = (s, q, \nu_1, \ldots, \nu_{|C_i|}) \in \widetilde{S}_i$, \\ $\widetilde{q} ' = (s', q', \nu_1', \ldots, \nu_{|C_i|}') \in \widetilde{S}_i$,
			\item[$\circ$] $act_i \in Act_i$,
			\item[$\circ$] $(s, act_i, s') \in \longrightarrow_i$, and
			\item[$\circ$] there exists $\gamma, R$, such that $(q, \gamma, R, q') \in E_i$, $\nu_1,\ldots,\nu_{|C_i|} \models \gamma$, $\nu_1',\ldots,\nu_{|C_i|}' \models Inv_i(q')$, and for all $i \in \{1,\ldots, |C_i|\}$
			\begin{equation*}
			\nu_i' =
			\begin{cases}
			0,      & \text{if } c_i \in R \\
			\nu_i + d_i(s, s'), &  \text{if }  c_i \not \in R \text{ and } \\ & \nu_i + d_i(s, s') \leq C^{\mathit{max}}_i \\			\infty, & \text{otherwise}.
			\end{cases}
			\end{equation*}
		\end{itemize}
		Then, $\widetilde{d}_i(\widetilde{q}, \widetilde{q}') = d_i(s, s')$.
		\item $\widetilde{F}_i = \{(s_i, q_i,\nu_1,\ldots,\nu_{|C_i|}) \in Q_i : q_i \in F_i\}$.
		\item $\widetilde{L}_i(s_i, q_i, \nu_1, \ldots, \nu_{|C_i|}) = L_i(s_i)$.
	\end{itemize}
\end{definition}

Each B\"uchi WTS $\widetilde{\mathcal{T}}_i, i \in \mathcal{V}$ is in fact a WTS with a B\"uchi acceptance condition $\widetilde{F}_i$. A timed run of $\widetilde T_i$ can be written as $\widetilde{r}_i^t = (q_i(0), \tau_i(0))(q_i(1), \tau_i(1)) \ldots$ using the terminology of Def. \ref{run_of_WTS}. It is \textit{accepting} if $q_i(\mu) \in \widetilde F_i$ for infinitely many $j \geq 0$. An accepting timed run of  $\widetilde{\mathcal{T}}_i$ projects onto a timed run of $\mathcal{T}_i$ that satisfies the local specification formula $\varphi_i$ by construction. Formally, the following lemma, whose proof follows directly from the construction and and the principles of automata-based LTL model checking (see, e.g., \cite{katoen}), holds:

\begin{lemma} \label{eq: lemma_1}
Consider an accepting timed run $\widetilde{r}_i^t = (q_i(0), \tau_i(0))(q_i(1), \tau_i(1)) \ldots$ of the B\"uchi WTS $\widetilde T_i$ defined above, where $q_i(\mu) = (r_i(\mu), s_i(\mu), \nu_{i, 1}, \ldots, \nu_{i, |C_i|})$ denotes a state of $\mathcal{\widetilde T}_i$, for all $\mu \geq 0$. The timed run $\widetilde{r}_i^t$ projects onto the timed run $r_i^t = (r_i(0), \tau_i(0))(r_i(1), \tau_i(1)) \ldots$ of the WTS $\mathcal{T}_i$ that produces the timed word $w(r_i^t) = ({L}_i(r_i(0)), \tau_i(0))({L}_i(r_i(1)), \tau_i(1)) \ldots$ accepted by the TBA $\mathcal{A}_i$ via its run $\chi_i = s_i(0)s_i(1) \ldots$. Vice versa, if there exists a timed run $r_i^t = (r_i(0),\tau_i(0))(r_i(1),\tau_i(1))\ldots$ of the WTS $T_i$ that produces a timed word $w(r_i^t) = (L_i(r_i(0)), \tau_i(0))(L_i(r_i(1)), \tau_i(1)) \ldots$ accepted by the TBA $A_i$ via its run $\chi_i = s_i(0)s_i(1)\ldots$ then there exist the accepting timed run $\widetilde{r}_i^t = (q_i(0), \tau_i(0))(q_i(1), \tau_i(1)) \ldots$ of $\widetilde{T}_i$, where $q_i(z)$ denotes $(r_i(z), s_i(z), \nu_{i,1}, \ldots, \nu_{i,|C_i|})$ in $\widetilde{T}_i$.
\end{lemma}

The proposed framework is depicted in Fig. \ref{fig:solution_scheme}. The dynamics \eqref{eq:system} of each agent $i$ is abstracted into a WTS $\mathcal{T}_i$ (orange rectangles). Then the product between each WTS $\mathcal{T}_i$ and the $TBA$ $\mathcal{A}_i$ is computed according to Def. \ref{def: buchi_WTS}. The TBA $\mathcal{A}_i$ accepts all the words that satisfy the formula $\varphi_i$ (blue rectangles). For every B\"uchi WTS $\widetilde{\mathcal{T}}_i$ the controller synthesis procedure that was described in this Section (red rectangles) is performed and a sequence of accepted runs $\{\widetilde{r}_1^t, \dots, \widetilde{r}_N^t\}$ is designed. Every accepted run $\widetilde{r}_i^t$ maps into a decentralized controller $u_i(t)$ which is a solution to Problem \ref{problem: basic_prob}.

\begin{proposition}
The solution that we obtain from Steps 1-5, if one found, gives a sequence of controllers $u_1, \ldots, u_N$ that guarantees the satisfaction of the formulas formulas $\varphi_1, \ldots, \varphi_N$ of the agents $1, \ldots, N$ respectively, governed by dynamics as in \eqref{eq:system}. Thus, we solved Problem \ref{problem: basic_prob}.
\end{proposition}

\subsection{Complexity} \label{sec:complexity}

In the proposed abstraction technique $6^N$ MPC optimization problems are solved for every time interval $t \in [t_k, t_k + T]$. Assume that the desired horizon for the system to run is $M$ steps i.e. the timed sequence $\mathcal{S}$ is written as: $\mathcal{S} = \{t_0, t_1 = t_0+T, \dots, t_M = t_0+M T \}$. Then the complexity of the abstraction is $M 6^N.$ As for the controller synthesis framework now we have the following. Denote by $|\varphi|$ the length of an MITL formula $\varphi$. A TBA $\mathcal{A}_i, i \in \mathcal{V}$ can be constructed in space and time $2^{\mathbb{O}(|\varphi_i)|}, i \in \mathcal{V}$ ($\mathbb{O}$ stands for the ``big O" from complexity theory). Let $\varphi_{\text{max}} = \text{max} \{ |\varphi_i\}, i \in \mathcal{V}$ be the MITL formula with the longest length. Then, the complexity of Step 1 is $2^{\mathbb{O}(|\varphi_{\text{max}})|}$. The model checking of Step 2 costs $\mathbb{O}(|\mathcal{T}_i| 2^{|\varphi_i|}), i \in \mathcal{V}$ where $|\mathcal{T}_i|$ is the length of the WTS $\mathcal{T}_i$ i.e., the number of its states. Thus, $\mathbb{O}(|\mathcal{T}_i| 2^{|\varphi_i|}) = \mathbb{O}(|S_i| 2^{|\varphi_i|}) = \mathbb{O}(|{\mathbb{I}}| 2^{|\varphi_i|})$, where $|\mathbb{I}|$ is the number of cells of the cell decomposition $D$. The worst case of Step 2 costs $\mathbb{O}(|{\mathbb{I}}| 2^{|\varphi_{\text{max}}|})$ due to the fact that all WTSs $\mathcal{T}_i, i \in \mathcal{I}$ have the same number of states. Therefore, the complexity of the total framework is $\mathbb{O}(M |{\mathbb{I}}| 6^N 2^{|\varphi_{\text{max}}|})$.

\begin{figure}[t!]
	\centering
	\includegraphics[scale = 0.40]{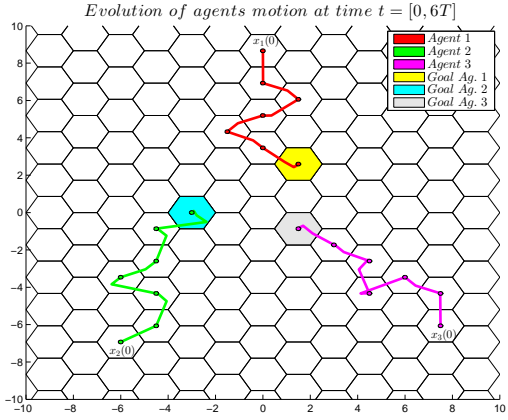}
	\caption{Evolution of the agent's trajectories up to time $6T$ in the workspace $W$. Each point-to-point transition has time duration $T = 3$. The depicted timed runs with $\rm red, \rm green$ and $\rm magenta$, of agents $1,2$ and $3$, satisfy the formulas $\varphi_1$, $\varphi_2$ and $\varphi_3$, respectively, while the agents remain connected.} 
	\label{fig:simulation}
\end{figure}

\section{Simulation Results} \label{sec: simulation_results}

For a simulation example, a system of three agents with $x_i \in \mathbb{R}^2, \ i \in \mathcal{V} = \{1, 2, 3\}$, $\mathcal{N}_1 = \{2,3\}$ $\mathcal{N}_2 = \{1,3\}$, $\mathcal{N}_3$ $= \{1, 1\}$ is considered. The workspace $W = [-10, 10] \times [-10, 10] \subseteq \mathbb{R}^2$ is decomposed into hexagonal regions with $R = 1, r_{\scriptscriptstyle h} = \frac{\sqrt{3}}{2}$, which are depicted in Fig. \ref{fig:simulation}. The initial agents' positions are set to $x_1(0) = (0,10r_{\scriptscriptstyle h}), x_2(0) = (-6,-8r_{\scriptscriptstyle h})$ and $x_3(0)=(7.5,-7r_{\scriptscriptstyle h})$. The sensing radius is $\underline{r} = 18$. The dynamics are set to: $\dot{x}_1 = -2x_1+x_2+x_3-\sin^2(x_1-x_2)+u_1$, $\dot{x}_2 = -2x_2+x_1+x_3-\sin^2(x_2-x_1)+u_2$ and $\dot{x}_3 = -2x_3+x_1+x_2+u_3$. The time step is $T = 3$. The specification formulas are set to $\varphi_1 = \Diamond_{[15, 27]} \{\rm red\}, \varphi_2 = \Diamond_{[7.5, 22]} \{\rm green\}, \varphi_3 = \Diamond_{[0, 19]} \{grey\}$ respectively. We set: $Q_i, P_i, R_i = I_2, \forall i \in \mathcal{V}$. Fig. \ref{fig:simulation} shows a sequence of transitions for agents $1,2$ and $3$ which form the accepting timed words $\widetilde{r}_1^t$, $\widetilde{r}_2^t$ and $\widetilde{r}_3^t$, respectively. Every timed word maps to a sequence of admissible control inputs for each agent, which is the outcome of solving the ROCPs. The agents remain connected for all $t \in [0, 6T]$. The simulations were carried out in MATLAB Environment by using the NMPC toolbox \cite{grune_2011_nonlinear_mpc}, on a desktop with 8 cores, 3.60GHz CPU and 16GB of RAM.

\section{Conclusions and Future Work} \label{sec: conclusions}
A systematic method of both decentralized abstractions and controller synthesis of a general class of coupled multi-agent systems has been proposed in which timed temporal specifications are imposed to the system. The solution involves a repetitive solving of an ROCP for every agent and for every desired region in order to build decentralized Transition Systems that are then used in the derivation of the controllers that satisfy the timed temporal formulas. Future work includes further computational improvement of the proposed decentralized abstraction method.

%%%%%%%%%%%%%%%%%%%%%%%%%%%%%%%%%%%%%%%%%%%%%%%%%%%%%%%%%%%%%%%%%%%%%%%%%%%%%%%%
\appendices
%\section*{Appendix}

\section{Proof of Property 1} \label{app:proof_of_property_1}
\begin{proof}
By integrating \eqref{eq:system} in the time interval $s \in [t_k, t_k+T]$ and taking the norms, we get:
\begin{align}
\|e_i(t)\| &= \left\|e_i(t_k) + \int_{t_k}^{t} \left[ g_i(x_i(s), \bar{x}_i(s), u_i(s))\right]ds \right\| \notag \\
&\le \left\|e_i(t_k) \right\| + \left\| \int_{t_k}^{t} \left[ f_i(e_i(s), \bar{x}_i(s)) + u_i(s)\right]ds \right\| \notag \\
&\le \left\|e_i(t_k) \right\| + \int_{t_k}^{t} \left\| f_i(e_i(s), \bar{x}_i(s) + u_i(s) \right\|ds \notag \\
&\le \left\|e_i(t_k) \right\| + \int_{t_k}^{t} \Big\{  \left\| f_i(e_i(s), \bar{x}_i(s) \right\| + \left\| u_i(s) \right\| \Big\} ds \notag \\
&\le \left\|e_i(t_k) \right\| +\int_{t_k}^{t} (M+u_{\max}) ds \notag \\
&= \left\|e_i(t_k) \right\| +(t-t_k) (M+u_{\max}), \notag
\end{align}
which concludes the proof.
\end{proof}

\section{Proof of Lemma 1} \label{app:lemma_1}

\begin{proof}
	\noindent For every $e_1,e_2 \in E_i, u \in \mathcal{U}_i, i \in \mathcal{V}$, the following holds:
	\begin{align}
	|F_i(e_1, u)-F_i(e_2, u) | &= \notag \\ 
	&\hspace{-12mm}\|e_1^\top Q_i e_1+u^\top R_i u - e_2^\top Q_i e_2 - u^\top R_i u | \notag
    \end{align}
    \begin{align}
	&\hspace{-12mm}= |e_1^\top Q_i e_1 - e_2^\top Q_i e_2 | \notag \\
	&\hspace{-12mm}= |e_1^\top Q_i e_1 +e_1^\top Q_i e_2 -e_1^\top Q_i e_2 - e_2^\top Q_i e_2 | \notag \\
	&\hspace{-12mm}= |e_1^\top Q_i (e_1-e_2) - e_2^\top Q_i ( e_1 - e_2) | \notag \\
	&\hspace{-12mm}\le |e_1^\top Q_i (e_1-e_2)| + |e_2^\top Q_i ( e_1 - e_2) |. \label{eq:lemma_1_proof_step_11}
	\end{align}
	By employing the property that:
	\begin{equation} \label{eq:sigma_max}
	|x^\top A y| \le \sigma_{\max}(A) \|x\| \|y\|, \forall \ x,y \in \mathbb{R}^n, A \in \mathbb{R}^{n \times n},
	\end{equation}
	\eqref{eq:lemma_1_proof_step_11} is written as:
	\begin{align*}
	&|F_i(e_1, u)-F_i(e_2, u) | \le \sigma_{\max}(Q_i) \|e_1\| \|e_1-e_2\| \notag \\ 
	&\hspace{35mm} +\sigma_{\max}(Q_i) \|e_2\| \|e_1-e_2\| \notag \\
	&\hspace{4mm}= \sigma_{\max}(Q_i) (\|e_1\| + \|e_2\|) \|e_1-e_2\| \notag \\
	&\hspace{4mm}= \sigma_{\max}(Q_i) \left[\sup_{e_1, e_2 \in E_i} \{\|e_1\| + \|e_2\|\}\right] \|e_1-e_2\| \notag \\
	&\hspace{4mm}= 2 \sigma_{\max}(Q_i) \left[\sup_{e_i \in E_i} \{\|e\|\}\right] \|e_1-e_2\| \notag \\
	&\hspace{4mm}= \left[ 2 \bar{\varepsilon}_i \sigma_{\max}(Q_i) \right] \|e_1-e_2\|.
	\end{align*}
	which completes the proof.
\end{proof}

\section{Proof of Lemma 2} \label{app:lemma_e_ehat}

\begin{proof}
Let us denote by:
\begin{align*}
u_i(\cdot) \triangleq u_i(s; e(t_{\scriptscriptstyle k_z})), \notag \\
e_i(s) \triangleq e_i(s; u_i(\cdot), e_i(t_{\scriptscriptstyle k_z})). \notag
\end{align*}
the control input and real trajectory of the system \eqref{eq:system} for $s \in [t_{\scriptscriptstyle k_z}, t_{\scriptscriptstyle k_z} + T_{\scriptscriptstyle z}]$. Also, denote for sake of simplicity:
\begin{equation*}
\hat{e}_i(s) \triangleq \hat{e}_i(s; u_i(\cdot), e_i(t_{\scriptscriptstyle k_z})).
\end{equation*}
the corresponding estimated trajectory. By integrating \eqref{eq:system}, \eqref{eq:diff_mpc} for the time interval $[t_{\scriptscriptstyle k_z}, t_{\scriptscriptstyle k_z}+s]$ we have the following:
\begin{align*}
e_i(s) &= e_i(t_{\scriptscriptstyle k_z}) +  \int_{t_{\scriptscriptstyle k_z}}^{s} \left[ g_i(e_i(s'), \bar{x}_i(s'), u_i(\cdot)) \right] ds', \notag \\
\hat{e}_i(s) &= e_i(t_{\scriptscriptstyle k_z}) +  \int_{t_{\scriptscriptstyle k_z}}^{s} \left[ g_i(\hat{e}_i(s'), \hat{\bar{x}}_i(s'), u_i(\cdot)) \right] ds', \notag
\end{align*}
respectively. Then, we have that:
\begin{align}
&\|e_i(s)-\hat{e}_i(s)\| \notag \\
&= \Bigg\|\int_{t_{\scriptscriptstyle k_z}}^{s} \left[ g(e_i(s'), \bar{x}_i(s'), u_i(\cdot)) \right] ds'- \notag \\
&\hspace{35mm} \int_{t_{\scriptscriptstyle k_z}}^{s} \left[ g(\hat{e}_i(s'), \hat{\bar{x}}_i(s'), u_i(\cdot)) \right] ds' \Bigg\| \notag \\
&= \Bigg\|\int_{t_{\scriptscriptstyle k_z}}^{s} \big[   f(e_i(s'), \bar{x}_i(s'))+u_{i}(s') \notag \\ 
&\hspace{35mm}-f(\hat{e}_i(s'), \hat{\bar{x}}_i(s'))-u_{i}(s')\big] ds' \Bigg\| \notag \\
&= \Bigg\|\int_{t_{\scriptscriptstyle k_z}}^{s} \big[  f(e_i(s'), \bar{x}_i(s'))- f(\hat{e}_i(s'), \hat{\bar{x}}_i(s'))\big] ds' \Bigg\| \notag \\
&\le \int_{t_{\scriptscriptstyle k_z}}^{s} \Big\| f(e_i(s'), \bar{x}_i(s'))- f(\hat{e}_i(s'), \hat{\bar{x}}_i(s')) \Big\| ds'  \notag
\end{align}
\begin{align}
&= \int_{t_{\scriptscriptstyle k_z}}^{s} \Big\| f(e_i(s'), \bar{x}_i(s')) -f(\hat{e}_i(s'), \bar{x}_i(s')) \notag \\
&\hspace{20mm}+ f(\hat{e}_i(s'), \bar{x}_i(s'))- f(\hat{e}_i(s'), \hat{\bar{x}}_i(s'))\Big\| ds'  \notag \\
&\le \int_{t_{\scriptscriptstyle k_z}}^{s} \Big\| f(e_i(s'), \bar{x}_i(s')) -f(\hat{e}_i(s'), \bar{x}_i(s')) \big\|ds' \notag \\ 
&\hspace{15mm}+ \int_{t_{\scriptscriptstyle k_z}}^{s} \Big\| f(\hat{e}_i(s'), \bar{x}_i(s'))- f(\hat{e}_i(s'), \hat{\bar{x}}_i(s'))\Big\| ds'.  \notag
\end{align}
By using the bounds of \eqref{eq:lip_1}-\eqref{eq:lip_2} we obtain:
\begin{align}
&\|e_i(s)-\hat{e}_i(s)\| \notag \\
&\le \int_{t_{\scriptscriptstyle k_z}}^{s} L_i \big\| e_i(s') -\hat{e}_i(s') \big\|ds' \notag \\
&\hspace{30mm}+ \int_{t_{\scriptscriptstyle k_z}}^{s} \bar{L}_i \big\| \bar{x}_i(s')- \hat{\bar{x}}_i(s') \big\| ds'.  \label{eq:before_bellman}
\end{align}
The following property holds:
\begin{equation*}
\|\bar{x}_i-\hat{\bar{x}}_i\| \le \sum_{j \in \mathcal{N}_i} \|x_j-\hat{x}_j\|, \forall i \in \mathcal{V}, j \in \mathcal{N}_i.
\end{equation*}
Then, by combining the last inequality with \eqref{eq:bound_geitonon} from Property 2, we have that:
\begin{equation*}
\|\bar{x}_i-\hat{\bar{x}}_i\| \le \sum_{j \in \mathcal{N}_i} 2 \sqrt{3} R = 2 \sqrt{3} R N_i, \forall i \in \mathcal{V}, j \in \mathcal{N}_i.
\end{equation*}
By combining the last result with \eqref{eq:before_bellman} we get:
\begin{align}
&\|e_i(s)-\hat{e}_i(s)\| \le  \notag \\
&\int_{t_{\scriptscriptstyle k_z}}^{s} L_i \big\| e_i(s') -\hat{e}_i(s') \big\|ds' + \int_{t_{\scriptscriptstyle k_z}}^{s} \bar{L}_i 2 \sqrt{3} R N_i ds' \notag \\
&= \int_{t_{\scriptscriptstyle k_z}}^{s} L_i \big\| e_i(s') -\hat{e}_i(s') \big\|ds' + 2 \sqrt{3} R \bar{L}_i N_i (s-t_{\scriptscriptstyle k_z}). \label{eq:bbefore_bellman}
\end{align}
By employing the Gronwall-Bellman inequality from \cite{khalil_nonlinear_systems}, \eqref{eq:bbefore_bellman} becomes:
\begin{align}
&\|e_i(s)-\hat{e}_i(s)\| \notag \\
&\le 2 \sqrt{3} R L_i \bar{L}_i N_i \int_{t_{\scriptscriptstyle k_z}}^{s}   (s'-t_{\scriptscriptstyle k_z}) \exp\left[\int_{s'}^{s} L_i ds'' \right] ds'  \notag \\
&\hspace{20mm} + 2 \sqrt{3} R \bar{L}_i N_i (s-t_{\scriptscriptstyle k_z}) \notag \\
&= 2 \sqrt{3} R L_i \bar{L}_i N_i \int_{t_{\scriptscriptstyle k_z}}^{s}   (s'-t_{\scriptscriptstyle k_z}) \exp\left[ L_i(-s'+s) \right] ds'  \notag \\
&\hspace{20mm} + 2 \sqrt{3} R \bar{L}_i N_i (s-t_{\scriptscriptstyle k_z}) \notag \\
&= -2 \sqrt{3} R \bar{L}_i N_i(s-t_{\scriptscriptstyle k_z}) +2 \sqrt{3} R \bar{L}_i N_i (s-t_{\scriptscriptstyle k_z}) \notag \\ 
&\hspace{10mm} + 2 \sqrt{3} R \bar{L}_i N_i\int_{t_{\scriptscriptstyle k_z}}^{s}  \exp\left[ L_i(-s'+s) \right] ds'  \notag \\
&= 2 \sqrt{3} R \bar{L}_i N_i\int_{t_{\scriptscriptstyle k_z}}^{s}  \exp\left[ L_i(-s'+s) \right] ds' \notag \\
&= -\frac{2 \sqrt{3} R \bar{L}_i N_i}{L_i} \left[ 1- e^{L_i(s-t_{\scriptscriptstyle k_z})} \right] \notag \\
&= \frac{2 \sqrt{3} R \bar{L}_i N_i}{L_i} \left[ e^{L_i(s-t_{\scriptscriptstyle k_z})}- 1 \right]. \label{eq:e_e_hat_bound_1}
\end{align}
By employing \eqref{eq:error_bound} of Property \ref{property:error_bound} for the terms $e(s), \hat{e}(s)$ we have that:
\begin{align}
&\|e_i(s)-\hat{e}_i(s)\| \le \|e_i(s)\| + \|\hat{e}_i(s)\| \notag \\
&\le \|e_i(t_{\scriptscriptstyle k_z}; u_i(\cdot), e_i(t_{\scriptscriptstyle k_z}))\|+ (s-t_{\scriptscriptstyle k_z})(M+u_{\max}) + \notag \\ 
&\hspace{15mm} \|\hat{e}_i(t_{\scriptscriptstyle k_z}; u_i(\cdot), e_i(t_{\scriptscriptstyle k_z}))\|+ (s-t_{\scriptscriptstyle k_z})(M+u_{\max}) \notag \\
&\le 2 \|e_i(t_{\scriptscriptstyle k_z}; u_i(\cdot), e_i(t_{\scriptscriptstyle k_z}))\|+ 2 (s-t_{\scriptscriptstyle k_z})(M+u_{\max})\notag \\ 
&= 2 \|e_i(t_{\scriptscriptstyle k_z})\|+ 2 (s-t_{\scriptscriptstyle k_z})(M+u_{\max}). \label{eq:e_e_hat_bound_2}
\end{align}
By combining \eqref{eq:e_e_hat_bound_1}, \eqref{eq:e_e_hat_bound_2} we get:
\begin{align*}
&\|e_i(s)-\hat{e}_i(s)\| \le \min \big\{ \frac{2 \sqrt{3} R \bar{L}_i N_i}{L_i} \left[ e^{L_i(s-t_{\scriptscriptstyle k_z})}- 1 \right],  \notag \\ 
&\hspace{10mm} 2 \|e_i(t_{\scriptscriptstyle k_z})\|+ 2 (s-t_{\scriptscriptstyle k_z})(M+u_{\max}) \big\}.
\end{align*}
which leads to the conclusion of the proof.
\end{proof}

\section{Proof of Property 3} \label{app:property_tighthened_sets}
\begin{proof}
	Let $s \in [t_{\scriptscriptstyle k_z}, t_{\scriptscriptstyle k_z}+T_{\scriptscriptstyle z}]$. Let us also define:
	\begin{equation*}
	z_i(s) \triangleq e_i(s) -\hat{e}_i(s; u_i(s; e(t_{\scriptscriptstyle k_z})), e_i(t_{\scriptscriptstyle k_z})).
	\end{equation*}
	Then, according to Lemma \ref{lemma:e-hat_e}, for $s \in [t_{\scriptscriptstyle k_z}, t_{\scriptscriptstyle k_z}+T_{\scriptscriptstyle z}]$, we get:
	\begin{align*}
	\|z_i(s)\| &= \| e_i(s) -\hat{e}_i(s; u_i(s; e_i(t_{\scriptscriptstyle k_z})), e_i(t_{\scriptscriptstyle k_z})) \| \notag \\
	&\le \rho_i(s-t_{\scriptscriptstyle k_z}).
	\end{align*}
	Hence, $z_i \in B^i_{s-t_{\scriptscriptstyle k_z}}$, which implies that: $-z_i \in B^i_{s-t_{\scriptscriptstyle k_z}}$. The following implications hold:
	\begin{align*}
	\hat{e}_i(s; u_i(s; e_i(t_{\scriptscriptstyle k_z})), e_i(t_{\scriptscriptstyle k_z})) &\in E_i \sim B^i_{s-t_{\scriptscriptstyle k_z}} \notag \\
	\Rightarrow e_i(s) - z_i &\in E \sim B^i_{s-t_{\scriptscriptstyle k_z}} \notag \\
	\Rightarrow e_i(s) +  (-z_i) &\in E \sim B^i_{s-t_{\scriptscriptstyle k_z}} \notag \\
	\Rightarrow e_i(s) &\in E_i, \forall \ s \in [t_{\scriptscriptstyle k_z}, t_{\scriptscriptstyle k_z}+T_{\scriptscriptstyle z}],
	\end{align*}
	which concludes the proof.
\end{proof}

\section{Proof of Lemma 3} \label{app:lemma_lip_L_v}

\begin{proof}
	For every $e_i \in \Phi_i$ we have that:
	\begin{align}
	V_i(e_i) \le \alpha_{1,i} &\Rightarrow e_i^\top P_i e_i \le \alpha_{1,i} \notag \\
	&\Rightarrow \lambda_{\min}(P_i) \|e_i\|^2 \le e_i^\top P_i e_i \le \alpha_{1,i} \notag \\
	&\Rightarrow  \|e_i\| \le \sqrt{ \frac{\alpha_{1,i}}{\lambda_{\min}(P_i)}}. \label{eq:lamda_min_P}
	\end{align}	
	
	\noindent For every $e_1,e_2 \in \Phi_i$, it also holds:
	\begin{align}
	|V_i(e_{1})-V_i(e_{2}) | &= |e_{1}^\top P_i e_{1} - e_{2}^\top P_i e_{2} | \notag \\
	&\hspace{-2mm}= |e_1^\top P_i e_1 +e_1^\top P_i e_2 -e_1^\top P_i e_2 - e_2^\top P_i e_2 | \notag \\
	&\hspace{-2mm}= |e_1^\top P_i (e_1-e_2) - e_2^\top P_i ( e_1 - e_2) | \notag \\
	&\hspace{-2mm}\le |e_1^\top P_i (e_1-e_2)| + |e_2^\top P_i ( e_1 - e_2) |, \notag
	\end{align}
	which by using \eqref{eq:sigma_max} leads to:
	\begin{align*}
	|V_i(e_1)-V_i(e_2) | &\le \sigma_{\max}(P_i) \|e_1\| \|e_1-e_2\| \notag \\ 
	&\hspace{5mm} +\sigma_{\max}(P_i) \|e_2\| \|e_1-e_2\| \notag \\
	&= \sigma_{\max}(P_i) (\|e_1\| + \|e_2\|) \|e_1-e_2\|, \notag
	\end{align*}
	which by employing \eqref{eq:lamda_min_P}, becomes:
	\begin{align*}
	|V_i(e_1)-V_i(e_2) | &\le \sigma_{\max}(P_i) \|e_1\| \|e_1-e_2\| \notag \\ 
	&\hspace{5mm} +\sigma_{\max}(P_i) \|e_2\| \|e_1-e_2\| \notag \\
	&= \sigma_{\max}(P_i) (\|e_1\| + \|e_2\|) \|e_1-e_2\| \notag \\
	&\le 2 \sigma_{\max}(P_i) \sqrt{ \frac{\alpha_1}{\lambda_{\min}(P_i)}} \|e_1-e_2\|, \notag
	\end{align*}
	which completes the proof.
\end{proof}

\section{Proof of Lemma 4} \label{app:y_E_t_i_plus_1}
\begin{proof}
	For every $s \ge t_{\scriptscriptstyle k_{z+1}}, L_g > 0$ the following inequality holds:
	\begin{equation*}
	\left[e^{L_g(t_{\scriptscriptstyle k_{z+1}}-t_{\scriptscriptstyle k_z})}-1\right]+\left[e^{L_g(s-t_{\scriptscriptstyle k_{z+1}})}-1\right] \le \left[e^{L_g(s-t_{\scriptscriptstyle k_z})}-1\right],
	\end{equation*}
	which implies that:
	\begin{align}
	&\widetilde{\rho}_i \left[e^{L_g(t_{\scriptscriptstyle k_{z+1}}-t_{\scriptscriptstyle k_z})}-1\right]+\widetilde{\rho}_i \left[e^{L_g(s-t_{\scriptscriptstyle k_{z+1}})}-1\right] \notag \\ 
	&\hspace{40mm}\le \widetilde{\rho}_i \left[e^{L_g(s-t_{\scriptscriptstyle k_z})}-1\right]. \label{eq:inequaliti_min_rwho_1}
	\end{align}
	It holds also that:
	\begin{align}
	& t_{\scriptscriptstyle k_{z+1}}-t_{\scriptscriptstyle k_z}+s-t_{\scriptscriptstyle k_{z+1}} \le s-t_{\scriptscriptstyle k_z} \notag \\
	\Leftrightarrow \ & 2 \|e_i(t_{\scriptscriptstyle k_z})\|+ 2 (t_{\scriptscriptstyle k_{z+1}}-t_{\scriptscriptstyle k_z})(M+u_{\max}) \notag \\ 
	&+ 2 \|e_i(t_{\scriptscriptstyle k_z})\|+ 2 (s-t_{\scriptscriptstyle k_{z+1}})(M+u_{\max}) \le \notag \\
	& 2 \|e_i(t_{\scriptscriptstyle k_z})\|+ 2 (s-t_{\scriptscriptstyle k_z})(M+u_{\max}). \label{eq:inequaliti_min_rwho_2}
	\end{align}
   By setting:
   \begin{align}
   A_1 &= \widetilde{\rho}_i \left[e^{L_g(t_{\scriptscriptstyle k_{z+1}}-t_{\scriptscriptstyle k_z})}-1\right], \notag \\
   A_2 &= \widetilde{\rho}_i \left[e^{L_g(s-t_{\scriptscriptstyle k_{z+1}})}-1\right], \notag \\
   A_3 &= \widetilde{\rho}_i \left[e^{L_g(s-t_{\scriptscriptstyle k_z})}-1\right], \notag \\
   B_1 &= 2 \|e_i(t_{\scriptscriptstyle k_z})\|+ 2 (t_{\scriptscriptstyle k_{z+1}}-t_{\scriptscriptstyle k_z})(M+u_{\max}), \notag \\
   B_2 &= 2 \|e_i(t_{\scriptscriptstyle k_z})\|+ 2 (s-t_{\scriptscriptstyle k_{z+1}})(M+u_{\max}), \notag \\
   B_3 &= 2 \|e_i(t_{\scriptscriptstyle k_z})\|+ 2 (s-t_{\scriptscriptstyle k_z})(M+u_{\max}), \notag
   \end{align}
   and taking account \eqref{eq:inequaliti_min_rwho_1}, \eqref{eq:inequaliti_min_rwho_2} we get:
   \begin{align}
   &\rho_i(t_{\scriptscriptstyle k_{z+1}}-t_{\scriptscriptstyle k_z})+\rho_i(s-t_{\scriptscriptstyle k_{z+1}}) \notag \\
   &\hspace{15mm} \le \min\{A_1, B_1\} + \min\{B_1, B_2\} \notag \\
   &\hspace{15mm} \le \min\{A_1+A_2, B_1+B_2\} \notag \\
   &\hspace{15mm} \le \min\{A_3, B_3\} \notag \\
      &\hspace{15mm} = \rho_i(s-t_{\scriptscriptstyle k_z}), \notag 
   \end{align}
   or:
	\begin{equation} \label{eq:gamma_inequality}
	\rho_i(t_{\scriptscriptstyle k_{z+1}}-t_{\scriptscriptstyle k_z})+\rho_i(s-t_{\scriptscriptstyle k_{z+1}}) \le \rho_i(s-t_{\scriptscriptstyle k_z}).
	\end{equation}
	Let us consider $\phi \in B^i_{s-t_{\scriptscriptstyle k_{z+1}}}$. Then, it holds $\|\phi\| \le \rho_i(s-t_{\scriptscriptstyle k_{z+1}})$. Let us denote $z = x-y+\phi$. It is clear that:
	\begin{align}
	\|z\| &\le \|x-y+\phi \| \notag \\
	&\le \|x-y\| + \|\phi\| \notag \\
	&\le \rho_i(t_{\scriptscriptstyle k_{z+1}}-t_{\scriptscriptstyle k_z})+\rho_i(s-t_{\scriptscriptstyle k_{z+1}}). \label{eq:z_inequlity}
	\end{align}
	By employing \eqref{eq:gamma_inequality}, \eqref{eq:z_inequlity} becomes:
	\begin{equation*}
	\|z\| \le \rho_i(s-t_{\scriptscriptstyle k_z}),
	\end{equation*}
	which implies that $z \in B^i_{s-t_{\scriptscriptstyle k_z}}$. We have that:
	\begin{align}
	x + (-z) &= y + (-\phi), \notag \\
	x &\in E_{s-t_{\scriptscriptstyle k_z}} = E \sim B_{s-t_{\scriptscriptstyle k_z}}, \notag \\ 
	-z &\in B^i_{s-t_{\scriptscriptstyle k_z}}, \notag \\
	-\rho &\in B^i_{s-t_{\scriptscriptstyle k_{z+1}}}, \notag
	\end{align}
	which implies that $y \in E_{s-t_{\scriptscriptstyle k_{z+1}}} = E \sim B_{s-t_{\scriptscriptstyle k_{z+1}}}$.
\end{proof}

\section{Proof of Lemma 5} \label{app:bounded_trajectories}

\begin{proof}
	Let $s \ge t_{\scriptscriptstyle k_z}$. The following equalities hold:
	\begin{align}
	&\| \hat{e}_i(s; u_i(\cdot), e_i(t_{\scriptscriptstyle k_{z+1}})) - \hat{e}_i(s; u_i(\cdot), e_i(t_{\scriptscriptstyle k_z})) \| \notag \\
	&= \Bigg\| \hat{e}_i(s; u_i(\cdot), e_i(t_{\scriptscriptstyle k_{z+1}})) + \int_{t_{\scriptscriptstyle k_{z+1}}}^{s} g_i(\hat{e}_i(s'), \hat{\bar{x}}_i(s'), u_i(\cdot))ds' \notag \\
	&-\hat{e}_i(t_{\scriptscriptstyle k_{z}}; u_i(\cdot), e_i(t_{\scriptscriptstyle k_z})) - \int_{t_{\scriptscriptstyle k_z}}^{s} g_i(\hat{e}_i(s'), \hat{\bar{x}}_i(s'), u_i(\cdot)))ds \Bigg\| \notag \\
	&=\Bigg\| e_i(t_{\scriptscriptstyle k_{z+1}})  -e_i(t_{\scriptscriptstyle k_z}) - \int_{t_{\scriptscriptstyle k_z}}^{s} g_i(\hat{e}_i(s'), \hat{\bar{x}}_i(s'), u_i(\cdot)) ds'\notag
			\end{align}
			\begin{align}
	&\hspace{30mm}- \int_{s}^{t_{\scriptscriptstyle k_{z+1}}} g_i(\hat{e}_i(s'), \hat{\bar{x}}_i(s'), u_i(\cdot)) ds'\Bigg\| \notag \\
	&= \left\| e_i(t_{\scriptscriptstyle k_{z+1}})  -e_i(t_{\scriptscriptstyle k_z}) - \int_{t_{\scriptscriptstyle k_z}}^{t_{\scriptscriptstyle k_{z+1}}} g_i(\hat{e}_i(s'), \hat{\bar{x}}_i(s'), u_i(\cdot)) ds'\right\| \notag \\
	&= \left\| e_i(t_{\scriptscriptstyle k_{z+1}})  -e_i(t_{\scriptscriptstyle k_z}) - \int_{t_{\scriptscriptstyle k_z}}^{t_{\scriptscriptstyle k_{z+1}}} \frac{d}{dt} \left[ \hat{e}_i(s'; u_i(\cdot), e_i(t_{\scriptscriptstyle k_z})\right] ds'\right\| \notag \\
	&= \big\| e_i(t_{\scriptscriptstyle k_{z+1}})  -e_i(t_{\scriptscriptstyle k_z}) - \hat{e}(t_{\scriptscriptstyle k_{z+1}}; u(\cdot), e_i(t_{\scriptscriptstyle k_z})) \notag \\ &\hspace{45mm} +\hat{e}_i(t_{\scriptscriptstyle k_z}; u_i(\cdot), e_i(t_{\scriptscriptstyle k_z}))\big\| \notag \\
	&= \left\| e_i(t_{\scriptscriptstyle k_{z+1}})  -e_i(t_{\scriptscriptstyle k_z}) - \hat{e}(t_{\scriptscriptstyle k_{z+1}}; u(\cdot), e_i(t_{\scriptscriptstyle k_z}))+e_i(t_{\scriptscriptstyle k_z})\right\| \notag \\
	&= \left\| e_i(t_{\scriptscriptstyle k_{z+1}}) - \hat{e}_i(t_{\scriptscriptstyle k_{z+1}}; u_i(\cdot), e_i(t_{\scriptscriptstyle k_z}))\right\|, \notag
	\end{align}
	which, by employing Lemma \ref{lemma:e-hat_e} for $s = t_{\scriptscriptstyle k_{z+1}}$, becomes:
	\begin{align}
	&\| \hat{e}_i(s; u(\cdot), e_i(t_{\scriptscriptstyle k_{z+1}})) - \hat{e}_i(s; u_i(\cdot), e_i(t_{\scriptscriptstyle k_z})) \| \le \notag \\
	&\hspace{30mm} \rho_i(t_{\scriptscriptstyle k_{z+1}}-t_{\scriptscriptstyle k_z}) = \rho_i(h), \notag
	\end{align}
	since $t_{\scriptscriptstyle k_{z+1}}-t_{\scriptscriptstyle k_z} = h$, which concludes the proof.
\end{proof}

\section{Feasibility and Convergence} \label{app:feasibility_convergence}

\begin{proof}
	The proof consists of two parts: in the first part it is established that initial feasibility implies feasibility afterwards. Based on this result it is then shown that the error $e_i(t)$ converges to the terminal set $\mathcal{E}_i$.
	
	\emph{Feasibility Analysis}: Consider any sampling time instant for which a solution exists, say $t_{\scriptscriptstyle k_z}$. In between $t_{\scriptscriptstyle k_z}$ and $t_{\scriptscriptstyle k_{z+1}}$, the optimal control input $\hat{u}_i^\star (s; e_i(t_{\scriptscriptstyle k_z})), s \in [t_{\scriptscriptstyle k_z}, t_{\scriptscriptstyle k_{z+1}})$ is implemented. The remaining part of the optimal control input $\hat{u}_i^\star (s; e_i(t_{\scriptscriptstyle k_z})), s \in [t_{\scriptscriptstyle k_{z+1}}, t_{\scriptscriptstyle k_z}+T_{\scriptscriptstyle z}],$ satisfies the state and input constraints $E_i, \mathcal{U}_i$, respectively. Furthermore, since the problem is feasible at time $t_{\scriptscriptstyle k_z}$, it holds that:
	\begin{subequations}
		\begin{align}
		\hat{e}_i(s; \hat{u}^\star (s; e_i(t_{\scriptscriptstyle k_z})), e_i(t_{\scriptscriptstyle k_z})) &\in E^i_{s-t_{\scriptscriptstyle k_z}}, \label{eq:e_s_hat_in_E_s_t_i}\\
		\hat{e}_i(t_{\scriptscriptstyle k_z}+T; \hat{u}_i^\star (s; e_i(t_{\scriptscriptstyle k_z})), e_i(t_{\scriptscriptstyle k_z})) &\in \mathcal{E}_i, \label{eg:t_k_z_plus_T_in_E_i}
		\end{align}
	\end{subequations}
	for $s \in [t_{\scriptscriptstyle k_z}, t_{\scriptscriptstyle k_z}+T_{\scriptscriptstyle z}]$. By using Property 1, \eqref{eq:e_s_hat_in_E_s_t_i} implies also that $e_i(s; \hat{u}_i^\star (s; e_i(t_{\scriptscriptstyle k_z})), e_i(t_{\scriptscriptstyle k_z})) \in E_i$. We know also from Assumption \ref{ass:admissible_u_f} that for all $e_i \in \mathcal{E}_i$, there exists at least one control input $u_{f,i}(\cdot)$ that renders the set $\mathcal{E}_i$ invariant over $h$. Picking any such input, a feasible control input $\bar{u}_i(\cdot; e_i(t_{\scriptscriptstyle k_{z+1}}))$, at time instant $t_{\scriptscriptstyle k_{z+1}}$, may be the following:
	\begin{align} \label{eq:u_bar_feas}
	&\bar{u}_i(s; e(t_{\scriptscriptstyle k_{z+1}})) = \notag \\
	&\hspace{-2mm} \begin{cases}
	\hat{u}_i^\star (s; e_i(t_{\scriptscriptstyle k_z})), \ & s \in [t_{\scriptscriptstyle k_{z+1}}, t_{\scriptscriptstyle k_{z}}+ T_{\scriptscriptstyle z+1}], \\
	u_{f,i} (t_{\scriptscriptstyle k_{z}}+ T_{\scriptscriptstyle z+1}; \hat{u}^\star(\cdot), e(t_i))), & s \in [t_{\scriptscriptstyle k_{z}}+ T_{\scriptscriptstyle z+1}, t_{\scriptscriptstyle k_{z}}+ T_{\scriptscriptstyle z}]. 
	\end{cases}
	\end{align}
	For the time intervals it holds that (see Fig. \ref{fig:time_sequence}):
	\begin{align}
	&t_{\scriptscriptstyle k_{z}}+ T_{\scriptscriptstyle z+1} = t_{\scriptscriptstyle k_{z}} + T_{\scriptscriptstyle z} - h = t_{\scriptscriptstyle k_{z}} + T - h.\notag
	\end{align}
    For the feasibility of the ROCP, we have to prove the following three statements for every $s \in [t_{\scriptscriptstyle k_{z+1}}, t_{\scriptscriptstyle k_{z}}+T_{\scriptscriptstyle z}]$:
	\begin{enumerate}
		\item $\bar{u}_i(s; e(t_{\scriptscriptstyle k_{z+1}})) \in \mathcal{U}_i$.
		\item $\hat{e}_i(t_{\scriptscriptstyle k_{z}}+T_{\scriptscriptstyle z}; \bar{u}(s; e(t_{\scriptscriptstyle k_{z+1}})), e(t_{\scriptscriptstyle k_{z+1}})) \in \mathcal{E}_i$.
		\item $\hat{e}_i(s; \bar{u}_i(s; e(t_{\scriptscriptstyle k_{z+1}})), e(t_{\scriptscriptstyle k_{z+1}})) \in E^i_{s-t_{\scriptscriptstyle k_{z+1}}}$.
	\end{enumerate}
	Statement 1: From the feasibility of $\hat{u}_i^\star(s, e(t_{\scriptscriptstyle k_{z}}))$ and the fact that $u_{f,i}(e_i(\cdot)) \in \mathcal{U}_i$, for all $e_i(\cdot) \in \Phi_i$, it follows that:
	\begin{equation*}
	\bar{u}_i(s; e_i(t_{\scriptscriptstyle k_{z+1}})) \in \mathcal{U}_i, \forall \ s \in [t_{\scriptscriptstyle k_{z+1}}, t_{\scriptscriptstyle k_{z}}+T_{\scriptscriptstyle z}].
	\end{equation*}
	Statement 2: We need to prove in this step that for every $s \in [t_{\scriptscriptstyle k_{z+1}}, t_{\scriptscriptstyle k_{z}}+T_{\scriptscriptstyle z}]$ it holds that $\hat{e}_i(t_{\scriptscriptstyle k_{z}}+T_{\scriptscriptstyle z}; \bar{u}_i(s; e_i(t_{\scriptscriptstyle k_{z+1}}))), e_i(t_{\scriptscriptstyle k_{z+1}})) \in \mathcal{E}_i$.	Since $V_i(\cdot)$ is Lipschitz continuous, we get:
	\begin{align}
	&\hspace{-2mm} V_i(\hat{e}_i(t_{\scriptscriptstyle k_{z}}+T_{\scriptscriptstyle z+1}; \bar{u}_i(\cdot), e_i(t_{\scriptscriptstyle k_{z+1}}))) - \notag \\ 
	&\hspace{20mm} V_i(\hat{e}_i(t_{\scriptscriptstyle k_{z}}+T_{\scriptscriptstyle z+1}; \bar{u}_i(\cdot), e_i(t_{\scriptscriptstyle k_{z}}))) \le  \notag \\
	&\hspace{-2mm} L_{\scriptscriptstyle V_i} \| \hat{e}_i(t_{\scriptscriptstyle k_{z}}+T_{\scriptscriptstyle z+1}; \bar{u}_i(\cdot), e_i(t_{\scriptscriptstyle k_{z+1}})) \notag \\ 
	&\hspace{20mm}- \hat{e}(t_{\scriptscriptstyle k_{z}}+T_{\scriptscriptstyle z+1}; \bar{u}_i(\cdot), e(t_{\scriptscriptstyle k_{z}})) \|. \label{eq:feasib_lipsitch}
	\end{align}
	for the same control input $\bar{u}_i(\cdot) = u_i^\star(s; e_i(t_{\scriptscriptstyle k_{z}}))$. By employing Lemma \ref{lemma:bounded_trajectories} for $\alpha = t_{\scriptscriptstyle k_{z}}+T_{\scriptscriptstyle z+1}$ and $u(\cdot) = \bar{u}_i(\cdot) = u_i^\star(s; e_i(t_{\scriptscriptstyle k_{z}}))$, we have that:
	\begin{align}
    &\hspace{-2mm}\| \hat{e}_i(t_{\scriptscriptstyle k_{z}}+T_{\scriptscriptstyle z}; \bar{u}_i(\cdot), 
    e_i(t_{\scriptscriptstyle k_{z+1}})) \notag \\ 
    &\hspace{-1mm}- \hat{e}_i(t_{\scriptscriptstyle k_{z}}+T_{\scriptscriptstyle z+1}; 
    \bar{u}_i(\cdot), e(t_{\scriptscriptstyle k_{z}})) \| \le \rho_i(t_{\scriptscriptstyle k_{z+1}}-t_{\scriptscriptstyle k_{z}}) = \rho_i(h). \label{eq:bound_gamma_h}
	\end{align}
	Note also that for the function $\rho_i(\cdot)$ the following implication holds: $$h \le T_{\scriptscriptstyle z} \Rightarrow \rho_i(h) \le \rho_i(T_{\scriptscriptstyle z}).$$ By employing the latter result, \eqref{eq:bound_gamma_h} becomes:
	\begin{align}
    &\hspace{-2mm}\| \hat{e}_i(t_{\scriptscriptstyle k_{z}}+T_{\scriptscriptstyle z}; \bar{u}_i(\cdot), 
    e_i(t_{i+1})) \notag \\ 
    &\hspace{-1mm}- \hat{e}_i(t_{\scriptscriptstyle k_{z}}+T_{\scriptscriptstyle z+1}; 
    \bar{u}_i(\cdot), e(t_{\scriptscriptstyle k_{z}})) \| \le \rho_i(h) \le \rho_i(T_{\scriptscriptstyle z}). \label{eq:bound_gamma_h2}
	\end{align}
	By combining \eqref{eq:bound_gamma_h2} and \eqref{eq:feasib_lipsitch} we get:
	\begin{align}
	&\hspace{-2mm} V_i(\hat{e}_i(t_{\scriptscriptstyle k_{z}}+T_{\scriptscriptstyle z+1}; \bar{u}_i(\cdot), e_i(t_{\scriptscriptstyle k_{z+1}}))) - \notag \\ 
		&\hspace{10mm} V_i(\hat{e}_i(t_{\scriptscriptstyle k_{z}}+T_{\scriptscriptstyle z+1}; \bar{u}_i(\cdot), e_i(t_{\scriptscriptstyle k_{z}}))) \le L_{\scriptscriptstyle V_i} \rho_i(T_{\scriptscriptstyle z}), \notag
	\end{align}
	or equivalently:
	\begin{align}
	&\hspace{-2mm} V_i(\hat{e}_i(t_{\scriptscriptstyle k_{z}}+T_{\scriptscriptstyle z+1}; \bar{u}_i(\cdot), e_i(t_{\scriptscriptstyle k_{z+1}}))) \le \notag \\ 
	&\hspace{10mm} V_i(\hat{e}_i(t_{\scriptscriptstyle k_{z}}+T_{\scriptscriptstyle z+1}; \bar{u}_i(\cdot), e_i(t_{\scriptscriptstyle k_{z}}))) + L_{\scriptscriptstyle V_i} \rho_i(T_{\scriptscriptstyle z}). \label{eq:feasib_lipsitch_1}
	\end{align}
	By using \eqref{eg:t_k_z_plus_T_in_E_i}, we have that $\hat{e}_i(t_{\scriptscriptstyle k_{z}}+T_{\scriptscriptstyle z+1}; \bar{u}_i(\cdot), e_i(t_{\scriptscriptstyle k_{z}})) \in \mathcal{E}_i$. Then, \eqref{eq:feasib_lipsitch_1} gives:
	\begin{align}
	V_i(\hat{e}_i(t_{\scriptscriptstyle k_{z}}+T_{\scriptscriptstyle z+1}; \bar{u}_i(\cdot), e_i(t_{\scriptscriptstyle k_{z+1}}))) \le \alpha_{2,i} +  L_{\scriptscriptstyle V_i} \rho_i(T_{\scriptscriptstyle z}) \label{eq:feasib_lipsitch_2}
	\end{align}
	From \eqref{eq:rho_bar} of the Theorem 1, we get equivalently:
	\begin{align}
	&  \rho_i(T_{\scriptscriptstyle z}) \le \frac{\alpha_{1,1}-\alpha_{2,i}}{L_{\scriptscriptstyle V_i}} \notag \\ 
	\Leftrightarrow  &   \alpha_{2,i} +  L_{\scriptscriptstyle V_i} \rho_i(T_{\scriptscriptstyle z}) \le  \alpha_{1,i}. \label{eq:alpha_1_L_V}
	\end{align}
	By combining \eqref{eq:feasib_lipsitch_2} and \eqref{eq:alpha_1_L_V}, we get:
	\begin{align*}
	V_i(\hat{e}_i(t_{\scriptscriptstyle k_{z}}+T_{\scriptscriptstyle z+1}; \bar{u}_i(\cdot), e_i(t_{\scriptscriptstyle k_{z+1}}))) \le \alpha_{1,i},
	\end{align*}
	which, from Assumption \ref{ass:admissible_u_f}, implies that:
	\begin{equation} \label{eq:e_hat_in_phi}
	\hat{e}_i(t_{\scriptscriptstyle k_{z}}+T_{\scriptscriptstyle z+1}; \bar{u}_i(\cdot), e_i(t_{\scriptscriptstyle k_{z+1}})) \in \Phi_i.
	\end{equation}
	But since $\bar{u}_i(\cdot)$ is chosen to be local admissible controller from Assumption \ref{ass:admissible_u_f}, according to our choice of terminal set $\mathcal{E}_i$,  \eqref{eq:e_hat_in_phi} leads to:
	\begin{equation*}
	\hat{e}_i(t_{\scriptscriptstyle k_{z}}+T_{\scriptscriptstyle z}; \bar{u}_i(\cdot), e_i(t_{\scriptscriptstyle k_{z+1}})) \in \mathcal{E}_i.
	\end{equation*}
	Thus, statement 2 holds.
	
	\noindent Statement 3: By employing Lemma \ref{lemma:bounded_trajectories} for:
		\begin{align}
		x &= \hat{e}_i(s; \bar{u}_i(s; e(t_{\scriptscriptstyle k_{z}})), e(t_{\scriptscriptstyle k_{z}})) \in E^i_{s-t_i}, \notag \\
		y &= \hat{e}_i(s; \bar{u}_i(s; e(t_{\scriptscriptstyle k_{z+1}})), e(t_{\scriptscriptstyle k_{z+1}})), \notag
		\end{align}
	we get that:
	\begin{align*}
	&\|y-x\| = \|\hat{e}(s; \bar{u}(s; e(t_{i+1})), e(t_{i+1})) \notag \\
	&\hspace{10mm} - \hat{e}(s; \bar{u}(s; e(t_{i})), e(t_{i})) \in E_{s-t_i} \| \le \rho_i(h).
	\end{align*}
	Furthermore, by employing Lemma \ref{lemma:x_y_proof} for $s \in [t_{\scriptscriptstyle k_{z+1}}, t_{\scriptscriptstyle k_{z}}+T_{\scriptscriptstyle z}]$ and the same $x, y$ as previously
	we get that $ y = \hat{e}_i(s; \bar{u}(s; e_i(t_{\scriptscriptstyle k_{z+1}})), e(t_{\scriptscriptstyle k_{z+1}})) \in E^i_{s-t_{\scriptscriptstyle k_{z+1}}}$, which according to Property 1, implies that $e_i(s; \bar{u}_i(s; e_i(t_{\scriptscriptstyle k_{z+1}})), e_i(t_{\scriptscriptstyle k_{z+1}})) \in E_i$.  Thus, Statement 3 holds. Hence, the feasibility at time $t_{\scriptscriptstyle k_{z}}$ implies feasibility at time $t_{\scriptscriptstyle k_{z+1}}$. Therefore, if the ROCP \eqref{eq:mpc_minimazation} - \eqref{eq:mpc_terminal_set} is feasible at time $t_{\scriptscriptstyle k_{z}}$, i.e., it remains feasible for every $t \in [t_k, t_k + T]$.
	
	\emph{Convergence Analysis}: The second part involves proving convergence of the state $e_i$ to the terminal set $\mathcal{E}_i$. In order to prove this, it must be shown that a proper value function is decreasing along the solution trajectories starting at a sampling time $t_i$. Consider the optimal value function $J_i^\star(e_i(t_{\scriptscriptstyle k_z}))$, as is given in \eqref{eq:J_star}, to be a Lyapunov-like function. Consider also the cost of the feasible control input, indicated by:
	\begin{equation} \label{eq:J_bar}
	\bar{J}_i(e_i(t_{\scriptscriptstyle k_{z+1}})) \triangleq \bar{J}_i(e_i(t_{\scriptscriptstyle k_{z+1}}), \bar{u}_i(\cdot; e_i(t_{\scriptscriptstyle k_{z+1}}))),
	\end{equation}
	where $t_{\scriptscriptstyle k_{z+1}} = t_{\scriptscriptstyle k_z} + h$. Define:
	\begin{subequations}
	\begin{align}
	&\hspace{-4mm} \bar{u}_1(s) \triangleq \bar{u}_i(s; e_i(t_{\scriptscriptstyle k_{z+1}})), \label{eq:u_1} \\
	&\hspace{-3mm} \bar{e}_1(s)  \triangleq \bar{e}_i(s; u_1(s), e_i(t_{\scriptscriptstyle k_{z+1}})), s \in [t_{\scriptscriptstyle k_{z+1}}, t_{\scriptscriptstyle k_{z}}+T], \label{eq:e_1}
	\end{align}
	\end{subequations}
	where $\bar{e}_1(s)$ stands for the predicted state $e_i$ at time $s$, based on the measurement of the state $e_i$ at time $t_{\scriptscriptstyle k_{z+1}}$, while using the feasible control input $\bar{u}_i(s; e(t_{\scriptscriptstyle k_{z+1}}))$ from \eqref{eq:u_bar_feas}.
	Let us also define the following terms:
	\begin{align}
	\hat{u}_2(s) &\triangleq \hat{u}_i^\star(s; e_i(t_{\scriptscriptstyle k_z})), \label{eq:u_2}\\
	\hat{e}_2(s)  &\triangleq \hat{e}_i(s; \hat{u}_2(s), e_i(t_{\scriptscriptstyle k_z})), s \in [t_{\scriptscriptstyle k_{z}}, t_{\scriptscriptstyle k_{z}}+T-h]. \notag
	\end{align}
	where $\hat{e}_1(s)$ stands for the predicted state $e_i$ at time $s$, based on the measurement of the state $e_i$ at time $t_{\scriptscriptstyle k_{z}}$, while using the control input $\hat{u}_i(s; e(t_{\scriptscriptstyle k_{z}})), s \in [t_{\scriptscriptstyle k_{z}}, t_{\scriptscriptstyle k_{z}}+T-h]$ from \eqref{eq:u_bar_feas}.
	By employing \eqref{eq:mpc_minimazation}, \eqref{eq:J_star} and \eqref{eq:J_bar}, the difference between the optimal and feasible cost is given by:
	\begin{align}
	&\bar{J}(e_i(t_{\scriptscriptstyle k_{z+1}})) - J^\star(e_i(t_{\scriptscriptstyle k_z})) =  \notag \\
	& V_i(\bar{e}_1(t_{\scriptscriptstyle k_{z}}+T) + \int_{t_{\scriptscriptstyle k_{z+1}}}^{t_{\scriptscriptstyle k_{z}}+T} \Big[ F_i(\bar{e}_1(s), \bar{u}_1(s)) \Big] ds \notag \\
	&\hspace{-1mm}-V_i(\hat{e}_2(t_{\scriptscriptstyle k_z}+T-h) - \int_{t_{\scriptscriptstyle k_z}}^{t_{\scriptscriptstyle k_z}+T-h} \Big[ F_i(\hat{e}_2(s), \hat{u}_2(s)) \Big] ds \notag \\
	&=V_i(\bar{e}_1(t_{\scriptscriptstyle k_{z}}+T)) + \int_{t_{\scriptscriptstyle k_{z+1}}}^{t_{\scriptscriptstyle k_{z+1}}+T-h} \Big[ F_i(\bar{e}_1(s), \bar{u}_1(s)) \Big] ds \notag \\ 
	&+\int_{t_{\scriptscriptstyle k_z}+T-h}^{t_{\scriptscriptstyle k_z}+T} \Big[ F_i(\bar{e}_1(s), \bar{u}_1(s)) \Big] ds -V_i(\hat{e}_2(t_{\scriptscriptstyle k_z}+T - h)) \notag \\
	&\hspace{13mm}-\int_{t_{\scriptscriptstyle k_z}}^{t_{\scriptscriptstyle k_{z+1}}} \Big[ F_i(\hat{e}_2(s), \hat{u}_2(s)) \Big] ds \notag \\
	&\hspace{23mm}-\int_{t_{\scriptscriptstyle k_{z+1}}}^{t_{\scriptscriptstyle k_z}+T - h} \Big[ F_i(\hat{e}_2(s), \hat{u}_2(s)) \Big] ds. \label{eq:lyap1}
	\end{align}
	Note that, from \eqref{eq:u_bar_feas}, the following holds:
	\begin{equation} \label{eq:verify_u_bar}
	\bar{u}_i(s; e_i(t_{\scriptscriptstyle k_{z+1}})) = \hat{u}_i^\star(s; e_i(t_{\scriptscriptstyle k_z})), \forall \ s \in [t_{\scriptscriptstyle k_{z+1}}, t_{\scriptscriptstyle k_z}+T - h].
	\end{equation}
	By combining \eqref{eq:u_1}, \eqref{eq:u_2} and \eqref{eq:verify_u_bar}, we have that:
	\begin{align}
	&\bar{u}_1(s) = \hat{u}_2(s) = \bar{u}_i(s; e_i(t_{\scriptscriptstyle k_{z+1}})) = \hat{u}_i^\star(s; e_i(t_{\scriptscriptstyle k_z})), \notag \\ 
	&\hspace{25mm}\forall \ s \in [t_{\scriptscriptstyle k_{z+1}}, t_{\scriptscriptstyle k_z}+T - h], \label{eq:controllers_equality_convergence}
	\end{align}
	By applying the last result and the fact that $F_i(e,u)$ is Lipschitz, the following holds:
	\begin{align}
	&\int_{t_{\scriptscriptstyle k_{z+1}}}^{t_{\scriptscriptstyle k_z}+T - h} \Big[ F_i(\bar{e}_1(s), \bar{u}_1(s)) \Big] ds \notag \\ 
	&\hspace{30mm}-\int_{t_{\scriptscriptstyle k_{z+1}}}^{t_{\scriptscriptstyle k_z}+T - h} \Big[ F_i(\hat{e}_2(s), \hat{u}_2(s)) \Big] ds \notag \\
	&=\int_{t_{\scriptscriptstyle k_{z+1}}}^{t_{\scriptscriptstyle k_z}+T - h} \Big[ F_i(\bar{e}_1(s), \bar{u}_1(s)) - F_i(\hat{e}_2(s), \hat{u}_2(s)) \Big] ds \notag
	\end{align}
	\begin{align}
	&= \notag \int_{t_{\scriptscriptstyle k_{z+1}}}^{t_{\scriptscriptstyle k_z}+T - h} \Big[ F_i(\bar{e}_1(s), \bar{u}_i(s; e_i(t_{\scriptscriptstyle k_{z+1}}))) \notag \\
	&\hspace{30mm}-F_i(\hat{e}_2(s), \bar{u}_i(s; e_i(t_{\scriptscriptstyle k_{z+1}}))) \Big] ds \notag \\
	&\le \Bigg\| \int_{t_{\scriptscriptstyle k_{z+1}}}^{t_{\scriptscriptstyle k_z}+T - h} \Big[ F_i(\bar{e}_1(s), \bar{u}_i(\cdot)) - F_i(\hat{e}_2(s), \bar{u}_i(\cdot)) \Big] ds \Bigg\| \notag \\
	&\le \int_{t_{\scriptscriptstyle k_{z+1}}}^{t_{\scriptscriptstyle k_z}+T - h} \Big\| F_i(\bar{e}_1(s), \bar{u}_i(\cdot)) - F_i(\hat{e}_2(s), \bar{u}_i(\cdot)) \Big\| ds \notag \\
	&\le L_{\scriptscriptstyle F_i} \int_{t_{\scriptscriptstyle k_{z+1}}}^{t_{\scriptscriptstyle k_z}+T - h} \left\| \bar{e}_1(s) - \hat{e}_2(s) \right\| ds. \label{eq:F_1_F_2}
	\end{align}
	By employing the fact that $\forall s \in [t_{\scriptscriptstyle k_{z+1}}, t_{\scriptscriptstyle k_z}+T - h]$ the following holds:
	\begin{equation} \label{eq:bar_e_equals_hat_e}
	\bar{e}_i(s; \bar{u}_i(\cdot), e_i(t_{\scriptscriptstyle k_{z+1}})) = \hat{e}_i(s; \bar{u}_i(\cdot), e_i(t_{\scriptscriptstyle k_z})), 
	\end{equation}
	the term $\left\| \bar{e}_1(s) - \hat{e}_2(s) \right\|$ can be written as:
	\begin{align}
	&\left\| \bar{e}_1(s) - \hat{e}_2(s) \right\| \notag \\
	&= \|\bar{e}_i(s; \bar{u}_i(\cdot), e_i(t_{\scriptscriptstyle k_{z+1}})) - \hat{e}_i(s; \hat{u}_i(\cdot), e_i(t_{\scriptscriptstyle k_z})) \| \notag \\
	&= \Bigg\| \bar{e}_i(t_{\scriptscriptstyle k_{z+1}}; \bar{u}_i(\cdot), e_i(t_{\scriptscriptstyle k_{z+1}})) \notag \\ 
	&+ \int_{t_{\scriptscriptstyle k_{z+1}}}^{s} g_i(\bar{e}_i(s'), \hat{\bar{x}}_i(s') , \bar{u}_i(\cdot))ds' \notag \\
	&- \hat{e}_i(t_{\scriptscriptstyle k_z}; \hat{u}_i(\cdot), e_i(t_{\scriptscriptstyle k_z})) - \int_{t_{\scriptscriptstyle k_z}}^{t_{\scriptscriptstyle k_{z+1}}} g_i(\hat{e}_i(s),  \hat{\bar{x}}_i(s'), \hat{u}_i(\cdot))ds \notag \\
	&\hspace{25mm}-\int_{t_{\scriptscriptstyle k_{z+1}}}^{s} g_i(\hat{e}_i(s'), \hat{\bar{x}}_i(s'), \bar{u}_i(\cdot))ds' \Bigg\| \notag \\
	&\le \Bigg\| e_i(t_{\scriptscriptstyle k_{z+1}}) -e_i(t_{\scriptscriptstyle k_z}) - \int_{t_{\scriptscriptstyle k_z}}^{t_{\scriptscriptstyle k_{z+1}}} g_i(\hat{e}_i(s'), \hat{\bar{x}}_i(s'), \hat{u}_i(\cdot))ds' \Bigg\| \notag \\ 
	&+\Bigg\| \int_{t_{\scriptscriptstyle k_{z+1}}}^{s} g_i(\bar{e}_i(s'), \hat{\bar{x}}_i(\cdot), \bar{u}_i(\cdot))ds' \notag \\ 
	&\hspace{20mm} -\int_{t_{\scriptscriptstyle k_{z+1}}}^{s} g_i(\hat{e}_i(s'), \hat{\bar{x}}_i(\cdot), \bar{u}_i(\cdot))ds' \Bigg\| \notag \\
	&= \Bigg\| e_i(t_{\scriptscriptstyle k_{z+1}}) -e_i(t_{\scriptscriptstyle k_z}) - \int_{t_{\scriptscriptstyle k_z}}^{t_{\scriptscriptstyle k_{z+1}}} \frac{d}{dt} \left[ \hat{e}_i(s; \hat{u}_i(\cdot), e_i(t_{\scriptscriptstyle k_z}))\right] ds \Bigg\| \notag \\
	&+\Bigg\| \int_{t_{\scriptscriptstyle k_{z+1}}}^{s} \frac{d}{dt} \left[ \bar{e}_i(s'; \bar{u}_i(\cdot), e_i(t_{\scriptscriptstyle k_{z+1}}))\right]ds' \notag \\ 
	&\hspace{25mm} -\int_{t_{\scriptscriptstyle k_{z+1}}}^{s} \frac{d}{dt} \left[ \hat{e}_i(s'; \bar{u}_i(\cdot), e_i(t_{\scriptscriptstyle k_z}))\right]ds' \Bigg\| \notag \\
	&= \big\| e_i(t_{\scriptscriptstyle k_{z+1}}) -e_i(t_{\scriptscriptstyle k_z}) -  \hat{e}_i(t_{\scriptscriptstyle k_{z+1}}; \hat{u}_i(\cdot), e_i(t_{\scriptscriptstyle k_z}))  \notag \\
	&\hspace{40mm}+\hat{e}_i(t_{\scriptscriptstyle k_z}; \hat{u}_i(\cdot), e_i(t_{\scriptscriptstyle k_z})) \big\| \notag \\ 
	&+\big\| \bar{e}_i(s; \bar{u}_i(\cdot), e_i(t_{\scriptscriptstyle k_{z+1}}))-\bar{e}_i(t_{i+1}; \bar{u}_i(\cdot), e_i(t_{\scriptscriptstyle k_{z+1}})) \notag \\
	&-  \hat{e}_i(s; \bar{u}_i(\cdot), e_i(t_{\scriptscriptstyle k_z}))+\hat{e}_i(t_{\scriptscriptstyle k_{z+1}}; \bar{u}_i(\cdot), e_i(t_{\scriptscriptstyle k_z})) \big\| \notag \\
	&= \left\| e_i(t_{\scriptscriptstyle k_{z+1}}) -e_i(t_{\scriptscriptstyle k_z}) -  \hat{e}_i(t_{\scriptscriptstyle k_{z+1}}; \hat{u}_i(\cdot), e_i(t_{\scriptscriptstyle k_z})) +e_i(t_{\scriptscriptstyle k_z}) \right\| \notag \\ 
	&= \left\| e_i(t_{\scriptscriptstyle k_{z+1}})-  \hat{e}_i(t_{\scriptscriptstyle k_{z+1}}; \hat{u}_i(\cdot), e_i(t_{\scriptscriptstyle k_z}))\right\|, \notag
	\end{align}
	which, by employing Lemma \ref{lemma:e-hat_e}, leads to:
	\begin{equation*}
	\left\| \bar{e}_1(s) - \hat{e}_2(s) \right\| \le \rho_i(t_{\scriptscriptstyle k_{z+1}}-t_{\scriptscriptstyle k_z}) = \rho_i(h).
	\end{equation*}
	By combining the last result with \eqref{eq:F_1_F_2} we get:
	\begin{align}
	&\int_{t_{\scriptscriptstyle k_{z+1}}}^{t_{\scriptscriptstyle k_z}+T - h} \Big[ F_i(\bar{e}_1(s), \bar{u}_1(s)) \Big] ds \notag \\ 
	&\hspace{27mm}-\int_{t_{\scriptscriptstyle k_{z+1}}}^{t_{\scriptscriptstyle k_z}+T - h} \Big[ F_i(\hat{e}_2(s), \hat{u}_2(s)) \Big] ds \notag \\
	&\le L_{\scriptscriptstyle F_i} \int_{t_{\scriptscriptstyle k_{z+1}}}^{t_{\scriptscriptstyle k_z}+T - h} \rho_i(h) ds =  (T - 2h) \rho_i(h) L_{\scriptscriptstyle F_i} . \label{eq:F_1_F_23}
	\end{align}
	By combining the last result with \eqref{eq:F_1_F_2}, \eqref{eq:lyap1} becomes:
	\begin{align}
	&\bar{J}(e_i(t_{\scriptscriptstyle k_{z+1}})) - J^\star(e_i(t_{\scriptscriptstyle k_z})) \le (T -2 h) \rho_i(h) L_{\scriptscriptstyle F_i} \notag \\
	&\hspace{13mm}+V_i(\bar{e}_1(t_{\scriptscriptstyle k_{z}}+T)) - V_i(\hat{e}_2(t_{\scriptscriptstyle k_z}+T-h))\notag \\ 
	&\hspace{20mm}+\int_{t_{\scriptscriptstyle k_z}+T-h}^{t_{\scriptscriptstyle k_{z}}+T} \Big[ F_i(\bar{e}_1(s), \bar{u}_1(s)) \Big] ds \notag \\ 
	&\hspace{30mm}-\int_{t_{\scriptscriptstyle k_z}}^{t_{\scriptscriptstyle k_{z+1}}} \Big[ F_i(\hat{e}_2(s), \hat{u}_2(s)) \Big] ds \label{eq:lyap_2}
	\end{align}
	By integrating inequality \eqref{eq:admiss_controller} from $t_{\scriptscriptstyle k_z}+T - h$ to $t_{\scriptscriptstyle k_{z}}+T$ and we get the following:
	\begin{align}
	&\hspace{-3mm} \int_{t_{\scriptscriptstyle k_z}+T - h}^{t_{\scriptscriptstyle k_{z}}+ T} \Big[ \frac{\partial V}{\partial{e}} \cdot g_i(\bar{e}_1(s), \hat{\bar{x}}_i(s), \bar{u}_1(s)) \notag \\ 
	&\hspace{30mm}+ F_i(\bar{e}_1(s), \bar{u}_1(s)) \Big] ds \le 0 \notag \\ 
	\Leftrightarrow & V_i(\bar{e}_1(t_{\scriptscriptstyle k_{z}}+T) - V_i(\bar{e}_1(t_{\scriptscriptstyle k_z}+T-h))   \notag \\
	&\hspace{22mm} + \int_{t_{\scriptscriptstyle k_z}+ T - h}^{t_{\scriptscriptstyle k_{z}}+ T} \Big[F_i(\bar{e}_1(s), \bar{u}_1(s)) \Big] ds \le 0, \notag
	\end{align}
	which by adding and subtracting the term $V_i(\hat{e}_2(t_{\scriptscriptstyle k_z}+T-h))$ becomes:
		\begin{align} 
		& V_i(\bar{e}_1(t_{\scriptscriptstyle k_{z}}+T) - V_i(\hat{e}_2(t_{\scriptscriptstyle k_z}+T-h)) \notag \\
		&\hspace{10mm} + \int_{t_{\scriptscriptstyle k_z}+ T - h}^{t_{\scriptscriptstyle k_{z}}+T} \Big[F_i(\bar{e}_1(s), \bar{u}_1(s)) \Big] ds  \le \notag \\ 
		&\hspace{15mm} V_i(\bar{e}_1(t_{\scriptscriptstyle k_z}+T-h))-V_i(\hat{e}_2(t_{\scriptscriptstyle k_z}+T-h)). \notag 
		\end{align}
	By employing the property $y \le |y|, \forall y \in \mathbb{R}$, we get:
	\begin{align}
		& V_i(\bar{e}_1(t_{\scriptscriptstyle k_{z}}+T) - V_i(\hat{e}_2(t_{\scriptscriptstyle k_z}+T-h)) \notag \\
		&\hspace{4mm} + \int_{t_{\scriptscriptstyle k_z}+ T - h}^{t_{\scriptscriptstyle k_{z}}+T} \Big[F_i(\bar{e}_1(s), \bar{u}_1(s)) \Big] ds  \le \notag \\ 
		&\hspace{4mm} \left|V_i(\bar{e}_1(t_{\scriptscriptstyle k_z}+T-h)) - V_i(\hat{e}_2(t_{\scriptscriptstyle k_z}+T-h)) \right|. \label{eq:V_e_1-V_e_2}
	\end{align}
	By employing Lemma \ref{eq:lemma_L_v}, we have that:
	\begin{align*}
	&\left|V_i(\bar{e}_1(t_{\scriptscriptstyle k_z}+T-h)) - V_i(\hat{e}_2(t_{\scriptscriptstyle k_z}+T-h)) \right| \le \notag \\ 
	&\hspace{18mm} L_{\scriptscriptstyle V_i} \|\bar{e}_1(t_{\scriptscriptstyle k_z}+T-h) - \hat{e}_2(t_{\scriptscriptstyle k_z}+T-h) \|, 
	\end{align*}
	which by employing Lemma \ref{lemma:bounded_trajectories} and \eqref{eq:controllers_equality_convergence}, becomes:
	\begin{align*}
	&\left|V_i(\bar{e}_1(t_{\scriptscriptstyle k_z}+T-h)) - V_i(\hat{e}_2(t_{\scriptscriptstyle k_z}+T-h)) \right| \le \notag \\ 
	&\hspace{18mm} L_{\scriptscriptstyle V_i} \rho_i(t_{\scriptscriptstyle k_{z+1}}-t_{\scriptscriptstyle k_z}) =  \rho_i(h) L_{\scriptscriptstyle V_i}.
	\end{align*}
	By combining the last result with \eqref{eq:V_e_1-V_e_2}, we get:
	\begin{align}
	& V_i(\bar{e}_1(t_{\scriptscriptstyle k_{z}}+T) - V_i(\hat{e}_2(t_{\scriptscriptstyle k_z}+T-h)) \notag \\
	&\hspace{4mm} + \int_{t_{\scriptscriptstyle k_z}+ T - h}^{t_{\scriptscriptstyle k_{z}}+T} \Big[F_i(\bar{e}_1(s), \bar{u}_1(s)) \Big] ds  \le \rho_i(h) L_{\scriptscriptstyle V_i}. \notag 
	\end{align}
	The last inequality along with \eqref{eq:lyap_2} leads to:
	\begin{align}
	&\bar{J}(e(t_{\scriptscriptstyle k_{z+1}})) - J^\star(e(t_{\scriptscriptstyle k_z})) \le (T -2 h) \rho_i(h) L_{\scriptscriptstyle F_i} +\rho_i(h) L_{\scriptscriptstyle V_i} \notag \\ 
	&\hspace{25mm}-\int_{t_{\scriptscriptstyle k_z}}^{t_{\scriptscriptstyle k_{z+1}}} \Big[ F_i(\hat{e}_2(s), \hat{u}_2(s)) \Big] ds. \label{eq:lyap_3}
	\end{align}
	By substituting $e_i = \hat{e}_2(s), u_i = \hat{u}_2(s)$ in \eqref{eq:F_lower_bound} we get $F_i(\hat{e}_2(s), \hat{u}_2(s)) \ge m_i \|\hat{e}_2(s)\|^2$, or equivalently:
	\begin{align*}
	&\int_{t_{\scriptscriptstyle k_z}}^{t_{\scriptscriptstyle k_{z+1}}} \Big[ F_i(\hat{e}_2(s), \hat{u}_2(s)) \Big] ds \ge \underline{m}_i \int_{t_{\scriptscriptstyle k_z}}^{t_{\scriptscriptstyle k_{z+1}}} \|\hat{e}_2(s)\|^2  ds \notag \\
	\Leftrightarrow &-\int_{t_{\scriptscriptstyle k_z}}^{t_{\scriptscriptstyle k_{z+1}}} \Big[ F(\hat{e}_2(s), \hat{u}_2(s)) \Big] ds \le -\underline{m}_i \int_{t_{\scriptscriptstyle k_z}}^{t_{\scriptscriptstyle k_{z+1}}} \|\hat{e}_2(s)\|^2  ds.
	\end{align*}
	By combining the last result with \eqref{eq:lyap_3}, we get:
	\begin{align}
	&\bar{J}_i(e(t_{\scriptscriptstyle k_{z+1}})) - J^\star_i(e(t_{\scriptscriptstyle k_z})) \le (T -2 h) \rho_i(h) L_{\scriptscriptstyle F_i} +\rho_i(h) L_{\scriptscriptstyle V_i} \notag \\ 
	&\hspace{35mm}-\underline{m}_i \int_{t_{\scriptscriptstyle k_z}}^{t_{\scriptscriptstyle k_{z+1}}} \|\hat{e}_2(s)\|^2  ds \label{eq:lyap_4}
	\end{align}
	It is clear that the optimal solution at time $t_{\scriptscriptstyle k_{z+1}}$ i.e., $J^\star(e_i(t_{\scriptscriptstyle k_{z+1}}))$  will not be worse than the feasible one at the same time i.e. $\bar{J}(e_i(t_{\scriptscriptstyle k_{z+1}}))$. Therefore, \eqref{eq:lyap_4} implies:
	\begin{align*}
	&J^\star_i(e_i(t_{\scriptscriptstyle k_{z+1}})) - J^\star_i(e_i(t_{\scriptscriptstyle k_z})) \le (T -2 h) \rho_i(h) L_{\scriptscriptstyle F_i} +\rho_i(h) L_{\scriptscriptstyle V_i} \notag \\ 
	&\hspace{35mm}-\underline{m}_i \int_{t_{\scriptscriptstyle k_z}}^{t_{\scriptscriptstyle k_{z+1}}} \|\hat{e}_2(s)\|^2  ds,
	\end{align*}
	which is equivalent to:
	\begin{align*}
	&J^\star_i(e_i(t_{\scriptscriptstyle k_{z+1}})) - J^\star_i(e_i(t_{\scriptscriptstyle k_z})) \le \notag \\
	&\hspace{10mm}-m_i \int_{t_{\scriptscriptstyle k_z}}^{t_{\scriptscriptstyle k_{z+1}}} \|\hat{e}_i(s; \hat{u}_i^\star(s; e_i(t_{\scriptscriptstyle k_z})), e_i(t_{\scriptscriptstyle k_z}))\|^2  ds \notag \\
	&\hspace{20mm}+(T -2 h) \rho_i(h) L_{\scriptscriptstyle F_i} +\rho_i(h) L_{\scriptscriptstyle V_i}.
	\end{align*}
	which, according to \eqref{eq:lyapunov_iss}, is in the form:
	\begin{align}
	&J^\star_i(e_i(t_{\scriptscriptstyle k_{z+1}})) - J^\star_i(e_i(t_{\scriptscriptstyle k_z})) \le -\alpha(\|e_i\|)+\sigma(\|\bar{x}_i\|) \label{eq:lyap_7}
	\end{align}
	Thus, the optimal cost $J$ has been proven to be decreasing, and according to Definition \ref{def:ISS} and Theorem \ref{def:ISS_lyapunov}, the closed loop system is ISS stable. Therefore, the closed loop trajectories converges to the closed set $\mathcal{E}_{i}$.
\end{proof}

%%%%%%%%%%%%%%%%%%%%%%%%%%%%%%%%%%%%%%%%%%%%%%%%%%%%%%%%%%%%%%%%%%%%%%%%%%%%%%%%
\bibliography{references}

\begin{thebibliography}{10}

\bibitem{loizou_2004}
S.~Loizou and K.~Kyriakopoulos, ``{A}utomatic {S}ynthesis of {M}ulti-{A}gent
  {M}otion {T}asks {B}ased on {LTL} {S}pecifications,'' {\em 43rd IEEE
  Conference on Decision and Control (CDC)}, vol.~1, pp.~153--158, 2004.

\bibitem{muray_2010_receding}
T.~Wongpiromsarn, U.~Topcu, and R.~Murray, ``{R}eceding {H}orizon {C}ontrol for
  {T}emporal {L}ogic {S}pecifications,'' {\em 13th ACM International Conference
  on Hybrid Systems: Computation and Control (HSCC)}, pp.~101--110, 2010.

\bibitem{guo_2015_reconfiguration}
M.~Guo and D.~Dimarogonas, ``{M}ulti-{A}gent {P}lan {R}econfiguration {U}nder
  {L}ocal {LTL} {S}pecifications,'' {\em The International Journal of Robotics
  Research}, vol.~34, no.~2, pp.~218--235, 2015.

\bibitem{frazzoli_vehicle_routing}
S.~Karaman and E.~Frazzoli, ``{L}inear {T}emporal {L}ogic {V}ehicle {R}outing
  with {A}pplications to {M}ulti-{UAV} {M}ission {P}lanning,'' {\em
  International Journal of Robust and Nonlinear Control}, vol.~21, no.~12,
  pp.~1372--1395, 2011.

\bibitem{zavlanos_2016_multi-agent_LTL}
Y.~Kantaros and M.~Zavlanos, ``{A} {D}istributed {LTL}-{B}ased {A}pproach for
  {I}ntermittent {C}ommunication in {M}obile {R}obot {N}etworks,'' {\em
  American Control Conference (ACC)}, pp.~5557--5562, 2016.

\bibitem{fainekos_girard_2009_temporal}
G.~Fainekos, A.~Girard, K.~Hadas, and G.~Pappas, ``Temporal logic motion
  planning for {D}ynamic {R}obots,'' {\em Automatica}, vol.~45, no.~2,
  pp.~343--352, 2009.

\bibitem{belta_2010_product_system}
M.~Kloetzer and C.~Belta, ``{A}utomatic {D}eployment of {D}istributed {T}eams
  of {R}obots {F}rom {T}emporal {M}otion {S}pecifications,'' {\em IEEE
  Transactions on Robotics}, vol.~26, no.~1, pp.~48--61, 2010.

\bibitem{belta_cdc_reduced_communication}
M.~Kloetzer, X.~C. Ding, and C.~Belta, ``{M}ulti-{R}obot {D}eployment from
  {LTL} {S}pecifications with {R}educed {C}ommunication,'' {\em 50th IEEE
  Conference on Decision and Control (CDC)}, pp.~4867--4872, 2011.

\bibitem{liu_MTL}
J.~Liu and P.~Prabhakar, ``{S}witching {C}ontrol of {D}ynamical {S}ystems from
  {M}etric {T}emporal {L}ogic {S}pecifications,'' {\em IEEE International
  Conference on Robotics and Automation (ICRA)}, pp.~5333--5338, 2014.

\bibitem{murray_2015_stl}
V.~Raman, A.~Donz{\'e}, D.~Sadigh, R.~Murray, and S.~Seshia, ``{R}eactive
  {S}ynthesis from {S}ignal {T}emporal {L}ogic {S}pecifications,'' {\em 18th
  International Conference on Hybrid Systems: Computation and Control (HSCC)},
  pp.~239--248, 2015.

\bibitem{topcu_2015}
J.~Fu and U.~Topcu, ``{C}omputational {M}ethods for {S}tochastic {C}ontrol with
  {M}etric {I}nterval {T}emporal {L}ogic {S}pecifications,'' {\em 54th IEEE
  Conference on Decision and Control (CDC)}, pp.~7440--7447, 2015.

\bibitem{fainekos_mtl_2015_robot}
B.~Hoxha and G.~Fainekos, ``{P}lanning in {D}ynamic {E}nvironments {T}hrough
  {T}emporal {L}ogic {M}onitoring,'' 2016.

\bibitem{baras_MTL_2016_new}
Y.~Zhou, D.~Maity, and J.~S. Baras, ``{T}imed {A}utomata {A}pproach for
  {M}otion {P}lanning {U}sing {M}etric {I}nterval {T}emporal {L}ogic,'' {\em
  European Control Conference (ECC)}, 2016.

\bibitem{frazzoli_MTL}
S.~Karaman and E.~Frazzoli, ``{V}ehicle {R}outing {P}roblem with {M}etric
  {T}emporal {L}ogic {S}pecifications,'' {\em 47th IEEE Conference on Decision
  and Control (CDC 2008)}, pp.~3953--3958, 2008.

\bibitem{alex_2016_acc}
A.~Nikou, J.~Tumova, and D.~Dimarogonas, ``{C}ooperative {T}ask {P}lanning of
  {M}ulti-{A}gent {S}ystems {U}nder {T}imed {T}emporal {S}pecifications,'' {\em
  American Control Conference (ACC)}, pp.~13--19, 2016.

\bibitem{alur_2000_discrete_abstractions}
R.~Alur, T.~Henzinger, G.~Lafferriere, and G.~Pappas, ``{D}iscrete
  {A}bstractions of {H}ybrid {S}ystems,'' {\em Proceedings of the IEEE},
  vol.~88, no.~7, pp.~971--984, 2000.

\bibitem{belta_2004_abstraction}
C.~Belta and V.~Kumar, ``{A}bstraction and {C}ontrol for {G}roups of
  {R}obots,'' {\em IEEE Transactions on Robotics}, vol.~20, no.~5,
  pp.~865--875, 2004.

\bibitem{helwa2014block}
M.~Helwa and P.~Caines, ``{I}n-{B}lock {C}ontrollability of {A}ffine {S}ystems
  on {P}olytopes,'' {\em 53rd IEEE Conference on Decision and Control (CDC)},
  pp.~3936--3942, 2014.

\bibitem{habets_2006_reachability}
L.~Habets, P.~Collins, and V.~Schuppen, ``{R}eachability and {C}ontrol
  {S}ynthesis for {P}iecewise-{A}ffine {H}ybrid {S}ystems on {S}implices,''
  {\em IEEE Transactions on Automatic Control}, vol.~51, no.~6, pp.~938--948,
  2006.

\bibitem{zamani_2012_symbolic}
M.~Zamani, G.~Pola, M.~Mazo, and P.~Tabuada, ``{S}ymbolic {M}odels for
  {N}onlinear {C}ontrol {S}ystems without {S}tability {A}ssumptions,'' {\em
  IEEE Transactions on Automatic Control}, vol.~57, no.~7, 2012.

\bibitem{liu_2016_abstraction}
J.~Liu and N.~Ozay, ``{F}inite {A}bstractions {W}ith {R}obustness {M}argins for
  {T}emporal {L}ogic-{B}ased {C}ontrol {S}ynthesis,'' {\em Nonlinear Analysis:
  Hybrid Systems}, vol.~22, pp.~1--15, 2016.

\bibitem{abate_2014_finite_abstractions}
M.~Zamani, M.~Mazo, and A.~Abate, ``{F}inite {A}bstractions of {N}etworked
  {C}ontrol {S}ystems,'' {\em 53rd IEEE Conference on Decision and Control
  (CDC)}, pp.~95--100, 2014.

\bibitem{pola_2016_symbolic}
G.~Pola, P.~Pepe, and M.~D.~D. Benedetto, ``{S}ymbolic {M}odels for {N}etworks
  of {C}ontrol {S}ystems,'' {\em IEEE Transactions on Automatic Control}, 2016.

\bibitem{boskos_cdc_2015}
D.~Boskos and D.~Dimarogonas, ``{D}ecentralized {A}bstractions {F}or
  {M}ulti-{A}gent {S}ystems {U}nder {C}oupled {C}onstraints,'' {\em 54th IEEE
  Conference on Decision and Control (CDC)}, pp.~282--287, 2015.

\bibitem{alex_acc_8_pages}
A.~Nikou, D.~Boskos, J.~Tumova, and D.~V. Dimarogonas, ``{C}ooperative
  {P}lanning for {C}oupled {M}ulti-{A}gent {S}ystems under {T}imed {T}emporal
  {S}pecifications,'' {\em Americal Control Conference 2017}.

\bibitem{khalil_nonlinear_systems}
H.~Khalil, {\em {N}oninear {S}ystems}.
\newblock Prentice-Hall, New Jersey, 1996.

\bibitem{sontag_2008_ISS}
E.~D. Sontag, ``{I}nput to {S}tate {S}tability: {B}asic {C}oncepts and
  {R}esults,'' pp.~163--220, 2008.

\bibitem{sontag_1995_ISS_proofs}
E.~D. Sontag and Y.~Wang, ``{O}n {C}haracterizations of the {I}nput-to-{S}tate
  {S}tability {P}roperty,'' {\em Systems and Control Letters}, vol.~24, no.~5,
  pp.~351--359, 1995.

\bibitem{alur1994}
R.~Alur and D.~Dill, ``{A} {T}heory of {T}imed {A}utomata,'' {\em Theoretical
  Computer Science}, vol.~126, no.~2, pp.~183--235, 1994.

\bibitem{pavithra_expressiveness}
D.~D. Souza and P.~Prabhakar, ``{O}n the {E}xpressiveness of {MTL} in the
  {P}ointwise and {C}ontinuous {S}emantics,'' {\em International Journal on
  Software Tools for Technology Transfer}, vol.~9, no.~1, pp.~1--4, 2007.

\bibitem{quaknine_decidability}
J.~Ouaknine and J.~Worrell, ``{O}n the {D}ecidability of {M}etric {T}emporal
  {L}ogic,'' {\em 20th Annual IEEE Symposium on Logic in Computer Science
  (LICS)}, pp.~188--197, 2005.

\bibitem{alur_mitl}
R.~Alur, T.~Feder, and T.~A. Henzinger, ``{T}he {B}enefits of {R}elaxing
  {P}unctuality,'' {\em Journal of the ACM (JACM)}, vol.~43, no.~1,
  pp.~116--146, 1996.

\bibitem{reynold}
M.~Reynolds, ``{M}etric {T}emporal {L}ogics and {D}eterministic {T}imed
  {A}utomata,'' 2010.

\bibitem{bouyer_phd}
P.~Bouyer, ``{F}rom {Q}ualitative to {Q}uantitative {A}nalysis of {T}imed
  {S}ystems,'' {\em M{\'e}moire d’habilitation, Universit{\'e} Paris},
  vol.~7, pp.~135--175, 2009.

\bibitem{alex_acc_2017}
A.~Nikou, D.~Boskos, J.~Tumova, and D.~V. Dimarogonas, ``{C}ooperative
  {P}lanning for {C}oupled {M}ulti-{A}gent {S}ystems under {T}imed {T}emporal
  {S}pecifications,'' {\em http://arxiv.org/pdf/1603.05097v2.pdf}.

\bibitem{tripakis_tba}
S.~Tripakis, ``{C}hecking {T}imed {B}uchi {A}utomata {E}mptiness on
  {S}imulation {G}raphs,'' {\em ACM Transactions on Computational Logic
  (TOCL)}, vol.~10, no.~3, 2009.

\bibitem{maler_MITL_TA}
O.~Maler, D.~Nickovic, and A.~Pnueli, ``{F}rom {MITL} to {T}imed {A}utomata,''
  {\em International Conference on Formal Modeling and Analysis of Timed
  Systems}, pp.~274--289, 2006.

\bibitem{nickovic_timed_aut}
D.~Ni{\v{c}}kovi{\'c} and N.~Piterman, ``{F}rom {MTL} to {D}eterministic
  {T}imed {A}utomata,'' {\em Formal Modeling and Analysis of Timed Systems},
  2010.

\bibitem{mesbahi_2010_graph_theory}
M.~Mesbahi and M.~Egerstedt, {\em {G}raph {T}heoretic {M}ethods in {M}ultiagent
  {N}etworks}.
\newblock Princeton University Press, 2010.

\bibitem{andreasson_2014_smart_building}
M.~Andreasson, D.~V. Dimarogonas, H.~Sandberg, and K.~H. Johansson,
  ``{D}istributed {C}ontrol of {N}etworked {D}ynamical {S}ystems: {S}tatic
  {F}eedback, {I}ntegral {A}ction and {C}onsensus,'' {\em IEEE Transactions on
  Automatic Control}, vol.~59, no.~7, pp.~1750--1764, 2014.

\bibitem{morrari_npmpc}
K.~S. de~Oliveira and M.~Morari, ``{C}ontractive {M}odel {P}redictive {C}ontrol
  for {C}onstrained {N}onlinear {S}ystems,'' {\em IEEE Transactions on
  Automatic Control}, vol.~45, no.~6, pp.~1053--1071, 2000.

\bibitem{cannon_2001_nmpc}
B.~Kouvaritakis and M.~Cannon, {\em {N}onlinear {P}redictive {C}ontrol:
  {T}heory and {P}ractice}.
\newblock No.~61, Iet, 2001.

\bibitem{frank_2003_nmpc_bible}
R.~Findeisen, L.~Imsland, F.~Allgower, and B.~A. Foss, ``{S}tate and {O}utput
  {F}eedback {N}onlinear {M}odel {P}redictive {C}ontrol: {A}n {O}verview,''
  {\em European Journal of Control}, vol.~9, no.~2-3, pp.~190--206, 2003.

\bibitem{frank_1998_quasi_infinite}
H.~Chen and F.~Allg{\"o}wer, ``{A} {Q}uasi-{I}nfinite {H}orizon {N}onlinear
  {M}odel {P}redictive {C}ontrol {S}cheme with {G}uaranteed {S}tability,'' {\em
  Automatica}, vol.~34, no.~10, pp.~1205--1217, 1998.

\bibitem{frank_2003_towards_sampled-data-nmpc}
R.~Findeisen, L.~Imsland, F.~Allg{\"o}wer, and B.~Foss, ``{T}owards a
  {S}ampled-{D}ata {T}heory for {N}onlinear {M}odel {P}redictive {C}ontrol,''
  {\em New Trends in Nonlinear Dynamics and Control and their Applications},
  pp.~295--311, 2003.

\bibitem{grune_2011_nonlinear_mpc}
L.~Gr{\"u}ne and J.~Pannek, {\em {N}onlinear {M}odel {P}redictive {C}ontrol}.
\newblock Springer London, 2011.

\bibitem{camacho_2007_nmpc}
E.~Camacho and C.~Bordons, ``{N}onlinear {M}odel {P}redictive {C}ontrol: {A}n
  {I}ntroductory {R}eview,'' pp.~1--16, 2007.

\bibitem{borrelli_2013_nmpc}
J.~Frasch, A.~Gray, M.~Zanon, H.~Ferreau, S.~Sager, F.~Borrelli, and M.~Diehl,
  ``{A}n {A}uto-{G}enerated {N}onlinear {MPC} {A}lgorithm for {R}eal-{T}ime
  {O}bstacle {C}voidance of {G}round {V}ehicles,'' {\em European Control
  Conference (ECC)}, 2013.

\bibitem{fontes_2001_nmpc_stability}
F.~Fontes, ``{A} {G}eneral {F}ramework to {D}esign {S}tabilizing {N}onlinear
  {M}odel {P}redictive {C}ontrollers,'' {\em Systems and Control Letters},
  vol.~42, no.~2, pp.~127--143, 2001.

\bibitem{camacho_2002_input_to_state}
D.~Marruedo, T.~Alamo, and E.~Camacho, ``{I}nput-to-{S}tate {S}table {MPC} for
  {C}onstrained {D}iscrete-{T}ime {N}onlinear {S}ystems with {B}ounded
  {A}dditive {U}ncertainties,'' {\em IEEE Conference on Decision and Control},
  vol.~4, pp.~4619--4624, 2002.

\bibitem{alina_ecc_2011}
A.~Eqtami, D.~Dimarogonas, and K.~Kyriakopoulos, ``{N}ovel {E}vent-{T}riggered
  {S}trategies for {M}odel {P}redictive {C}ontrollers,'' {\em IEEE Conference
  on Decision and Control}, pp.~3392--3397, 2011.

\bibitem{mayne_2000_nmpc}
D.~Mayne, J.~Rawlings, C.~Rao, and P.~Scokaert, ``{C}onstrained {M}odel
  {P}redictive {C}ontrol: {S}tability and {O}ptimality,'' {\em Automatica},
  vol.~36, no.~6, pp.~789--814, 2000.

\bibitem{katoen}
C.~Baier, J.~Katoen, and K.~G. Larsen, {\em Principles of Model Checking}.
\newblock MIT Press, 2008.

\end{thebibliography}
\bibliographystyle{ieeetr}

%%%%%%%%%%%%%%%%%%%%%%%%%%%%%%%%%%%%%%%%%%%% Balance Command %%%%%%%%%%%%%%%%%%%%%%%%%%%%%%%%%%%%%%%%%%%%%%
\addtolength{\textheight}{-12cm}

\end{document}